\documentclass[11pt]{article}
\usepackage{amsmath,amssymb}
\usepackage{amsthm}
\usepackage{bbm}
\usepackage{bm}
\usepackage{centernot}
\usepackage{color}
\usepackage{enumitem}
\usepackage{framed}
\usepackage{graphicx}
\usepackage{interval}
\usepackage{mathtools}
\usepackage{multibib}
\usepackage{multirow}
\usepackage[round,longnamesfirst]{natbib}
\usepackage{rotating}
\usepackage{setspace}
\usepackage{subcaption}
\usepackage{threeparttable}
\usepackage{tikz}
\usepackage{titlesec}
\usepackage{titling}
\usepackage{authblk}
\usepackage{url}

\mathtoolsset{showonlyrefs}
\newtheorem{thm}{Theorem}
\newtheorem{prop}{Proposition}
\newtheorem{cor}{Corollary}

\newtheorem{asm}{Assumption}
\newtheorem{con}{Condition}
\newtheorem{rem}{Remark}
\theoremstyle{definition}
\newtheorem{definition}{Definition}
\newtheorem{exm}{Example}
\newtheorem{algo}{Algorithm}

\renewenvironment{proof}[1][\proofname]{{\bfseries #1. }}{\qed}
\newtheorem{lemapp}{Lemma}

\definecolor{gold(metallic)}{rgb}{0.83, 0.69, 0.22}
\newcites{app}{References to Appendices}
\DeclareRobustCommand\full  {\tikz[baseline=-0.6ex]\draw[thick] (0,0)--(0.5,0);}
\DeclareRobustCommand\dotted{\tikz[baseline=-0.6ex]\draw[thick,dotted] (0,0)--(0.54,0);}
\DeclareRobustCommand\dashed{\tikz[baseline=-0.6ex]\draw[thick,dashed] (0,0)--(0.54,0);}
\DeclareRobustCommand\longdash{\tikz[baseline=-0.6ex]\draw[thick,dash pattern=on 8 off 4] (0,0)--(0.54,0);}

\title{Estimating Time-Varying Parameters of Various Smoothness in Linear Models via Kernel Regression}

\author[]{Mikihito Nishi\footnote{I am deeply grateful to Atsushi Inoue for his constructive comments on the earlier version of this article. I also thank the participants in workshops at Hitotsubashi University and the University of Tokyo. All errors are mine. A part of this research was conducted while the author was at Hitotsubashi University. The author gratefully acknowledges support by JSPS KAKENHI Grant number 25KJ0041. Address correspondence to: Graduate School of Economics, University of Tokyo, 7-3-1 Hongo, Bunkyou-ku, Tokyo, 113-0033, Japan; e-mail: mnishi@g.ecc.u-tokyo.ac.jp}}

\affil[]{Graduate School of Economics, University of Tokyo}

\date{\today}

\begin{document}
\onehalfspacing

    \begin{titlingpage}

        \maketitle

        \begin{abstract}
		We study kernel-based estimation of nonparametric time-varying parameters (TVPs) in linear models. Our contributions are threefold. First, we establish consistency and asymptotic normality of the kernel-based estimator for a broad class of TVPs including deterministic smooth functions, the rescaled random walk, structural breaks, the threshold model and their mixtures. Our analysis exploits the smoothness of the TVP. Second, we show that the bandwidth rate must be determined according to the smoothness of the TVP. For example, the conventional $T^{-1/5}$ rate is valid only for sufficiently smooth TVPs, and the bandwidth should be proportional to $T^{-1/2}$ for random-walk TVPs, where $T$ is the sample size. We show this highlighting the overlooked fact that the bandwidth determines a trade-off between the convergence rate and the size of the class of TVPs that can be estimated. Third, we propose a data-driven procedure for bandwidth selection that is adaptive to the latent smoothness of the TVP. Simulations and an application to the capital asset pricing model suggest that the proposed method offers a unified approach to estimating a wide class of TVP models.
        \end{abstract}

\medskip

\noindent


\medskip

\noindent


\end{titlingpage}

\section{Introduction}

Parameter instabilities are widely observed in econometric analysis. One of the most common specifications for parameter changes is the following linear model:
\begin{align}
	y_t = x_{t}'\beta_{T,t} + \varepsilon_t, \ t = 1,2,\ldots,T, \label{model:tvp_linear}
\end{align}
where $T$ is the sample size, $p\times 1$ vector $x_t$ is the regressor, $p\times 1$ triangular array $\beta_{T,t}$ is the time-varying coefficient, and $\varepsilon_t$ is the disturbance.

In the literature, time-varying parameters are often estimated via kernel regression where observations are weighted by some kernel function. Starting from \citet{robinsonNonparametricEstimationTimeVarying1989}, a large literature develops kernel-based estimation and inference for time-varying coefficient models; e.g., \citet{caiTrendingTimevaryingCoefficient2007}, \citet{chenTestingSmoothStructural2012}, \citet{zhangInferenceTimeVaryingRegression2012}, \citet{inoueRollingWindowSelection2017}, and \citet{friedrichSieveBootstrapInference2024}. We follow this strand of literature and study kernel-based estimation of $\beta_{T,t}$.

Our contributions are threefold. First, we consider a broader class of time-varying parameters than is typically assumed. The most common assumption adopted in existing works is that $\beta_{T,t}$ is so smooth that it is continuously differentiable \citep[e.g.,][]{caiTrendingTimevaryingCoefficient2007,zhangInferenceTimeVaryingRegression2012,inoueRollingWindowSelection2017}. However, smooth functions are not the only model for parameter instability popular in economics and statistics. The random walk model, in which $\beta_{T,t}$ is modeled as $\beta_{T,t} = \sum_{i=1}^{t}u_i$ with $u_i$ a transitory process, is a popular alternative \citep[e.g.,][]{nyblomTestingConstancyParameters1989,stockMedianUnbiasedEstimation1998,cogleyDriftsVolatilitiesMonetary2005}. Another example is (abrupt) structural breaks in $\beta_{T,t}$ \citep{andrewsTestsParameterInstability1993,baiEstimatingTestingLinear1998}. These two modeling schemes have received less attention in the literature on kernel-based estimation, and it is
largely unknown what the consequence is if one applies kernel regression to these models.\footnote{\citet{giraitisInferenceStochasticTimevarying2014} and \citet{giraitisTimevaryingInstrumentalVariable2021} are among few exceptions. They show that random-walk type parameters can be estimated via a kernel-based method.}$^{,}$\footnote{\citet{pesaranSelectionEstimationWindow2007}, \citet{pesaranOptimalForecastsPresence2013} and \citet{hiranoAnalyzingCrossvalidationForecasting2022} apply kernel-based approaches
for random walk and structural break type parameter instabilities, but their focus is on optimal forecasting of $y_t$, rather than estimation of $\beta_{T,t}$.} We develop kernel-based estimation theory that covers a wide class of time-varying parameters, including smooth functions, the rescaled random walk, and structural breaks. This class also includes the threshold regression model of \citet{hansenSampleSplittingThreshold2000}, which has rarely been considered in the context of kernel regression.

Let us emphasize that the class of time-varying parameters considered in this article also includes the mixtures of the aforementioned models. The relationship between $y_t$ and $x_t$, for example, may evolve smoothly but exhibit discontinuities during global financial crises or pandemics. The literature has acknowledged the importance of taking into account several types of parameter instability. For instance, \citet{mullerEfficientEstimationParameter2010} consider inference in models with time-varying parameters approximated by Gaussian processes and continuous functions possibly with finitely many jumps. Although their framework allows for nonlinear models and thus is more general than ours in this respect, they focus on small parameter instabilities (relative to those considered in this work). Therefore, large instabilities are not allowed in their model. \citet{chenTestingSmoothStructural2012} develop tests for smooth parameter changes with finitely many breaks but do not develop estimation theory for time-varying parameters of this type. \citet{kristensenNonparametricDetectionEstimation2012} proposes a nonparametric estimation method for time-varying coefficients by developing a framework that he argues allows for smooth functions, structural breaks, and the rescaled random walk. However, his analysis is restricted to smooth functional parameters only and cannot be extended to the other specifications. \citet{giraitisTimevaryingInstrumentalVariable2021} allow smooth deterministic functions, the rescaled random walk, and their mixture in an IV setting, but exclude (large) discontinuous breaks. Unlike these earlier works, we employ a general framework that accommodates all the aforementioned models and their mixtures.

We develop this framework by considering the class of time-varying parameters characterized by smoothness parameter $\alpha>0$, which generalizes the H\"{o}lder-class functions studied in the literature on nonparametric estimation \citep[e.g.,][]{tsybakovIntroductionNonparametricEstimation2009}. For example, continuously differentiable functions have smoothness of $\alpha=1$ as in the usual H\"{o}lder condition, and the random walk divided by $\sqrt{T}$ has $\alpha=1/2$. As with the H\"{o}lder condition, a smaller $\alpha$ means more roughness of the path of the time-varying parameter. This implies that the random walk divided by $\sqrt{T}$ is less smooth than continuously differentiable functions.

Our second contribution is to discuss the role of the bandwidth and its implications on bandwidth selection in the above general setting; beyond the usual bias-variance trade-off inherent in nonparametric estimation, we demonstrate that the bandwidth determines a trade-off between the rate of convergence and the size of the class of time-varying parameters that can be estimated. While this fact is consistent with and may be inferred from earlier results on nonparametric estimation of H\"{o}lder-class functions, its implications on bandwidth selection seem underappreciated in the literature on time-varying parameter models.

Specifically, we show that the rate of the bandwidth should be determined according to the smoothness of $\beta_{T,t}$. We illustrate this argument through two examples. First, we show that the conventional $T^{-1/5}$-rate bandwidth, which is specialized to continuously differentiable functions and often used in the literature \citep[e.g.,][]{zhouSimultaneousInferenceLinear2010}, is invalid if $\beta_{T,t}$ is less smooth. For example, we show that the bandwidth should be proportional to $T^{-1/2}$ if $\beta_{T,t}$ is the random walk divided by $\sqrt{T}$. Second, we demonstrate that, if the time-varying parameter experiences both smooth and abrupt parameter changes, the abrupt breaks of certain magnitudes are absorbed in smooth parameter changes so that the kernel-based estimation delivers valid inference, while discontinuous changes of a larger magnitude cause bias. We show that the bandwidth determines the break magnitudes at which this bias arises.

Our third contribution is to propose a data-driven bandwidth selection procedure. Unlike existing approaches that focus on smooth time-varying parameters and $T^{-1/5}$-rate bandwidths, the proposed method allows researchers to select the bandwidth from a wide range of candidate values, adaptively to the latent smoothness of the time-varying parameter. We evaluate its finite-sample performance via Monte Carlo simulations and illustrate the method using the capital asset pricing model (CAPM). In this application, our selection algorithm does not support the conventional $T^{-1/5}$-rate bandwidth, casting doubt on the extent to which this routinely selected bandwidth and the commonly used assumption of (continuously) differentiable parameters are justified. Furthermore, the proposed procedure partially supports a $T^{-1/2}$-rate bandwidth, producing an estimated trajectory close to that obtained from a Bayesian random-walk estimation.

The remainder of this paper is organized as follows. Section \ref{sec:smoothness_def} defines smoothness of time-varying parameters. Section \ref{sec:asym} establishes asymptotic properties of the kernel-based estimator. Section \ref{sec:bandwidth} discusses the consequence of an improper bandwidth choice and develops a bandwidth selection method adaptive to the smoothness of the time-varying parameter. Section \ref{sec:monte} conducts Monte Carlo experiments, and Section \ref{sec:empirics} gives a real data analysis. Section \ref{sec:conclusion} concludes. Mathematical proofs of the main results are relegated to Appendix A.

\textbf{Notation}: For any matrix $A$, $\|A\|=\mathrm{tr}(A'A)^{1/2}$ denotes the Frobenius norm of $A$. For any positive number $b$, $\lfloor b \rfloor$ denotes the integer part of $b$. $\stackrel{p}{\to}$ and $\stackrel{d}{\to}$ signify convergence in probability and convergence in distribution as $T\to\infty$, respectively. $\Rightarrow$ signifies weak convergence of the associated probability measures.

\section{Smoothness of Time-Varying Parameters}\label{sec:smoothness_def}

We consider estimating $\beta_{T,t}$ by using the local constant (Nadaraya-Watson) estimator:
\begin{align}
	\hat{\beta}_t \coloneqq \Biggl(\sum_{i=1}^{T}K\Bigl(\frac{t-i}{Th}\Bigr)x_ix_i'\Biggr)^{-1}\sum_{i=1}^{T}K\Bigl(\frac{t-i}{Th}\Bigr) x_iy_i,
\end{align}
where $K(\cdot)$ is a kernel function and $h$ is the bandwidth parameter satisfying $h\to0$ and $Th\to \infty$ as $T\to \infty$. Assumptions on the data generating process and kernel $K$ will be detailed in Section \ref{sec:asym}

In discussing the asymptotic properties of $\hat{\beta}_t$, the smoothness of the path of $\beta_{T,t}$ has a decisive effect. In the following definition, we quantify the smoothness of $\beta_{T,t}$ by a single parameter $\alpha$.

\begin{definition}
    Triangular array $\beta_{T,t}$ such that $\beta_{T,t} = O_p(1)$ as $T \to \infty$ for all $t$  is said to belong to the class \textit{type-a} $\mathrm{TVP}(\alpha)$ or \textit{type-b} $\mathrm{TVP}(\alpha)$, if the following condition (a) or (b) holds, respectively:
\begin{itemize}
		\item[(a)] There exists some real $\alpha>0$ such that for any sequence $\{a_T\}$ of positive integers satisfying $a_T = o(T)$ and $a_T \to \infty$ as $T \to \infty$, and for any $t$,
		\begin{align}
			\max_{j:|t-j|\leq a_T} \|\beta_{T,t}-\beta_{T,j}\| = O_p\Bigl(\Bigl(\frac{a_T}{T}\Bigr)^{\alpha}\Bigr), \ \text{as} \ T \to \infty.
		\end{align}
		\item[(b)] There exists some real $\alpha>0$ such that for any sequence $\{a_T\}$ of positive integers satisfying $a_T = o(T)$ and $a_T \to \infty$ as $T \to \infty$, and for any $t$,
		\begin{align}
			\max_{j:|t-j|\leq a_T} \|\beta_{T,t}-\beta_{T,j}\| = O_p\Bigl(\frac{1}{T^{\alpha}}\Bigr), \ \text{as} \ T \to \infty.
		\end{align}
	\end{itemize}	

    Furthermore, if $\beta_{T,t}$ belongs to $\mathrm{TVP}(\alpha)$ (type-a or type-b) but does not belong to $\mathrm{TVP}(\beta)$ for all $\beta>\alpha$, then it is said to belong to $\mathrm{TVP}(\alpha)$ \textit{on the boundary}.
\label{def:tvp_alpha_class}
\end{definition}

Definition \ref{def:tvp_alpha_class} essentially controls by $\alpha$ the smoothness of the path of $\beta_{T,t}$ on any interval of any length of a smaller order than $T$. In typical applications, $a_T$ will be set $a_T = \lfloor Th \rfloor$. Definition \ref{def:tvp_alpha_class}(a) allows the difference between the values of $\beta_{T,t}$ at distinct time points to grow as the time points gets further apart, while Definition \ref{def:tvp_alpha_class}(b) does not.\footnote{Therefore, $\beta_{T,t}$ belongs to type-a $\mathrm{TVP}(\alpha)$ if it belongs to type-b $\mathrm{TVP}(\alpha)$.}

Determining $\alpha$ such that a given $\beta_{T,t}$ belongs to $\mathrm{TVP}(\alpha)$ on the boundary enables us to derive the largest possible bandwidth under which $\hat{\beta}_t - \beta_{T,t}$ is asymptotically normally distributed; see Theorem \ref{thm:bias_bandwidth} below for this point.

Because $a_T/T<1$ and $\alpha>0$, a smaller $\alpha$ permits larger differences $\|\beta_{T,t}-\beta_{T,j}\|$, resulting in $\beta_{T,t}$ possibly having a rougher path. Note that triangular arrays unbounded in probability are excluded from Definition \ref{def:tvp_alpha_class}. We emphasize that Definition \ref{def:tvp_alpha_class} does not impose any parametric assumption on $\beta_{T,t}$ (other than smoothness $\alpha$), and that $\beta_{T,t}$ may be deterministic or stochastic. In addition, $\beta_{T,t}$ is allowed to have arbitrary correlation with $x_t$ and $\varepsilon_{t}$. Definition \ref{def:tvp_alpha_class} is quite general and accommodates many important time-varying parameters, as shown below.

\begin{rem}
    \citet{giraitisTimevaryingInstrumentalVariable2021} develop a kernel-based instrumental variable method to estimate time-varying parameters. The classes of time-varying parameters they consider are essentially type-a $\mathrm{TVP}(1)$ and $\mathrm{TVP}(1/2)$, albeit with slightly different definitions. They do not consider time-varying parameters belonging to type-a $\mathrm{TVP}(\alpha)$ with $\alpha \neq 1/2,1$ or type-b $\mathrm{TVP}(\alpha)$.
\end{rem}

\begin{exm}[Continuously differentiable functions]
	A popular model for time-varying parameters is deterministic smooth functions, accompanied by the formulation $\beta_{T,t} = \beta(t/T)$ for some continuously differentiable function $\beta(\cdot)$ on $[0,1]$ \citep[e.g.,][]{caiTrendingTimevaryingCoefficient2007,zhangInferenceTimeVaryingRegression2012,chenTestingSmoothStructural2012}. Under this formulation, the fact that $\sup_{0\leq r\leq 1}\|\beta'(r)\| \leq C$ for some constant $C>0$ implies that for any $s,t=1,2,\ldots,T$, $\|\beta_{T,t}-\beta_{T,s}\| = \|\beta(t/T) - \beta(s/T)\| \leq C |t-s|/T$ by the mean value theorem. Therefore, we have $\max_{j:|t-j|\leq a_T}\|\beta_{T,t} - \beta_{T,j}\| \leq Ca_T/T = O(a_T/T)$ uniformly in $t$, which implies $\beta_{T,t}$ belongs to the type-a $\mathrm{TVP}(1)$ class. Furthermore, if there exist some interval $(a,b)$ and constant $c>0$ such that $\inf_{x\in(a,b)}\|\beta'(x)\|\geq c$,\footnote{This condition excludes constant functions.} then we have $\max_{j:|t-j|\leq a_T}\|\beta_{T,t}-\beta_{T,j}\| \geq ca_T/T$ for $t\in (a,b)$ and sufficiently large $T$ by the mean value theorem, and hence $\max_{j:|t-j|\leq a_T}\|\beta_{T,t}-\beta_{T,j}\|$ is not $O((a_T/T)^{\alpha})$ for $\alpha>1$ for such $t$ and $T$. This implies $\beta_{T,t}$ belongs to type-a $\mathrm{TVP}(1)$ on the boundary. More generally, $\beta_{T,t}$ belongs to the type-a $\mathrm{TVP}(\alpha)$ class if it is H\"{o}lder continuous with exponent $\alpha$.
	\label{exm:cnt_diff}
\end{exm}


\begin{exm}[The random walk]
	Researchers often assume that the parameters of interest follow the random walk \citep[e.g.,][]{nyblomTestingConstancyParameters1989}. We consider the random walk scaled by $\sqrt{T}$: $\beta_{T,0} = \mu$ and $\beta_{T,t} = \mu + (1/\sqrt{T})\sum_{i=1}^{t}u_i, \ t\geq1$, where $\mu$ is a constant and $\{u_i\}$ is an i.i.d. sequence with $E[u_i]=0$ and $ V[u_i] = \Sigma_u > 0$. The functional central limit theorem (FCLT) implies that
	\begin{align}
		\beta_{T,\lfloor T\cdot\rfloor} = \mu + \frac{1}{\sqrt{T}}\sum_{i=1}^{\lfloor T\cdot \rfloor}u_i \Rightarrow \mu + \Sigma_u^{1/2}B_1(\cdot),
	\end{align}
    in the Skorokhod space $D^p_{[0,1]}$, where $B_1$ is a $p$-dimensional vector standard Brownian motion. Then, the following result holds: for any $a_T<t<T-a_T+1$, 
	\begin{align}
		\max_{j:|t-j|\leq a_T} \|\beta_{T,t} - \beta_{T,j}\| 
        &=\max\left\{ \max_{t-a_T\leq j\leq t-1} \|\beta_{T,t} - \beta_{T,j}\|, \max_{t+1\leq j \leq t+a_T} \|\beta_{T,j} - \beta_{T,t}\|\right\} \\
		&= \max\left\{\max_{t-a_T\leq j\leq t-1}\left\|\frac{1}{\sqrt{T}}\sum_{i=j+1}^{t} u_i\right\|, \max_{t+1\leq j \leq t+a_T} \left\|\frac{1}{\sqrt{T}}\sum_{i=t+1}^{j} u_i\right\|\right\} \\
		&\stackrel{d}{=}\max\left\{\max_{1\leq j\leq a_T}\left\|\frac{1}{\sqrt{T}}\sum_{i=1}^{j} u_i \right\|, \max_{1\leq j\leq a_T}\left\|\frac{1}{\sqrt{T}}\sum_{i=1}^{j} u_i^* \right\| \right\} \\
		&=\sqrt{\frac{a_T}{T}} \ \max\left\{\sup_{0\leq r \leq 1}\left\|\frac{1}{\sqrt{a_T}}\sum_{i=1}^{\lfloor a_Tr \rfloor} u_i \right\|, \sup_{0\leq r \leq 1}\left\|\frac{1}{\sqrt{a_T}}\sum_{i=1}^{\lfloor a_Tr \rfloor} u_i^* \right\| \right\},
	\end{align}
    where $u_i^*$ is an i.i.d. copy of $u_i$, and the third equality in distribution follows from the i.i.d property of $\{u_t\}$. Because $\sup_{0\leq r \leq 1}\|(a_T)^{-1/2}\sum_{i=1}^{\lfloor a_Tr \rfloor} u_i \| = O_p(1)$ by the continuous mapping theorem (CMT), $\max_{j:|t-j|\leq a_T} \|\beta_{T,t} - \beta_{T,j}\| = O_p(\sqrt{a_T/T})$. Moreover, since $\sup_{0\leq r \leq 1}\|(a_T)^{-1/2}\sum_{i=1}^{\lfloor a_Tr \rfloor} u_i \|\stackrel{d}{\to} \sup_{0\leq r \leq 1}\|B_1(r)\|$, where $\sup_{0\leq r \leq 1}\|B_1(r)\|>0$ a.s., it follows that $\max_{j:|t-j|\leq a_T} \|\beta_{T,t} - \beta_{T,j}\|$ is not $O_p((a_T/T)^{\alpha})$ for any $\alpha>1/2$.\footnote{This can be verified by using the strong approximation.} The same conclusion holds for the other $t$. Hence, the random walk divided by $\sqrt{T}$ belongs to the type-a $\mathrm{TVP}(1/2)$ class on the boundary. More generally, the random walk divided by $T^{\alpha}$ belongs to the type-a $\mathrm{TVP}(\alpha)$ class for $\alpha\geq 1/2$, while the random walk divided by $T^{\alpha}$ with $\alpha<1/2$ is excluded from Definition \ref{def:tvp_alpha_class} because it is unbounded in probability. 
	
    Because the random walk divided by $\sqrt{T}$ does not belong to $\mathrm{TVP}(1)$, it is less smooth than continuously differentiable functions on $[0,1]$. This is intuitively because the random walk divided by $\sqrt{T}$ weakly converges to Brownian motion, which is nowhere differentiable almost surely.
	\label{exm:rw}
\end{exm}

\begin{rem}
    \citet{mullerEfficientEstimationParameter2010} study an inferential problem concerning time-varying parameters approximated by Gaussian processes and piece-wise continuous functions scaled by a factor of $T^{-1/2}$. Leading examples are $T^{-1/2}\beta(t/T)$ with $\beta(\cdot)$ continuous on $[0,1]$ and $T^{-1/2}B_1(t/T)$, which is approximately equivalent (in distribution) to $T^{-1}\sum_{i=1}^{t}u_i=O_p(1/\sqrt{T})$. Therefore, non-vanishing smooth functions and random walks are not considered in their framework.
\end{rem}

\begin{exm}[Structural breaks]
    Structural breaks in parameters have attracted attention \citep[][provide a recent survey on this topic]{casiniStructuralBreaksTime2018}. Suppose time-varying coefficient $\beta_{T,t}$ experiences one abrupt break during the sample period:
	\begin{align}
		\beta_{T,t} = \begin{cases}
			\beta_1 & \text{for} \ t=1,2,\ldots,T_B \\
			\beta_2 & \text{for} \ t=T_B + 1, T_B + 2,\ldots,T \\
		\end{cases},
		\label{model:one_time_sb}
	\end{align}
	where $T_B = \lfloor \tau_B T\rfloor, \ \tau_B \in (0,1)$, and $\| \beta_1 - \beta_2\| = \delta/T^{\alpha}$ for some $\delta>0$ and $\alpha>0$. Under this formulation, the break is of shrinking magnitude, as considered in \citet{baiEstimationChangePoint1997}. $\beta_{T,t}$ belongs to the type-b $\mathrm{TVP}(\alpha)$ class on the boundary. Specifically, we have, for any $t \in\{1,\ldots,T_B-a_T\}\cup \{T_B+a_T+1,\ldots,T\}$, $\max_{j:|t-j|\leq a_T} \|\beta_{T,t}-\beta_{T,j}\| = 0$, and for any $t\in\{T_B-a_T+1,\ldots,T_B+a_T\}$, $\max_{j:|t-j|\leq a_T} \|\beta_{T,t}-\beta_{T,j}\| = \delta T^{-\alpha}$. Note that the asymptotically non-negligible discontinuity given by $\alpha=0$ is excluded from Definition \ref{def:tvp_alpha_class}.
	\label{exm:sb}
\end{exm}

\begin{exm}[Threshold models]
    \citet{hansenSampleSplittingThreshold2000} considers the threshold regression model obtained by letting $\beta_{T,t} = \theta_1 + \delta_T1\{q_t > \eta\}$, where $q_t$ is the threshold variable that determines the regime at time $t$, depending on whether it exceeds threshold parameter $\eta$. $\delta_T$, which \citet{hansenSampleSplittingThreshold2000} refers to as the threshold effect, expresses the magnitude of discontinuous changes in $\beta_{T,t}$. \citet{hansenSampleSplittingThreshold2000} assumes $\delta_T = c/T^{\alpha}$.\footnote{\citet{hansenSampleSplittingThreshold2000} also imposes $0<\alpha<1/2$, but this restriction is not necessary in our framework.} $\beta_{T,t}$ clearly satisfies $\|\beta_{T,t} - \beta_{T,j}\| \leq \|\delta_T\| = O_p(1/T^{\alpha})$, for all $t$ and $j$, which implies $\beta_{T,t}$ belongs to type-b $\mathrm{TVP}(\alpha)$.
    \label{exm:threshold}
\end{exm}

\begin{exm}[Mixed model]
    Suppose that $\beta_{T,t}$ is expressed as $\beta_{T,t}=\beta_{1,T,t} + \beta_{2,T,t}$, where $\beta_{1,T,t}$ is continuously differentiable and $\beta_{2,T,t} = \mu + (1/\sqrt{T})\sum_{i=1}^tu_i$ with $u_i$ defined as in Example 2. Then, it is straightforward to show that $\beta_{T,t}$ belongs to type-a $\mathrm{TVP}(1/2)$. More generally, for any finite positive integer $S$, if $\beta_{T,t}$ is expressed as the sum of $S$ time-varying parameters each of which belongs to the type-a $\mathrm{TVP}(\alpha_s)$ class ($s=1,\ldots,S)$, then $\beta_{T,t}$ belongs to type-a $\mathrm{TVP}(\min\{\alpha_1,\ldots,\alpha_S\})$.
\end{exm}

Mixture models where $\beta_{T,t}$ is the sum of both type-a and type-b time-varying parameters will be considered in Section \ref{sec:bandwidth}.

\section{Asymptotics}\label{sec:asym}

\subsection{Assumptions}

We suppose kernel $K(\cdot)$ satisfies the following condition:
\begin{asm}
	\item[(a)] $K(x)\geq0, \ x\in \mathbb{R}$, is Lipschitz continuous and has compact support $[-1,1]$.
	\item[(b)] $\int_{-1}^{1}K(x)dx = 1$.
	\label{asm:kernel}
\end{asm}
Commonly used kernels such as the uniform density on $[-1,1]$ and the Epanechnikov kernel satisfy Assumption \ref{asm:kernel}. Following the arguments of \citet{giraitisInferenceStochasticTimevarying2014,giraitisTimevaryingInstrumentalVariable2021}, kernels with non-compact support such as the Gaussian kernel are permitted under some stronger condition. We focus on kernels with a compact support as specified in condition (a) to avoid unessential complications. Note that $2Th$ is the effective sample size of the kernel-based estimation, since $K((t-i)/Th) = 0$ for $i$ such that $|t-i|> Th$.

Next, we impose the following assumption on model \eqref{model:tvp_linear}.\footnote{For the definition of near epoch dependence (NED), see, e.g., \citet{davidsonStochasticLimitTheory1994}.}

\begin{asm}
    \item[(a)] $\{(x_t', \varepsilon_t)\}_t$ is $\mathrm{L}_2$-NED of size $-(r-1)/(r-2)$ on an $\alpha$-mixing sequence of size $-r/(r-2)$ for some $r>2$, with respect to some positive constants $d_t$ satisfying $\sup_td_t < \infty$. Moreover, $\sup_t E[\|x_t\|^{2r}] + \sup_t E[|\varepsilon_t|^{2r}] < \infty$.
	
    \item[(b)] $\{x_t\varepsilon_t\}_t$ has mean zero and is serially uncorrelated.

    \item[(c)] For each $t=\lfloor Tr \rfloor, \ r\in(0,1)$, and $h$ such that $h\to0$ and $Th\to\infty$ as $T\to \infty$, there exist nonrandom symmetric matrices $\Omega(r)>0$ and $\Sigma(r)>0$ such that $(1/Th)\sum_{i=1}^{T}K((t-i)/Th)E[x_ix_i'] \to \Omega(r)$ and $\mathrm{Var}\Bigl((1/\sqrt{Th})\sum_{i=1}^{T}K((t-i)/Th)x_i\varepsilon_i\Bigr) \to \Sigma(r)$.
    \label{asm:dgp}
\end{asm}
Assumption \ref{asm:dgp}(a) allows the regressor and disturbance to be weakly serially dependent. The NED assumption is more general than mixing conditions commonly assumed in the literature \citep{caiTrendingTimevaryingCoefficient2007,chenTestingSmoothStructural2012,giraitisTimevaryingInstrumentalVariable2021,friedrichSieveBootstrapInference2024}. Also note that we do not impose strict or covariance stationarity unlike earlier works \citep{caiTrendingTimevaryingCoefficient2007,chenTestingSmoothStructural2012,friedrichSieveBootstrapInference2024}, and thus our framework allows for heteroskedasticity in $\varepsilon_t$. Assumption \ref{asm:dgp}(b) requires that the product of regressors and disturbance be serially uncorrelated, which is satisfied when, for example, $\varepsilon_t$ is a martingale difference sequence (m.d.s.) with respect to $\mathcal{F}_{T,t} \coloneqq \sigma(\{x_{t+1},x_t,\varepsilon_t,x_{t-1},\varepsilon_{t-1},\ldots\})$. The assumption of no serial correlation or m.d.s. is common in the literature \citep{chenTestingSmoothStructural2012,kristensenNonparametricDetectionEstimation2012,giraitisTimevaryingInstrumentalVariable2021}. Assumption \ref{asm:dgp}(c) holds under Assumptions \ref{asm:dgp}(a)-(b) if $x_t$ and $x_t\varepsilon_t$ are covariance-stationary (see Corollary \ref{cor:bias_bandwidth}).

\subsection{Asymptotic properties of $\hat{\beta}_t$}

In the following theorem, we establish the consistency and asymptotic normality of the kernel-based estimator, $\hat{\beta}_t$.

\begin{thm}
	Suppose Assumptions \ref{asm:kernel} and \ref{asm:dgp} hold. Then, for $t=\lfloor Tr \rfloor, \ r\in (0,1)$, we have
	\begin{align}
		\sqrt{Th}(\hat{\beta}_t - \beta_{T,t} - R_{T,t}) \stackrel{d}{\to} N(0,\Omega(r)^{-1}\Sigma(r)\Omega(r)^{-1}),
    \label{dist:general_thm}
	\end{align}
	where
	\begin{align}
		R_{T,t} = \begin{cases}
			O_p(h^{\alpha}) & \text{if $\beta_{T,t}$ satisfies Definition \ref{def:tvp_alpha_class}(a)} \\
			O_p(T^{-\alpha}) & \text{if $\beta_{T,t}$ satisfies Definition \ref{def:tvp_alpha_class}(b)}
		\end{cases}.
    \label{eqn:bias_order_thm}
	\end{align}
	In particular, for $h =cT^{\gamma}, \ c>0, \ \gamma\in(-1,0)$, we have
        \mathtoolsset{showonlyrefs=false}
	\begin{align}
		\sqrt{cT^{1+\gamma}}(\hat{\beta}_t - \beta_{T,t}) \stackrel{d}{\to} N(0,\Omega(r)^{-1}\Sigma(r)\Omega(r)^{-1}),
    \label{dist:undersmooth_thm}
	\end{align}
        \mathtoolsset{showonlyrefs=true}
	for $\gamma \in \Gamma(\alpha)$, where
	\begin{align}
		\Gamma(\alpha) = \begin{cases}
			(-1,-\frac{1}{2\alpha+1}) & \text{if $\beta_{T,t}$ satisfies Definition \ref{def:tvp_alpha_class}(a)} \\
			(-1,2\alpha -1)\cap (-1,0) & \text{if $\beta_{T,t}$ satisfies Definition \ref{def:tvp_alpha_class}(b)}
		\end{cases}.
        \label{def:gamma_alpha}
	\end{align}

    If $\beta_{T,t}$ belongs to $\mathrm{TVP}(\alpha)$ on the boundary, $\gamma$ close to the right endpoint of $\Gamma(\alpha)$ gives asymptotic normality and the fastest possible convergence rate.
\label{thm:bias_bandwidth}
\end{thm}

\begin{rem}
	We do not derive the asymptotic distribution of $\hat{\beta}_t$ at boundary points (near $t=0$ and $t=T$), but the derivation will proceed along the lines of \citet{caiTrendingTimevaryingCoefficient2007}. As shown by \citet{caiTrendingTimevaryingCoefficient2007}, the local constant estimator suffers from a larger bias at boundary points than the local linear estimator if $\beta_{T,t}$ is continuously differentiable. However, as discussed soon later (in Example \ref{exm:cnt_diff} below), the local linear estimator is available only when $\beta_{T,t}$ is (continuously) differentiable and is not applicable to nondifferentiable time-varying parameters such as the random walk. To accommodate both differentiable and nondifferentiable time-varying parameters, we focus on the local constant estimator.
\end{rem}

\begin{cor}
    Suppose Assumptions \ref{asm:kernel} and \ref{asm:dgp}(a)-(b) hold. Suppose also that $\{x_t\}_t$ and $\{x_t\varepsilon_t\}_t$ are covariance-stationary. Then, \eqref{dist:general_thm}-\eqref{def:gamma_alpha} hold with $\Omega(r)$ and $\Sigma(r)$ replaced by $\Omega \coloneqq E[x_1x_1']$ and $\Sigma \coloneqq \int_{-1}^{1}K(x)^2dxE[\varepsilon_1^2x_1x_1']$, respectively.
    \label{cor:bias_bandwidth}
\end{cor}

In what follows, we will set $h=cT^{\gamma}$ and call $\gamma$ (as well as $h$) the bandwidth parameter.

If $\beta_{T,t}$ belongs to type-a $\mathrm{TVP}(\alpha)$ on the boundary, then it can be estimated by setting $\gamma \approx -1/(2\alpha+1)$, yielding a convergence rate $\approx T^{-\alpha/(2\alpha+1)}$. The same convergence rate has been established in the literature for the minimax risk of kernel-based estimators, assuming that the parameter of interest belongs to a H\"{o}lder class with exponent $\alpha$ \citep[e.g.,][]{tsybakovIntroductionNonparametricEstimation2009}. Theorem \ref{thm:bias_bandwidth} shows that an analogous result holds under our definition of smoothness, Definition \ref{def:tvp_alpha_class}(a).

Smoothness parameter $\alpha$ affects $\Gamma(\alpha)$, the set of bandwidth parameter $\gamma$ that yields $\sqrt{Th}$-consistency and asymptotic normality. This makes the rate of convergence $T^{(1+\gamma)/2}$ dependent on $\alpha$. Letting $\alpha \to 0$, the kernel-based estimation can accommodate time-varying parameters of arbitrary smoothness, but this is accompanied by $\Gamma(\alpha) \to -1$, resulting in the rate of convergence $T^{(1+\gamma)/2} \to 1$. In contrast, if we let $\alpha\to\infty$, then $\Gamma(\alpha)$ tends to $(-1,0)$, and the choice $\gamma \approx 0$ yields a nearly $\sqrt{T}$-rate convergence, but only highly smooth parameters can be estimated. This observation shows that there is a trade-off between the rate of convergence and the size of the class of the time-varying parameters that can be estimated.

Because $\Gamma(\alpha)$ is the set of $\gamma$ that yields $\sqrt{Th}$-consistency and asymptotic normality under given $\alpha$, we can obtain the set of $\alpha$ that leads to $\sqrt{Th}$-consistency and asymptotic normality of $\hat{\beta}_{t}$ under given $\gamma$, by inverting the expression of $\Gamma(\alpha)$. Letting $A(\gamma)$ denote such a set, we can say that $\hat{\beta}_t$ calculated using given $\gamma$ is $\sqrt{Th}$-consistent and asymptotically normal for time-varying parameters with smoothness $\alpha\in A(\gamma)$, where
\begin{align}
    A(\gamma) = \begin{cases}
			(-\frac{1+1/\gamma}{2},\infty) & \text{if $\beta_{T,t}$ satisfies Definition \ref{def:tvp_alpha_class}(a)} \\
			(\frac{1+\gamma}{2},\infty) & \text{if $\beta_{T,t}$ satisfies Definition \ref{def:tvp_alpha_class}(b)}
		\end{cases}.
  \label{def:A_gamma}
\end{align}
Letting $\gamma \to -1$, $A(\gamma)$ tends to $(0,\infty)$, which implies that time-varying parameters with any smoothness $\alpha>0$ can be estimated, but the rate of convergence becomes $T^{(1+\gamma)/2} \to 1$. On the other hand, if we let $\gamma\uparrow0$, then the rate of convergence is as fast as $\sqrt{T}$, but $A(\gamma) \to \infty$ (the smoothness of constant parameters) in the type-a case. Hence, the bandwidth determines the trade-off between efficiency and robustness.

\setcounter{exm}{0}
\begin{exm}[Continued]
    Because continuously differentiable $\beta_{T,t}$ belongs to the type-a $\mathrm{TVP}(1)$ class, for any $\gamma \in \Gamma(1) = (-1,-1/3)$, we have $\sqrt{cT^{1+\gamma}}(\hat{\beta}_t - \beta_{T,t}) \stackrel{d}{\to} N(0,\Omega(r)^{-1}\Sigma(r)\Omega(r)^{-1})$. Setting $\gamma \approx -1/3$ gives the fastest rate of convergence of $T^{1/3}$. If $\beta_{T,t}$ is twice continuously differentiable, and the kernel is symmetric, then the set of the admissible bandwidths, $\Gamma(\alpha)$, widens to $(-1,-1/5)$, giving the faster rate of convergence of $T^{2/5}$ \citep[see][]{caiTrendingTimevaryingCoefficient2007}. In general, we will be able to enlarge $\Gamma(\alpha)$ to $(-1,-1/(4\alpha+1))\cup (-1,-1/3)$ in the type-a case if the following additional condition (mimicking the Taylor expansion) holds:
	\begin{align}
		\max_{j:|t-j|\leq a_T} \left\|\beta_{T,t} - \beta_{T,j} -c_t\left(\frac{t}{T} - \frac{j}{T}\right) \right\| = O_p\Bigl(\Bigl(\frac{a_T}{T}\Bigr)^{2\alpha}\Bigr),
		\label{asm:differentiability}
	\end{align}
    for some (possibly random) bounded vector $c_t$. Condition \eqref{asm:differentiability}, however, essentially requires differentiability of $\beta_{T,t}$ with respect to time, which is not satisfied by, e.g., the random walk divided by $\sqrt{T}$, so that the enlarged version of $\Gamma(\alpha)$ is only available to a limited class of time-varying parameters. For the same reason, the local linear estimator, which is based on the Taylor expansion of $\beta_{T,t}$, is not applicable to nondifferentiable time-varying parameters.
\end{exm}

\begin{exm}[Continued]
	If $\beta_{T,t}$ is the random walk divided by $\sqrt{T}$, then $\Gamma(1/2)=(-1,-1/2)$, and thus the fastest rate of convergence given by $\gamma \approx -1/2$ is $T^{1/4}$, slower than $T^{2/5}$ in the continuously differentiable case. The same set of admissible bandwidths is derived by \citet{giraitisInferenceStochasticTimevarying2014}, who consider a random-walk type time-varying coefficient in the context of univariate AR(1) models.

    Furthermore, we show in Appendix B that the bandwidth minimizing the MSE of $\hat{\beta}_t$ is proportional to $T^{-1/2}$ when $\beta_{T,t}$ is the random walk divided by $\sqrt{T}$. We prove this result under more restrictive conditions than Definition \ref{def:tvp_alpha_class}, and Assumptions \ref{asm:kernel} and \ref{asm:dgp}. Therefore, the choice of $\gamma = -1/2$ may also be justified as the minimizer of the MSE of $\hat{\beta}_t$.
\end{exm}

\begin{exm}[Continued]
    Suppose $\beta_{T,t}$ is defined as in \eqref{model:one_time_sb}. Because $\beta_{T,t}$ belongs to the type-b $\mathrm{TVP}(\alpha)$ class, arbitrary $\gamma$ in $(-1,0)$ yields the $\sqrt{Th}$-consistency and asymptotic normality of $\hat{\beta}_t$ as long as $\alpha\geq 1/2$. In particular, setting $\gamma \approx 0$ gives a near $\sqrt{T}$-consistency.\footnote{In fact, setting exactly $\gamma=0$ yields $\sqrt{T}$-consistency and asymptotic normality if $\alpha>1/2$. In this case, each $\beta_{T,t}, \ t=1,\ldots,T$ is estimated by using the full sample, but $\hat{\beta}_t$ and $\hat{\beta}_s$ ($t\neq s$) may take different values. This is because the weighting scheme (based on the kernel, $K(\cdot)$) is different for different time points. If $K(\cdot)$ is the uniform kernel, $\hat{\beta}_t$ equals the full-sample OLS estimator for all $t$.}
 
    For the case of $\alpha \in (0,1/2)$, however, a smaller $\alpha$ leads to a larger discontinuity in $\beta_{T,t}$ and thus a slower rate of convergence (through a narrower $\Gamma(\alpha)$). Therefore, if $\beta_{T,t}$ experiences large structural breaks given by $\alpha<1/2$, and if there is no other source of instability in the path of $\beta_{T,t}$, then a conventional structural-break approach that achieves $\sqrt{T}$-consistency \citep[e.g., the sequential procedure proposed by][]{baiEstimatingTestingLinear1998} will be more suitable.
\end{exm}

\begin{exm}[Continued]
    The argument given in Example 3 also applies to the threshold model: When $\delta_T=O_p(1/T^{\alpha})$ with $\alpha\geq 1/2$, the kernel-based method delivers a $\sqrt{Th}$-consistent, asymptotically normal estimation of $\beta_{T,t}$, whereas \citeauthor{hansenSampleSplittingThreshold2000}'s (\citeyear{hansenSampleSplittingThreshold2000}) method should be used when $\alpha<1/2$ and the threshold effect solely determines the parameter path.
\end{exm}

\subsection{Estimation of variance-covariance matrices}

To conduct inference, one needs to consistently estimate the asymptotic variance of $\hat{\beta}_t$. Under Assumptions \ref{asm:kernel} and \ref{asm:dgp}, $\Omega(r)$ can be consistently estimated by
\begin{align}
    \hat{\Omega}(r) \coloneqq \frac{1}{Th}\sum_{i=1}^TK\left(\frac{\lfloor Tr \rfloor-i}{Th}\right)x_ix_i' \ ;
\end{align}
see Lemma \ref{lemapp:variance_plim} in Appendix A. A natural estimator of $\Sigma(r)$ is
\begin{align}
    \hat{\Sigma}(r) \coloneqq \frac{1}{Th}\sum_{i=1}^TK\left(\frac{\lfloor Tr \rfloor - i }{Th}\right)^2\hat{\varepsilon}_i^2x_ix_i',
\end{align}
where $\hat{\varepsilon}_i = y_i - x_i'\hat{\beta}_i$.\footnote{If $\{x_t\varepsilon_t\}$ is serially correlated, $\Sigma(r)$ is typically the long-run variance of $\{x_t\varepsilon_t\}$. In this case, an appropriate estimator of $\Sigma(r)$ would be a nonparametric kernel estimator such as the Newey-West one, as suggested by \citet{caiTrendingTimevaryingCoefficient2007}. We do not explore in this direction to save space.} To prove the consistency of $\hat{\Sigma}(r)$, however, the current assumptions are not sufficient. This is because for each $t=\lfloor Tr \rfloor$, the estimation errors $\hat{\beta}_{t+j}-\beta_{T,t+j}$, where $j\in[-\lfloor Th \rfloor, \lfloor Th \rfloor]$, are required to be asymptotically negligible \textit{uniformly} over $j\in[-\lfloor Th \rfloor, \lfloor Th \rfloor]$, on which $\hat{\Sigma}(r)$ is calculated. To ensure the uniform consistency of $\hat{\beta}_{t+j}$ over $j\in[-\lfloor Th \rfloor, \lfloor Th \rfloor]$, we need the following additional conditions.

\begin{asm}
    Assumption \ref{asm:dgp} holds with part (a) replaced by the following condition:
    \begin{itemize}
        \item[(a')] $\{(x_t', \varepsilon_t)\}_t$ is $\mathrm{L}_2$-NED of size $-2(r-1)/(r-2)$ on an $\alpha$-mixing sequence of size $-2r/(r-2)$ for some $r>2$, with respect to some positive constants $d_t$ satisfying $\sup_td_t < \infty$. Moreover, $\sup_t E[\|x_t\|^{2r}] + \sup_t E[|\varepsilon_t|^{2r}] < \infty$.
    \end{itemize}
    \label{asm:dgp_2}
\end{asm}
Assumption \ref{asm:dgp_2}(a') strengthens Assumption \ref{asm:dgp}(a) by increasing the decaying rates of the mixing and NED coefficients, essentially weakening the serial dependence of $\{(x_t', \varepsilon_t)\}_t$.

\begin{asm}
    There exists some constant $\rho>0$ such that $\inf_{t\geq1}\lambda'E[x_tx_t']\lambda\geq \rho\left\|\lambda\right\|^2$ for any $\lambda\neq0$.
    \label{asm:min_eigen}
\end{asm}
Assumption \ref{asm:min_eigen} requires that there be enough variation in the data, as it implies that the minimum eigenvalue of $E[x_tx_t']$ is bounded away from zero uniformly in $t$.

\begin{asm}
    For each $t=\lfloor Tr \rfloor$, $r\in(0,1)$, it holds that
    \begin{align}
        \max_{-\lfloor Th \rfloor\leq j \leq \lfloor Th \rfloor}\left\|\frac{1}{Th}\sum_{i=1}^TK\left(\frac{t+j-i}{Th}\right)x_ix_i'\left(\beta_{T,i}-\beta_{T,t+j}\right)\right\| = o_p(1).
    \end{align}
    \label{asm:uniform_negligible}
\end{asm}
Assumption \ref{asm:uniform_negligible} is a high-level one that ensures the uniform consistency of $\hat{\beta}_{t+j}$ over $j\in[-\lfloor Th \rfloor, \lfloor Th \rfloor]$. In Appendix C, we show that Assumption \ref{asm:uniform_negligible} is satisfied in the time-varying models and under the implied bandwidths discussed in Examples \ref{exm:cnt_diff}-\ref{exm:threshold}.

\begin{thm}
    Suppose Assumptions \ref{asm:kernel} and \ref{asm:dgp_2}-\ref{asm:uniform_negligible} hold. Then, for each $t=\lfloor Tr \rfloor, \ r\in (0,1)$, we have $\hat{\Sigma}(r) \stackrel{p}{\to} \Sigma(r)$.
    \label{thm:variance_consistency}
\end{thm}

\section{On Bandwidth Selection: Implications and a Guide}\label{sec:bandwidth}
In Theorem \ref{thm:bias_bandwidth}, we showed the set of admissible bandwidth rates depends on the smoothness $\alpha$ of $\beta_{T,t}$.
This implies that an improperly selected bandwidth rate (given by $\gamma \notin \Gamma(\alpha)$) leads to misleading inference. In this section, we illustrate this implication through some examples where the evolutionary mechanism of $\beta_{T,t}$ is misspecified. We also discuss how to choose the bandwidth in empirical studies.

\subsection{When random-walk $\beta_{T,t}$ is assumed to be continuously differentiable}\label{subsec:rw_misspecify}

Suppose one assumes $\beta_{T,t}$ is a continuously differentiable function and sets $\gamma \approx -1/3$, but the fact is that $\beta_{T,t}$ follows the random walk divided by $\sqrt{T}$. Using the results given in Theorem \ref{thm:bias_bandwidth}, it is readily shown that the kernel-based estimator satisfies $\sqrt{cT^{1+\gamma}}(\hat{\beta}_t - \beta_{T,t}) = S_{T,t} + O_p(T^{1/2+\gamma})$, where $S_{T,t} \stackrel{d}{\to}N(0,\Omega(r)^{-1}\Sigma(r)\Omega(r)^{-1})$. Since the bias term is of order $O_p(T^{1/2+\gamma})$ and $\gamma > -1/2$, the difference $\hat{\beta}_{t} - \beta_{T,t}$ is dominated by the bias term. Because the bias term is not normal in general, confidence intervals based on a normal approximation will perform poorly.

In the literature on smooth (differentiable) time-varying parameters, researchers often use a rule-of-thumb or plug-in bandwidth $h=\mathrm{constant}\times T^{-1/5}$, or pick the bandwidth minimizing the cross-validation criterion over $h\in [c_1T^{-1/5},c_2T^{-1/5}]$ for some $0<c_1<c_2$ \citep{zhouSimultaneousInferenceLinear2010,zhangInferenceTimeVaryingRegression2012,kristensenNonparametricDetectionEstimation2012,chengBayesianBandwidthEstimation2019,sunPenalizedTimevaryingModel2023}. Although these selection rules lead to an efficient estimation of $\beta_{T,t}$ as long as it is correctly specified as a continuously differentiable function, they will yield a biased estimation if $\beta_{T,t}$ is a random walk, or more generally, if $\beta_{T,t}$ does not belong to type-a $\mathrm{TVP}(1)$.

\subsection{The effect of neglected breaks}\label{subsec:neglect_break}

Suppose $\beta_{T,t}=\mu_{T,t} + (1/\sqrt{T})\sum_{i=1}^{t}u_i$, where $u_t$ is defined as in Example \ref{exm:rw}, and $\mu_{T,t}$ satisfies
\begin{align}
	\mu_{T,t} = \begin{cases}
		\mu_1 & \text{for} \ t=1,2,\ldots,T_B \\
		\mu_2 & \text{for} \ t=T_B + 1, T_B + 2,\ldots,T \\
	\end{cases},
 \label{eqn:break_in_mean}
\end{align}
with $T_B = \lfloor \tau_BT \rfloor$ and $\mu_2-\mu_1 = \delta/T^{\alpha}$. Then, we can show
\begin{align}
	R_{T,t}  
	=\begin{cases}
		O_p(T^{\gamma/2}) & \text{for} \ t\in[1,T_B-\lfloor Th \rfloor]\cup[T_B+1+\lfloor Th \rfloor,T] \\
		O_p(\max\{T^{\gamma/2},T^{-\alpha}\}) & \text{for} \ t\in[T_B-\lfloor Th \rfloor+1, T_B+\lfloor Th \rfloor] 
	\end{cases},
\end{align}
where $R_{T,t}$ is defined in \eqref{dist:general_thm} and \eqref{eqn:bias_order_thm}. The asymptotic order of the bias term, $R_{T,t}$, is $O_p(T^{\gamma/2})$ for $t$ outside the $\lfloor Th \rfloor$-neighborhood of break point $T_B$. On the $\lfloor Th \rfloor$-neighborhood of $T_B$, it is $O_p(T^{\gamma/2})$ if $\alpha\geq-\gamma/2$, while it is $O_p(T^{-\alpha})$ if $0<\alpha<-\gamma/2$. 

Suppose we estimate $\beta_{T,t}$ by $\hat{\beta}_t$ assuming $\mu_{T,t} = \mu$, that is, the parameter instability is purely due to the zero-mean random walk. In this case, the (misleading) optimal rate of convergence is achieved by the choice of $\gamma=-1/2$, yielding $R_{T,t}=O_p(T^{-1/4})$ for $t\in[1,T_B-\lfloor Th \rfloor]\cup[T_B+1+\lfloor Th \rfloor,T]$, and
\begin{align}
	R_{T,t}  
	=\begin{cases}
		O_p(T^{-1/4}) & \text{if} \ \alpha\geq1/4 \\
		O_p(T^{-\alpha}) & \text{if} \ 0<\alpha<1/4
	\end{cases}
\end{align}
for $t\in[T_B-\lfloor Th \rfloor+1, T_B+\lfloor Th \rfloor]$.

When $\alpha\geq1/4$, the asymptotic order of $R_{T,t}$ is $O_p(T^{-1/4})$ for all $t$, the same order as in the pure random walk case (see \eqref{eqn:bias_order_thm}), so that the choice $\gamma =-1/2$ is valid and leads to the fastest rate of convergence. In contrast, if $0<\alpha<1/4$, $R_{T,t} = O_p(T^{-\alpha})$ for $t=T_B \pm \lfloor rTh \rfloor, \ r\in [0,1]$. Because the asymptotically normal component of the decomposition of $\hat{\beta}_t - \beta_{T,t}$ is $O_p(T^{-(1+\gamma)/2}) = O_p(T^{-1/4})$, the asymptotic behavior of $\hat{\beta}_t - \beta_{T,t}$ is dominated by the bias term. Therefore, the structural break induces a severe bias in the kernel-based estimation on the $\lfloor Th \rfloor$-neighborhood of the discontinuity point $t=T_B$.

The above result tells us that, if we set $\gamma=-1/2$, abrupt breaks of size $1/T^{\alpha}$ are absorbed in random walk parameter instabilities if $\alpha\geq1/4$, while the abrupt breaks ``stick out" and cause bias when $\alpha<1/4$.

A more general result can be derived if we invoke $\Gamma(\alpha)$ and $A(\gamma)$ defined in \eqref{def:gamma_alpha} and \eqref{def:A_gamma}, respectively. Suppose that $\beta_{T,t}$ can be expressed as $\beta_{T,t}^1 + \mu_{T,t}$, where $\beta_{T,t}^1$ belongs to type-a $\mathrm{TVP}(\alpha_1)$ with $\alpha_1 \in (0,\infty)$, and $\mu_{T,t}$ is defined as in \eqref{eqn:break_in_mean} but the magnitude of the break is $\mu_2-\mu_1 = \delta/T^{\alpha_2}$. Then, $\hat{\beta}_t$ is $\sqrt{Th}$-consistent and asymptotically normal under any $\gamma \in \Gamma(\alpha_1)=(-1,-1/(2\alpha_1+1))$ for $t\in[1,T_B-\lfloor Th \rfloor]\cup[T_B+1+\lfloor Th \rfloor,T]$. On $[T_B-\lfloor Th \rfloor+1, T_B+\lfloor Th \rfloor]$, the abrupt break is absorbed in $\beta_{T,t}^1$, and the same $\gamma$ leads to $\sqrt{Th}$-consistency and asymptotic normality if $\alpha_2\in A(\gamma) = ((1+\gamma)/2, \infty)$, while the bias term dominates the asymptotically normal term if $\alpha_2<(1+\gamma)/2$.

\subsection{A guide for bandwidth selection}\label{subsec:data-driven}

In the previous subsections, we have observed that an improperly selected $\gamma$ leads to misleading inference. Therefore, care must be taken in determining the bandwidth parameter. Because bandwidth parameter $h$ takes the form of $h=cT^{\gamma}$ with $c>0$ and $\gamma<0$, we first discuss how to determine $\gamma$ and then how to select $c$.

If one can identify the evolutionary mechanism of $\beta_{T,t}$ based on some prior information, they may select $\gamma$ appropriately, referring to the theoretical results derived in Section \ref{sec:asym}. For instance, if the random walk coefficient model is plausible, $\gamma = -1/2$ is an appealing choice. If $\beta_{T,t}$ is known to be twice continuously differentiable, then various methods for bandwidth selection proposed in the literature can be used to determine $\gamma$ and $c$ jointly (e.g., \citet{zhangInferenceTimeVaryingRegression2012}). 

\begin{rem}
    Whatever $\gamma \in (-1,0)$ may be selected, abrupt breaks and threshold effects of size $1/T^{\alpha}$ lead to biased estimation around the discontinuity points if $\alpha<(1+\gamma)/2$; recall Section \ref{subsec:neglect_break}. To avoid facing bias around the discontinuity points, one may be tempted to split the sample using some test for structural breaks \citep[e.g., as proposed in][]{baiEstimatingTestingLinear1998} or \citeauthor{hansenSampleSplittingThreshold2000}'s (\citeyear{hansenSampleSplittingThreshold2000}) approach, and then apply kernel regression within each subsample. However, our simulation shows that these sample-splitting approaches may lead to a misleading conclusion if latent discontinuous changes are mixed with smooth parameter changes. According to the simulation results, structural break tests can both underestimate and overestimate the number of discontinuous changes with a nonnegligible (or large in some cases) probability. Underestimating the number of discontinuous breaks implies that some latent abrupt breaks are overlooked, and an overestimation implies that spurious abrupt breaks are detected. Therefore, conventional structural break tests probably are not suitable for detecting abrupt breaks if they are mixed with smooth parameter changes. See Appendix D for details.
\end{rem}

Determining $\gamma$ is more delicate when there is no prior information that helps identify the evolutionary mechanism of $\beta_{T,t}$. Here, we propose two data-driven procedures to select the value of $\gamma$, which require prespecified lower and upper bounds $\underline{\gamma}, \overline{\gamma}\in(-1,0)$ to construct the set of candidate $\gamma$ values, $[\underline{\gamma},\overline{\gamma}]$. The theory developed in this article will serve as a guiding principle in determining $\underline{\gamma}$ and $\overline{\gamma}$.

The first procedure we propose is a naive cross-validation-based method: For each $\gamma\in[\underline{\gamma},\overline{\gamma}]$ (on some grid), calculate $\hat{\beta}_{-t,m}(\gamma)$, where $\hat{\beta}_{-t,m}(\gamma)$ are the leave-$(2m+1)$-out local constant estimators calculated without the data on $s\in[t-m,t+m]$ and with $h=T^{\gamma}$ for some $m\in\mathbb{N}\cup\{0\}$, compute the cross-validation criterion, $\mathrm{CV}(\gamma)\coloneqq T^{-1}\sum_{t=1}^T(y_t-x_t'\hat{\beta}_{-t,m}(\gamma))^2$, and then pick the minimizer of $\mathrm{CV}(\gamma)$. One may use the generalized cross-validation (GCV) considered in  \citet{zhouSimultaneousInferenceLinear2010,zhangInferenceTimeVaryingRegression2012}.

The second procedure is based on fixed-design wild bootstrap \citep{goncalvesBootstrappingAutoregressionsConditional2004}.

\begin{algo}[Bootstrap-based]
    \hfill
    \begin{itemize}
        \item[1.] For each $\gamma_1\in[\underline{\gamma},\overline{\gamma}]$, calculate the local constant estimators with $h=h_1\coloneqq T^{\gamma_1}$, denoted by $\hat{\beta}_t(\gamma_1)$, and obtain residuals $\hat{\varepsilon}_t(\gamma_1)=y_t-x_t'\hat{\beta}_t(\gamma_1)$.

        \item[2.] For each $\gamma_1\in[\underline{\gamma},\overline{\gamma}]$, apply fixed-design wild bootstrap to resample $y_t$: $y_t^*(\gamma_1)=x_t'\hat{\beta}_t(\gamma_1)+\varepsilon_t^*(\gamma_1)$, where $\varepsilon_t^*(\gamma_1)\coloneqq \eta_t\hat{\varepsilon}_t(\gamma_1)$ and $\eta_t \sim \mathrm{i.i.d.} \ N(0,1)$ independent of the data. For each $\gamma_2\leq\gamma_1$, calculate the local constant estimators with $h=h_2=T^{\gamma_2}$ using $(y_t^*(\gamma_1), x_t)$:
        \begin{align}
            \hat{\beta}_t^*(\gamma_1,\gamma_2) \coloneqq \left(\sum_{i=t-\lfloor Th_2\rfloor}^{t+\lfloor Th_2\rfloor}K\left(\frac{t-i}{Th_2}\right)x_ix_i'\right)^{-1} \sum_{i=t-\lfloor Th_2\rfloor}^{t+\lfloor Th_2\rfloor}K\left(\frac{t-i}{Th_2}\right)x_iy_i^*(\gamma_1).
        \end{align}

        \item[3.] For each pair $(\gamma_1,\gamma_2)$, construct the $100(1-q)\%$ confidence intervals for $\hat{\beta}_t(\gamma_1)$ based on $\hat{\beta}_t^*(\gamma_1,\gamma_2)$, its standard error, and the quantile of $N(0,1)$, and compute the empirical coverage rates (obtained from $B$ bootstrap intervals), denoted by $\mathrm{CR}(\gamma_1,\gamma_2)$.

        \item[4.] The selected value is the largest $\gamma_1$ such that $\mathrm{CR}(\gamma_1,\gamma_2)\geq 1-\bar{q}$ for all $\gamma_2\leq\gamma_1$ and some tolerance level $\bar{q}$; that is, $\hat{\gamma} = \max\Upsilon$, where
        \begin{align}
            \Upsilon \coloneqq \{\gamma_1:\gamma_1\in[\underline{\gamma},\overline{\gamma}], \mathrm{CR}(\gamma_1,\gamma_2)\geq 1-\bar{q} \ \text{for all} \ \gamma_2\in[\underline{\gamma},\gamma_1]\}.
        \end{align}
        If $\Upsilon$ is empty, then $\hat{\gamma} = \underline{\gamma}$.
    \end{itemize}
    \label{algo:gamma_selection_cr}
\end{algo}

The rationale behind Algorithm \ref{algo:gamma_selection_cr} is as follows. If $\gamma_1$ is sufficiently small that $\hat{\beta}_t(\gamma_1)$ is $\sqrt{Th_1}$-consistent for $\beta_{T,t}$, then $\hat{\varepsilon}_t(\gamma_1)$ are good approximations of unobserved $\varepsilon_t$, and bootstrap sample $y_t^*(\gamma_1)$ generated from $\varepsilon_t^*(\gamma_1)$ is ``well-behaved". Treating $\hat{\beta}_t(\gamma_1)$ as the pseudo-true parameters, $\hat{\beta}_t^*(\gamma_1,\gamma_2)$ ($\gamma_2\leq\gamma_1$) are $\sqrt{Th_2}$-consistent for $\hat{\beta}_t(\gamma_1)$ and asymptotically normal under the bootstrap probability measure, in probability. Then, the confidence interval for $\hat{\beta}_t(\gamma_1)$ based on $\hat{\beta}_t^*(\gamma_1,\gamma_2)$ and $N(0,1)$ should attain empirical coverage rates close to the nominal confidence level, with high probability. In Step 4, we pick the largest $\gamma_1$ such that the above argument applies, so that the fastest possible convergence rate can be achieved.

To theoretically justify the above reasoning, we impose the following regularity condition, strengthening Assumption \ref{asm:uniform_negligible}:
\begin{asm}
     For each $t=\lfloor Tr \rfloor$, $r\in(0,1)$, it holds that
    \begin{align}
        \max_{-\lfloor Th_2 \rfloor\leq j \leq \lfloor Th_2 \rfloor}\left\|\frac{1}{Th_1}\sum_{i=1}^TK\left(\frac{t+j-i}{Th_1}\right)x_ix_i'\left(\beta_{T,i}-\beta_{T,t+j}\right)\right\| = o_p(1/\sqrt{Th_2}).
    \end{align}
    \label{asm:uniform_negligible_2}
\end{asm}
If $\beta_{T,t}$ satisfies Condition \ref{con:holder} given in Appendix C, then Assumption \ref{asm:uniform_negligible_2} holds when $2\alpha\gamma_1+\gamma_2<-1$. If $\beta_{T,t}$ is a rescaled random walk, under Condition \ref{con:rw_assumption5} given in Appendix C, Lemma 5(iii) of \citet{giraitisTimevaryingInstrumentalVariable2021} shows that Assumption \ref{asm:uniform_negligible_2} holds when $\gamma_1+\gamma_2<-1$.

\begin{thm}
    Suppose that Assumptions \ref{asm:kernel}, \ref{asm:dgp_2}, \ref{asm:min_eigen}, and \ref{asm:uniform_negligible_2} hold, and that $\beta_{T,t}$ belongs to type-a $\mathrm{TVP}(\alpha)$. If $\gamma_1\in\Gamma(\alpha)=(-1,-(2\alpha+1)^{-1})$, we have, for each $t=\lfloor Tr \rfloor, \ r\in (0,1)$, and for $\gamma_2\in[\underline{\gamma},\gamma_1]$,
    \begin{align}
        \sup_{x\in\mathbb{R}^p}\left|P^*\left(\sqrt{Th_2}\left(\hat{\beta}_t^*(\gamma_1,\gamma_2)-\hat{\beta}_t(\gamma_1)-R_{T,t}^*\right)\leq x\right) - P\left(Z\leq x\right)\right| \stackrel{p}{\to}0,
    \end{align}
    where $P^*$ denotes the probability measure induced by the fixed-design wild bootstrap, $Z\sim N(0, \Omega(r)^{-1}\Sigma(r)\Omega(r)^{-1})$, and
    \begin{align}
        R_{T,t}^* = \begin{cases}
            O_{p^*} \left(1/\sqrt{Th_2}\right) & \mathrm{if} \ \gamma_2=\gamma_1 \\
            o_{p^*} \left(1/\sqrt{Th_2}\right) & \mathrm{if} \ \gamma_2<\gamma_1
        \end{cases}, 
    \end{align}
    with arbitrarily high probability for sufficiently large $T$. 
    \label{thm:bootstrap}
\end{thm}

\begin{rem}
    The statement of Theorem \ref{thm:bootstrap} still holds if $\beta_{T,t}$ has type-b discontinuities of size $1/T^{\alpha_1}$ with $\alpha_1\in A(\gamma_1)=((1+\gamma_1)/2, \infty)$.
\end{rem}

\begin{rem}
Algorithm \ref{algo:gamma_selection_cr} may reject a valid choice $\gamma_1\in\Gamma(\alpha)$ and suggest a conservative $\hat{\gamma}$ (undersmoothing) if the (bootstrap) distribution of $\hat{\beta}_t^*(\gamma_1,\gamma_1)-\hat{\beta}_t(\gamma_1)$ is poorly approximated by the normal distribution. There are two cases where this normal approximation is poor. First, if $T$ and $\gamma_1$ are small, the effective sample size can be quite small.\footnote{For example, if $\gamma_1=-1/2$, the effective sample size is as small as $2\lfloor Th\rfloor=28$ when $T=200$.} Second, the bias term, $R_{T,t}^*$, in Theorem \ref{thm:bootstrap} may be of the same order as the asymptotically normal part of $\hat{\beta}_t^*(\gamma_1,\gamma_1)-\hat{\beta}_t(\gamma_1)$ and distort the distribution of $\hat{\beta}_t^*(\gamma_1,\gamma_1)-\hat{\beta}_t(\gamma_1)$.\footnote{A solution in the second case would be to correct the bias term, $R_{T,t}^*$, but this requires an explicit formula for $R_{T,t}^*$, which seems not possible under the quite general smoothness condition, Definition \ref{def:tvp_alpha_class}. For example, bias formulae are typically derived assuming $\beta_{T,t}$ is twice continuously differentiable \citep[e.g.,][]{caiTrendingTimevaryingCoefficient2007,zhouSimultaneousInferenceLinear2010}. Bias correction in our general framework is left for future research.} While a conservative $\hat{\gamma}$ yields an asymptotically unbiased estimation, it causes an efficiency loss. In spite of such a limitation, Algorithm \ref{algo:gamma_selection_cr} leads to a more efficient estimation than the most conservative choice $\gamma=\underline{\gamma}$, at least when $T$ is sufficiently large; see the simulation results in Section \ref{subsec:mc_dd}.
\end{rem}

Once $\gamma$ is determined, one can select $c$ by minimizing some criterion that is a function of $c$. For instance, the cross-validation criterion, $\mathrm{CV}(c) \coloneqq T^{-1}\sum_{t=1}^T(y_t - x_t'\hat{\beta}_{-t}(c,\hat{\gamma}))^2$, where $\hat{\beta}_{-t}(c,\hat{\gamma})$ are the leave-one-out kernel estimators calculated under $h=cT^{\hat{\gamma}}$, can be used. Some other criteria may be used to determine $c$ such as the AIC as suggested in \citet{caiTrendingTimevaryingCoefficient2007} or the GCV considered in  \citet{zhouSimultaneousInferenceLinear2010} and \citet{zhangInferenceTimeVaryingRegression2012}.

\section{Monte Carlo Simulation}\label{sec:monte}
In this section, we conduct three Monte Carlo experiments to verify the implications provided in Section \ref{sec:bandwidth}. We use the following DGP: $y_t = \beta_{T,t}x_t + \varepsilon_{t}, \ t=1,\ldots,T$, where $x_t = 0.5x_{t-1} + \varepsilon_{x,t}$ with $\varepsilon_{x,t} \sim \mathrm{i.i.d.} \ N(0,1)$, and $\beta_{T,t}$ is defined differently in different experiments. For the specification of $\varepsilon_{t}$, we consider two cases: $\varepsilon_{t} = u_t$, where $u_t \sim \mathrm{i.i.d.} \ N(0,1)$ (i.i.d. case) and $\varepsilon_{t} = \sigma_tu_t$ with $\sigma_t^2 = 0.1 + 0.3\varepsilon_{y,t-1}^2 + 0.6\sigma_{t-1}^2$ (GARCH case). To obtain $\hat{\beta}_t$, we use the Epanechnikov kernel $K(x) = 0.75(1-x^2)I(|x|\leq 1)$.

\subsection{Simulation for Section \ref{subsec:rw_misspecify}} \label{subsec:mc_misspecify}

The first experiment is related to Section \ref{subsec:rw_misspecify}, and $\beta_{T,t}$ is generated as the rescaled random walk: $\beta_{T,t} = T^{-1/2}\sum_{i=1}^{t}v_i$. We consider two DGPs for driver process $v_t$: (i)$v_t\sim \mathrm{i.i.d.} \ N(0,1)$ and (ii) $v_t \sim \mathrm{i.i.d.} \ \mathrm{log \ normal}$ with parameters $\mu=0$ and $ \sigma = 1$.\footnote{Specifically, $X$ follows a log normal distribution if $X=\exp(Z)$, where $Z \sim N(\mu,\sigma^2)$.} Four sample sizes are used: $T \in \{100,200,400,800\}$. To evaluate the global performance of $\hat{\beta}_{t}$, we calculate $\mathrm{MSE} = T^{-1}\sum_{t=1}^{T}(\hat{\beta}_t - \beta_{T,t})^2$ (the reported MSE is the mean MSE over 2000 replications). To evaluate the normal approximation given in Corollary \ref{cor:bias_bandwidth}, we construct the $95\%$ confidence interval for $\beta_{T,0.5T}$ (the middle point of the sample). The variance estimators are $\hat{\Omega} \coloneqq T^{-1}\sum_{i=1}^{T}x_i^2$ and $\hat{\Sigma} \coloneqq \int_{-1}^{1}K(x)^2dx\times T^{-1}\sum_{i=1}^{T}\hat{\varepsilon}_i^2x_i^2$. We experiment with bandwidth parameter $h=T^{\gamma}$ and $\gamma \in \{-0.2,-0.33,-0.5,-0.55,-0.6,-0.7\}$, and evaluate the performance for each pair $(\gamma,T)$. We also analyze the performance of the data-driven selection procedures for $\gamma$ suggested in Section \ref{subsec:data-driven}.\footnote{For the cross-validation-based method, we select $\hat{\gamma}$ from $[-0.5,-0.2]$ using leave-three-out estimators. For the bootstrap-based method, we select $\hat{\gamma}$ from $\{-0.5,-0.4,-0.33,-0.2\}$ and set the tolerance level $\bar{q}=0.1$.} The results are presented in Tables \ref{tab:rw_sf_misspecify_iid} and \ref{tab:rw_sf_misspecify_garch}.

Because $\beta_{T,t}$ is the random walk divided by $\sqrt{T}$, our theoretical results predict that an appropriate bandwidth is $\gamma \approx -1/2$, while the kernel-based estimator leads to poor inference when $\gamma > -1/2$. Our simulation result corroborates this analysis. First, consider the case where $\varepsilon_{t} \sim \mathrm{i.i.d.} \ N(0,1)$ (Table \ref{tab:rw_sf_misspecify_iid}). In case (i) (Gaussian random-walk $\beta_{T,t}$), when $\gamma = -0.2$, the coverage rate is far below the 95\% confidence level. What is worse, it deviates from 0.95 as $T$ increases. Note that the MSE is relatively large. When $\gamma=-1/3$, the MSE takes the smallest value for all $T$ considered, but the coverage rate is still too small. This result warns researchers against using these bandwidths unless they are confident that $\beta_{T,t}$ can be well approximated by smooth functions with smoothness parameter $\alpha=1$. For $\gamma\leq -1/2$, the interval estimation performs well with coverage rate being 85-90\% and getting better as $T$ increases. However, $\gamma = -0.7$ leads to undercoverage when $T$ is small and the largest MSE for all $T$. $\gamma = -0.6$ also gives large MSEs. The choices $\gamma \approx -1/2$ lead to good coverage and small MSE, so that these choices are recommended for random-walk type parameters, or more generally, for time-varying parameters belonging to $\mathrm{TVP}(1/2)$ on the boundary.

Next consider the performance of $\gamma=\hat{\gamma}$ selected by data-dependent procedures. For the cross-validation method, the mean MSE is close to the smallest MSE attained by the deterministic choice of $\gamma=-0.33$, particularly when $T$ is large. On the other hand, the coverage ratio is far below the nominal level and takes values between 73\% and 78\%, although the coverage gradually improves as the sample size increases. For the bootstrap-based selection, the mean MSE and the coverage ratio is almost identical to those attained by the deterministic choice of $\gamma=-0.5$; that is, the MSE is relatively large when $T$ is small, but improves quickly as $T$ increases, and the coverage ratio is reasonably good. Based on these observations, the cross-validation method seems useful when a small MSE is desired, while the bootstrap-based method is a reasonable choice when unbiased estimation is prioritized.

The result for case (ii) (non-Gaussian random-walk $\beta_{T,t}$) is similar, so the same comment applies.

Results for the case where $\varepsilon_{t}$ is GARCH (Table \ref{tab:rw_sf_misspecify_garch}) are similar to those for the i.i.d case. Hence, we do not repeat the same analysis.

\subsection{Simulation for Section \ref{subsec:neglect_break}}\label{subsec:mc_neglected_break}

The second experiment is for verifying the implication provided in Section \ref{subsec:neglect_break}. In this simulation, we analyze the effect of (neglected) structural breaks. For this purpose, we generate $\beta_{T,t}$ according to $\beta_{T,t} = \mu_{T,t} + T^{-1/2}\sum_{i=1}^{t}v_i$, where $v_i \sim \mathrm{i.i.d.} \ N(0,1)$ and $\mu_{T,t}$ is an intercept term experiencing a break at $t=0.5T$. Specifically, we let
\begin{align}
    \mu_{T,t} = \begin{cases}
        0 & \mathrm{for} \ t=1, \ldots,0.5T \\
        2/T^\alpha & \mathrm{for} \ t=0.5T+1,\ldots,T
    \end{cases},
\end{align}
where $\alpha \in \{0.1,0.2,0.3,0.4\}$. A smaller $\alpha$ yields a larger break.

We consider estimating $\beta_{T,t}$ with the choice $h = T^{-1/2}$, reflecting the ignorance of the break. According to our theoretical analysis, the kernel-based estimator has a severe bias around $t=0.5T$ when $\alpha<0.25$, while breaks given by $\alpha>0.25$ have no effect asymptotically. To confirm this implication, we calculate the MSE and coverage rate of $\hat{\beta}_{t}$ for $t=\tau T$ with $\tau = 0.4,0.45,0.5,0.55,0.6$. The MSE is calculated for each $\tau$ as the mean squared error over 2000 replications, that is, $\mathrm{MSE}(\tau) = 2000^{-1}\sum_{i=1}^{2000}(\hat{\beta}_{\tau T}^{(i)} - \beta_{T,\tau T}^{(i)})^2$, where superscript $i$ signifies $\hat{\beta}_{\tau T}^{(i)}$ and $\beta_{T,\tau T}^{(i)}$ are obtained in the $i$th replication. We consider four sample sizes; (i) $T = 100$, (ii) $T=200$, (iii) $T=400$, and (iv) $T=800$. We use $\hat{\Omega}_t = (Th)^{-1}\sum_{i=1}^{T}K((t-i)/Th)x_i^2$ and $\hat{\Sigma}_t = (Th)^{-1}\sum_{i=1}^{T}K((t-i)/Th)^2\hat{\varepsilon}_i^2x_i^2$ as the variance estimators to evaluate the normal approximation given in Theorem \ref{thm:bias_bandwidth}. Results are reported in Tables \ref{tab:sb_ignore_iid} and \ref{tab:sb_ignore_garch}.

First, let us see the case of $\varepsilon_{t}$ being i.i.d. and $T=100$ (Table \ref{tab:sb_ignore_iid}, the row labeled (i)). The MSEs and coverage rates for $\tau = 0.4$ and $0.6$ are stable across $\alpha$. This is because the break only affects estimation around the discontinuity point, $t=0.5T$. The break has a severe effect on $\hat{\beta}_{\tau T}$ with $\tau=0.45,0.5,0.55$, both in terms of MSE and coverage. The smaller $\alpha$ is (i.e. the larger the break is), the worse the performance gets. Moreover, this effect is more profound for $\tau$ closer to $0.5$. In terms of the coverage rate, smaller breaks given by $\alpha \geq 0.25$ have a nonnegligible effect. This indicates that, although breaks of these magnitudes asymptotically have no impact, they do have nontrivial effects in finite samples.

For case (ii) ($T=200$), MSEs for $\tau=0.5$ and $\alpha<0.25$ are still large. Note that MSEs for $\tau = 0.45,0.55$ are comparable with those for $\tau = 0.4,0.6$. This is because the abrupt break affects $\hat{\beta}_{t}$ on the $Th$-neighborhood of the break date. Because $t=0.45T$ and $t=0.55T$ are outside the $Th$-neighborhood of $0.5T$, the performance of $\hat{\beta}_{\tau T}$ improves as $T$ increases for $\tau = 0.45,0.55$. $\hat{\beta}_{0.5T}$ also suffers from poor coverage for all $\alpha$. For the cases with $T=400,800$ (cases (iii) and (iv)), a similar comment applies. In particular, the coverage rates for $\alpha < 0.25$ and $\tau=0.5$ deteriorate as $T$ increases.

Examining the case with $\varepsilon_{t}$ being GARCH (see Table \ref{tab:sb_ignore_garch}), the same conclusion is drawn, so the detail is omitted.

\subsection{Balance between robustness and efficiency}\label{subsec:mc_dd}

In this subsection, we investigate the finite-sample performance of the data-driven bandwidth selection procedures in an environment where $\beta_{T,t}$ evolves smoothly but experiences a jump at some point. Specifically, we specify $\beta_{T,t}$ as $\beta_{T,t} = \beta(t/T)$, where $\beta(x) = x + \mu_T(x)$ with $\mu_T(x) = 0$ for $x\leq 0.5$ and $\mu_T(x) = 1.5/T^{0.4}$ for $x>0.5$. $\beta_{T,t}$ evolves smoothly and deterministically over time but experiences a break at the middle point. Recalling the theoretical analysis in Section \ref{subsec:neglect_break}, the break of size $T^{-0.4}$ can be accommodated as long as $\gamma < -0.2$. We analyze the finite-sample performance of $\hat{\beta}_{T,t}$ with $\gamma \in \{-0.2,-0.33,-0.5\}$ and $\gamma$ selected by the data-driven methods. We are interested in (i) whether the data-driven procedures can select $\gamma < -0.2$ (unbiasedness) and (ii) whether the selected $\gamma$ is close to $-0.2$ (efficiency). We study the mean MSE and the empirical coverage ratio at the break point, $t=0.5T$. The results are reported in Table \ref{tab:smooth_break_dd}. Because the results for the cases with i.i.d. error and GARCH error are qualitatively similar, we comment on the i.i.d. case only.

We first consider the deterministic $\gamma$. Although $\gamma=-0.2$ yields the smallest MSE, this choice results in undercoverage at the break point, as expected. For both choices $\gamma=-0.33,-0.5$, the coverage rate improves as $T$ increases, but $\gamma=-0.33$ gives a much smaller MSE. Next consider the data-dependent procedures. For the cross-validation method, the mean MSEs are almost identical to those obtained under $\gamma=-0.33$, but the empirical coverage rate is well below the nominal rate. For the bootstrap method, the mean MSEs and coverage rates are almost identical to those obtained by $\gamma=-0.5$ for $T\in\{100,200,400\}$. When $T=800$, however, the bootstrap-based procedure improves the MSE by about 20\% compared to $\gamma=-0.5$ while maintaining the same level of coverage rate.

\section{Empirical Application}\label{sec:empirics}

In this section, we apply kernel regression to estimate the time-varying CAPM.\footnote{The R code used for the empirical application is available on the author's website (\url{https://sites.google.com/view/mikihito-nishi/home}).} Parameter instabilities are widely observed in the CAPM literature
\citep[see][and refereces therein]{ghyselsStableFactorStructures1998,lewellenConditionalCAPMDoes2006,famaValuePremiumCAPM2006,angCAPMLongRun2007,angTestingConditionalFactor2012,guoTimeVaryingBetaValue2017}. We consider estimating the following factor model:
\begin{align}
    R_{j,t} = \alpha_{j,t} + \beta_{j,t}R_{M,t} + \varepsilon_t,
\end{align}
where $R_{j,t}$ denotes the excess return of portfolio $j$ at time $t$, and $R_{M,t}$ is the market excess return. The coefficients alpha and beta are allowed to be time-varying.

\subsection{Background}

In the CAPM literature, parameter instability is often modeled by letting parameters depend on observable instruments. But results drawn from this approach tend to be sensitive to the choice of instruments \citep{ghyselsStableFactorStructures1998}. To overcome this problem, researchers have proposed time-varying parameter models that do not utilize exogenous information.

Some assume that parameters experience abrupt changes, and others model parameter instability via the (near) random walk or smooth functions of time. For example, \citet{famaValuePremiumCAPM2006} and \citet{lewellenConditionalCAPMDoes2006} split the sample assuming that parameter changes occur based on calendar time (e.g., monthly or yearly), and apply OLS within subsamples. However, estimates obtained in this fashion suffer from bias if the timing of structural breaks is misspecified. \citet{angCAPMLongRun2007} use a Bayesian approach assuming (near) random walk alpha and beta. \citet{liTestingConditionalFactor2011} and \citet{angTestingConditionalFactor2012} treat the parameters as deterministic continuously differentiable functions of time.

Given the fact that continuously differentiable functions and the random walk can be estimated under $\gamma=-1/5$ and $\gamma=-1/2$, respectively, we set the lower and upper bounds for $\gamma$ as $\underline{\gamma}=-0.5$ and $\overline{\gamma}=-0.2$.

\subsection{Data}
All data are extracted from Kenneth French's website (\url{https://mba.tuck.dartmouth.edu/pages/faculty/ken.french/data_library.html}). Following \citet{liCombinedApproachInference2015}, we form three portfolios denoted by G, V, and G-V, respectively, from the 25 size-B/M portfolios. G is the average of the five portfolios in the lowest B/M quintile, V is the average of the five portfolios in the highest B/M quintile, and V-G is simply their difference. All the data are monthly, spanning 1952:1-2019:12 ($T=816$).

\subsection{Results}

We use the Epanechnikov kernel and $\hat{\Omega}_t = (Th)^{-1}\sum_{i=1}^{T}K((t-i)/Th)x_ix_i'$ and $\hat{\Sigma}_t = (Th)^{-1}\sum_{i=1}^{T}K((t-i)/Th)^2\hat{\varepsilon}_i^2x_ix_i'$ as the variance estimators, where $x_t = (1,R_{M,t})'$. To save space, we only discuss the result for portfolio V-G. The results for portfolios G and V are given in Appendix E.

\subsubsection{Selection of the bandwidth}

We determine two tuning parameters for the bandwidth, $h=cT^{\gamma}$, as explained in Section \ref{subsec:data-driven}. For Algorithm \ref{algo:gamma_selection_cr}, we construct 95\% bootstrap confidence intervals and set the tolerance level to be $\bar{q}=0.1$, giving the threshold of 90\% empirical coverage rate. 

First, we consider selecting $\gamma$. Figure \ref{fig:CV_gamma_vmg} depicts the CV criterion computed using leave-($2m+1$)-out estimators for $m=0,1,2$. For $m=0,1$, the minimum is attained at $\gamma=-0.5$, whereas $\gamma=-0.32$ is the minimizer when $m=2$. Since there is little reason to prefer some specific value of $m$ to other values, we also use Algorithm \ref{algo:gamma_selection_cr} to seek further evidence. Reported in Table \ref{tab:gamma_selection_cr_vmg} are the mean empirical coverage rates taken over $t=1,\ldots,T$, $\overline{\mathrm{CR}}(\gamma_1,\gamma_2)\coloneqq T^{-1}\sum_{t=1}^T\mathrm{CR}_t(\gamma_1,\gamma_2)$. Each empirical coverage rate is calculated using $200$ bootstrap samples. For $\gamma_1=-0.33$ and $\gamma_1=-0.4$, the empirical coverage rates exceed the threshold of 0.9 for all $\gamma_2\leq\gamma_1$, and hence $\gamma_1=-0.33$ is supported by this procedure. Given these results, we set $\hat{\gamma}=-0.33$ since both CV- and bootstrap-based procedures support this choice. It is noteworthy that $\gamma=-1/5$, a prevalent choice in the literature, is rejected by our selection algorithm. This result highlights the importance of including other $\gamma$ values in the set of candidate bandwidths.

Given $\gamma=\hat{\gamma}$, we determine scaling constant $c$ via cross-validation. The selected value, $\hat{c}$, is the minimizer of the cross-validation criterion $\mathrm{CV}(c)$ over $c\in \{0.5,0.55,\ldots,1.5\}$.

\subsubsection{Interval estimation}

In Figure \ref{fig:vmg_int_kernel}, we plot the estimated time-varying alpha and its 95\% confidence band.\footnote{This confidence band is obtained by sequentially calculating the pointwise 95\% confidence intervals and is not a uniform 95\% confidence band.} The estimated alpha fluctuates around the value zero throughout the sample period, and the confidence band includes zero at all time points. Figure \ref{fig:vmg_slope_kernel} depicts the estimated time-varying beta. It starts with a positive value that is significantly different from zero and then fluctuates around zero up to $t=300$. Then, it starts to decrease and stays below zero with the confidence band excluding zero. It starts to increase from $t=600$, and fluctuates around zero from $t=660$ toward the end of the sample.

\subsubsection{Comparison with the Bayesian estimate}

The CV-based selection procedure suggests that $\gamma=-1/2$ is partly supported by the data. Noting that this choice accommodates parameters following the (rescaled) random walk, and that random walk parameters are often estimated via Bayesian methods, it is interesting to compare the kernel-based estimates obtained from $h=\hat{c}T^{-1/2}$ with the estimates obtained from a Bayesian procedure in which parameters are assumed to be the random walk.

Let $\theta_{t} \coloneqq (\alpha_{t}, \beta_{t})'$. In the Bayesian method, we estimate the time-varying alpha and beta by using the Markov Chain Monte Carlo algorithm, assuming that $\theta_{t} = \theta_{t-1} + u_t$, where $u_t \sim N(0,D^2)$ with $D^2 = \mathrm{diag}(D_1^2,D_2^2)$.\footnote{For computation, we use the R package \texttt{walker} developed by \citet{helskeWalkerBayesianGeneralized2023}.} As the prior distributions for parameters $\theta_{0}$, $D$ and $\mathrm{Var}(\varepsilon_t)=\sigma_{\varepsilon}^2$, we suppose $\theta_{0} \sim N(\mu\mathbf{1}_2,\sigma^2\mathbf{I}_2)$, $D_i \sim \mathrm{Gamma}(v_1,v_2), \ i=1,2$, and $\sigma_\varepsilon \sim \mathrm{Gamma}(\nu_1,\nu_2)$. We consider three configurations of hyperparameters. For each configuration, $(\mu,\sigma,v_2,\nu_1,\nu_2)$ are set to $(\mu,\sigma,v_2,\nu_1,\nu_2)= (0,32,10^{-4},2,10^{-4})$. The value of $v_1$ is varied, and we set $v_1 = 1, 2,$ and $4$.\footnote{We also changed the values for $(\mu,\sigma,v_2,\nu_1,\nu_2)$, but the estimates were insensitive to these parameters.}

In Figure \ref{fig:vmg_comp}, we compare the Bayesian estimates with the kernel-based estimates with $h=\hat{c}T^{-1/2}$. For the estimated alpha (Figure \ref{fig:vmg_int_comp}), the trajectory obtained from the kernel method is more volatile (with a larger amplitude) than that obtained from the Bayesian algorithm, but the trajectories seem to share the same frequency. More striking is the similarity between the estimates of the time-varying beta. The estimated trajectories obtained from the two distinct methods are almost indistinguishable throughout the sample period, irrespective of the value of $v_1$. 

We also compare the quantitative performances of the kernel and Bayesian estimators in terms of the in-sample fit (SSR). Standardizing the SSR obtained from the kernel-based method to be 1, the relative SSR's for the Bayesian estimators with $v_1=1,2,$ and $4$ are 1.019, 1.000, and 0.953, respectively. The kernel estimator yields a comparable in-sample fit relative to the Bayesian estimator. 

\section{Conclusion}\label{sec:conclusion}

We studied kernel-based estimation of time-varying parameters over a wide range of smoothness. We set up a general framework that quantifies the smoothness of the time-varying parameter by a single parameter $\alpha$, and established consistency and asymptotic normality of the kernel-based estimator under this framework. The results cover many important time-varying parameter models, including continuously differentiable functions, the rescaled random walk, abrupt structural breaks, the threshold regression model, and their mixtures.

Our analysis also highlights an often-overlooked role of the bandwidth and its implications on bandwidth selection. Beyond the bias-variance trade-off, when the parameter may be nondifferentiable, the bandwidth determines a trade-off between the rate of convergence and the size of the class of time-varying parameters that can be estimated. Theory and simulations show that the appropriate bandwidth rate depends on the smoothness of the time-varying parameter. In particular, a conventional $T^{-1/5}$-rate bandwidth is invalid in the case of nondifferentiable time-varying parameters such as the random walk. Another important implication from our result is that abrupt breaks of certain magnitudes cause bias in the kernel-based estimation.

Taking into account the diversity of existing time-varying parameter models, we proposed a data-dependent bandwidth selection procedure that adapts to unknown smoothness of the time-varying parameter. Monte Carlo experiments and an application to the time-varying CAPM suggest that the proposed method serves as a unified approach to estimating a variety of time-varying parameter models.

\bibliographystyle{standard}

\bibliography{TVP_Linear}

\clearpage
\begin{table} \caption{Mean MSE and coverage rate (CR) when $\varepsilon_{t}$ is i.i.d.}
	\centering
	\begin{threeparttable}
		\renewcommand{\arraystretch}{1.3}	\begin{tabular}{ccc@{\hskip 15pt}c@{\hskip 15pt}c@{\hskip 15pt}c@{\hskip 28pt}c@{\hskip 15pt}c@{\hskip 15pt}c@{\hskip 15pt}c}
			\hline & $\gamma$ & \multicolumn{3}{c}{MSE} & & \multicolumn{3}{r}{CR ($t=0.5T$)} & \\
			\hline  &  & $T$ & \multicolumn{2}{}{} & & $T$ & \multicolumn{2}{}{} & \\
             \cline{3-10} & & 100 & 200 & 400 & 800 & 100 & 200 & 400 & 800 \\
             \hline \multirow{6}{*}{(i)} & \multicolumn{1}{l}{-0.2} & 0.069 &	0.055 &	0.043 &	0.036 &	0.626 &	0.538 &	0.461 &	0.395  \\ 
             &\multicolumn{1}{l}{-0.33} &0.056 &	0.039 &	0.027 	&0.019 &	0.777 	&0.746 	&0.734 &	0.709 \\
             &\multicolumn{1}{l}{-0.5} &0.073 &	0.048 &	0.032 &	0.022 &	0.850 &	0.853 &	0.874 &	0.899 \\
             &\multicolumn{1}{l}{-0.55} &0.087 &	0.058 &	0.040 &	0.028 &	0.842 &	0.876 &	0.886 &	0.914  \\
             & \multicolumn{1}{l}{-0.6} &0.107 &	0.074 &	0.053 &	0.038 &	0.837 &	0.866 &	0.884 &	0.910 \\
             &\multicolumn{1}{l}{-0.7} &0.198 &	0.138 &	0.103 &	0.077 &	0.792 &	0.835& 	0.848 &	0.872 \\
             & \multicolumn{1}{l}{CV} &0.062 &0.042 & 0.028 & 0.019 & 0.746 & 0.733 & 0.759 & 0.780  \\
             &\multicolumn{1}{l}{Boot} &0.073 & 0.048 & 0.032 & 	0.022 &	0.850 &	0.853 & 0.874 &	0.881  \\
             \hline \multirow{6}{*}{(ii)} & \multicolumn{1}{l}{-0.2} &0.070 &	0.054 &	0.044 &	0.036 &	0.628 &	0.553 &	0.460 &	0.373  \\ 
             &\multicolumn{1}{l}{-0.33} &0.056 &	0.039 &	0.027 &	0.020 &	0.790 &	0.771 	&0.736& 	0.696  \\ 
             &\multicolumn{1}{l}{-0.5} &0.073 &	0.048 &	0.032 &	0.022 &	0.853 &	0.865 	&0.877 &	0.906 \\
             &\multicolumn{1}{l}{-0.55} &0.087 &	0.058 &	0.040 &	0.028 &	0.852 &	0.874 	&0.881 &	0.906 \\
             & \multicolumn{1}{l}{-0.6} &0.107 &	0.074 &	0.053 &	0.038 &	0.843 &	0.875 &	0.882 &	0.903 \\
             &\multicolumn{1}{l}{-0.7} &0.198 &	0.138 &	0.103 &	0.077 &	0.791 &	0.828 	&0.850 &	0.865  \\
             & \multicolumn{1}{l}{CV} &0.062& 0.042 & 0.028 & 0.019 & 0.750 & 0.748 & 0.769 & 0.780   \\
             &\multicolumn{1}{l}{Boot} &0.073 & 0.048 & 0.032 & 0.022 & 0.853 & 0.865 & 0.875 & 0.881   \\
             \hline 
		\end{tabular}
		\begin{tablenotes}
			\footnotesize
			\item[]Note: $\beta_{T,t} = T^{-1/2} \sum_{i=1}^{T}v_i$, where $v_i \sim \mathrm{i.i.d.} \ N(0,1)$ for case (i) and $v_i$ is log-normally distributed with $\mu=0, \ \sigma=1$ for case (ii). $\hat{\beta}_t$ is calculated using bandwidth parameter $h=T^{\gamma}$. The rows labeled ``CV" and ``Boot" signify the results for cross-validation-based and bootstrap-based selections, respectively 
		\end{tablenotes}
	\end{threeparttable} \label{tab:rw_sf_misspecify_iid}
\end{table}

\clearpage
\begin{table} \caption{Mean MSE and coverage rate (CR) when $\varepsilon_{t}$ is GARCH}
	\centering
	\begin{threeparttable}
		\renewcommand{\arraystretch}{1.3}	\begin{tabular}{ccc@{\hskip 15pt}c@{\hskip 15pt}c@{\hskip 15pt}c@{\hskip 28pt}c@{\hskip 15pt}c@{\hskip 15pt}c@{\hskip 15pt}c}
			\hline & $\gamma$ & \multicolumn{3}{c}{MSE} & & \multicolumn{3}{r}{CR ($t=0.5T$)} & \\
			\hline  &  & $T$ & \multicolumn{2}{}{} & & $T$ & \multicolumn{2}{}{} & \\
             \cline{3-10} & & 100 & 200 & 400 & 800 & 100 & 200 & 400 & 800 \\
             \hline \multirow{6}{*}{(i)} & \multicolumn{1}{l}{-0.2} &0.070 &	0.054 &	0.043 &	0.036 &	0.603 &	0.529 	&0.460 	&0.394  \\ 
             &\multicolumn{1}{l}{-0.33} &0.057 &	0.039 &	0.027& 	0.019 &	0.753 &	0.736 &	0.727 &	0.702  \\ 
             &\multicolumn{1}{l}{-0.5} &0.074 &	0.048 &	0.032 &	0.022 &	0.847 &	0.866 	&0.886 &	0.900   \\
             &\multicolumn{1}{l}{-0.55} &0.090 &	0.059 &	0.040 &	0.028 &	0.847 &	0.878 &	0.893 &	0.912  \\
             & \multicolumn{1}{l}{-0.6} &0.111 &	0.075 &	0.053 &	0.038 &	0.844 &	0.878 &	0.889 &	0.918  \\
             &\multicolumn{1}{l}{-0.7} &0.206 &	0.141 &	0.104 &	0.077 &	0.819 	&0.855 &	0.868 &	0.890  \\
             & \multicolumn{1}{l}{CV} &0.063 & 0.042 & 0.028 & 0.019 & 0.738 & 0.743 & 0.767 & 0.777    \\
             &\multicolumn{1}{l}{Boot} &0.074 & 0.048 & 0.032 & 0.021 & 0.847 & 0.866 & 0.884 & 0.883    \\
             \hline \multirow{6}{*}{(ii)} & \multicolumn{1}{l}{-0.2} &0.070 &	0.054 &	0.044 &	0.036 &	0.604 &	0.525 &	0.459 &	0.376  \\ 
             &\multicolumn{1}{l}{-0.33} &0.057 &	0.039 &	0.027 &	0.019 &	0.770 &	0.759 &	0.725 &	0.695  \\ 
             &\multicolumn{1}{l}{-0.5} &0.074 &	0.048 &	0.032 &	0.022 &	0.862 &	0.874 &	0.881 &	0.903  \\
             &\multicolumn{1}{l}{-0.55} &0.089 &	0.059 &	0.040 &	0.028 &	0.858 	&0.879& 	0.892 &	0.906  \\
             & \multicolumn{1}{l}{-0.6} &0.111 &	0.075 &	0.053 &	0.038 &	0.855 &	0.885 &	0.887 &	0.909  \\
             &\multicolumn{1}{l}{-0.7} &0.206 &	0.141 &	0.104 &	0.077 &	0.815 &	0.850 	&0.863 &	0.884  \\
             & \multicolumn{1}{l}{CV} &0.062 & 0.042 & 0.028 & 0.019 & 0.751 & 0.754 & 0.770 & 0.783    \\
             &\multicolumn{1}{l}{Boot} &0.074 & 0.048 & 0.032 & 0.021 & 0.862 & 0.874 & 0.878 & 0.888    \\
             \hline 
		\end{tabular}
		\begin{tablenotes}
			\footnotesize
			\item[]Note: $\beta_{T,t} = T^{-1/2} \sum_{i=1}^{T}v_i$, where $v_i \sim \mathrm{i.i.d.} \ N(0,1)$ for case (i) and $v_i$ is log-normally distributed with $\mu=0, \ \sigma=1$ for case (ii). $\hat{\beta}_t$ is calculated using bandwidth parameter $h=T^{\gamma}$. The rows labeled ``CV" and ``Boot" signify the results for cross-validation-based and bootstrap-based selections, respectively 
		\end{tablenotes}
	\end{threeparttable} \label{tab:rw_sf_misspecify_garch}
\end{table}

\clearpage
\begin{table} \caption{MSE and coverage rate when $\varepsilon_{t}$ is i.i.d}
	\centering
	\begin{threeparttable}
		\renewcommand{\arraystretch}{1.3}	\begin{tabular}{ccc@{\hskip 12pt}c@{\hskip 13pt}c@{\hskip 13pt}c@{\hskip 13pt}c@{\hskip 26pt}c@{\hskip 13pt}c@{\hskip 13pt}c@{\hskip 13pt}c@{\hskip 13pt}c}
			\hline & $\alpha$ & \multicolumn{4}{c}{MSE} & & \multicolumn{4}{c}{Coverage Rate} & \\
			\hline  &  & $\tau$ & \multicolumn{3}{}{} & & $\tau$ & \multicolumn{3}{}{} & \\
             \cline{3-12} & & 0.4 & 0.45 & 0.5 & 0.55 &0.6 & 0.4 & 0.45 & 0.5 & 0.55 & 0.6 \\
             \hline \multirow{4}{*}{(i)}& \multicolumn{1}{l}{0.1} &0.07& 	0.11& 	0.44 &	0.16 	&0.07 &	0.84 &	0.80 	&0.36 &	0.72 &	0.85   \\ 
             &\multicolumn{1}{l}{0.2} &0.07 &	0.08 &	0.22 &	0.10 &	0.07 &	0.84 &	0.82 &	0.54 &	0.77 &	0.84   \\
             &\multicolumn{1}{l}{0.3} &0.07 &	0.07 &	0.13& 	0.08 &	0.07 &	0.83 &	0.84 	&0.69 &	0.80 &	0.84 \\
             & \multicolumn{1}{l}{0.4} &0.07 &	0.07 &	0.09 	&0.07 &	0.07 &	0.83 &	0.84 &	0.77 	&0.81 &	0.83 \\
             \hline \multirow{4}{*}{(ii)}& \multicolumn{1}{l}{0.1} &0.04 &	0.05 &	0.39 &	0.06 &	0.05 	&0.86 &	0.86 &	0.30 &	0.83 &	0.85   \\ 
             &\multicolumn{1}{l}{0.2} &0.04 &	0.04 &	0.17 &	0.05 &	0.05 &	0.86 &	0.86 &	0.52 &	0.84 	&0.85 \\
             &\multicolumn{1}{l}{0.3} &0.04 &	0.04 &	0.09 &	0.05 &	0.05 &	0.86& 	0.87 &	0.70& 	0.85 &	0.85 \\
             & \multicolumn{1}{l}{0.4} &0.04 &	0.04 	&0.06 &	0.05 &	0.05 &	0.86 &	0.87 	&0.79 	&0.85 &	0.85 \\
             \hline \multirow{4}{*}{(iii)}& \multicolumn{1}{l}{0.1} &0.03 &	0.03 &	0.33 &	0.03 &	0.03 &	0.87 &	0.87 &	0.21 &	0.88 	&0.86 \\ 
             &\multicolumn{1}{l}{0.2} &0.03 &	0.03 &	0.12 &	0.03 &	0.03 &	0.87 &	0.86 &	0.51 &	0.87 &	0.86 \\
             &\multicolumn{1}{l}{0.3} &0.03 &	0.03 &	0.06 	&0.03 &	0.03 &	0.87 	&0.86 	&0.74 &	0.87 &	0.86 \\
             & \multicolumn{1}{l}{0.4} &0.03 &	0.03 &	0.04 &	0.03 &	0.03 &	0.87 	&0.86 &	0.82 	&0.87 &	0.86 \\
             \hline \multirow{4}{*}{(iv)}& \multicolumn{1}{l}{0.1} &0.02 &	0.02 &	0.29 &	0.02 &	0.02 &	0.88 &	0.89 	&0.13 &	0.87 &	0.88 \\ 
             &\multicolumn{1}{l}{0.2} &0.02 &	0.02 &	0.09 &	0.02 &	0.02 &	0.88 &	0.89 &	0.48 &	0.87 &	0.88 \\
             &\multicolumn{1}{l}{0.3} &0.02 &	0.02 &	0.04 &	0.02 &	0.02 	&0.88 &	0.89 	&0.75 &	0.87 &	0.88 \\
             & \multicolumn{1}{l}{0.4} &0.02 	&0.02 &	0.03 &	0.02 &	0.02 &	0.88 &	0.89 	&0.85 &	0.87 &	0.88 \\
             \hline 
		\end{tabular}
		\begin{tablenotes}
			\footnotesize
			\item[]Note: $\beta_{T,t} = \mu_{T,t} + T^{-1/2} \sum_{i=1}^{t}v_i$, where $\mu_{T,t} = 0$ for $t\leq 0.5T$ and $\mu_{T,t}=2/T^{\alpha}$ for $t>0.5T$, and $v_i \sim \mathrm{i.i.d.} \ N(0,1)$. $\beta_{T,t}$ with $t=\tau T$ is estimated using bandwidth $h = T^{-0.5}$. The sample size is $T=100$ for case (i), $T=200$ for case (ii) $T=400$ for case (iii), and $T=800$ for case (iv).
		\end{tablenotes}
	\end{threeparttable} \label{tab:sb_ignore_iid}
\end{table}

\clearpage
\begin{table} \caption{MSE and coverage rate when $\varepsilon_{t}$ is GARCH}
	\centering
	\begin{threeparttable}
		\renewcommand{\arraystretch}{1.3}	\begin{tabular}{ccc@{\hskip 12pt}c@{\hskip 13pt}c@{\hskip 13pt}c@{\hskip 13pt}c@{\hskip 26pt}c@{\hskip 13pt}c@{\hskip 13pt}c@{\hskip 13pt}c@{\hskip 13pt}c}
			\hline & $\alpha$ & \multicolumn{4}{c}{MSE} & & \multicolumn{4}{c}{Coverage Rate} & \\
			\hline  &  & $\tau$ & \multicolumn{3}{}{} & & $\tau$ & \multicolumn{3}{}{} & \\
             \cline{3-12} & & 0.4 & 0.45 & 0.5 & 0.55 &0.6 & 0.4 & 0.45 & 0.5 & 0.55 & 0.6 \\
             \hline \multirow{4}{*}{(i)}& \multicolumn{1}{l}{0.1} &0.07 &	0.11 &	0.44 	&0.15 &	0.06 	&0.82 &	0.77 &	0.32 &	0.68 &	0.84 \\ 
             &\multicolumn{1}{l}{0.2} &0.07& 	0.09 &	0.22 	&0.10 &	0.06 	&0.81 &	0.79 &	0.48 &	0.74 &	0.82 \\
             &\multicolumn{1}{l}{0.3} &0.07 &	0.08 &	0.13 &	0.08 &	0.06 &	0.80 &	0.80 &	0.63 &	0.79 &	0.82\\
             & \multicolumn{1}{l}{0.4} &0.07 &	0.07 	&0.09 	&0.07 &	0.06 &	0.80 	&0.81 	&0.73 &	0.79 &	0.82 \\
             \hline \multirow{4}{*}{(ii)}& \multicolumn{1}{l}{0.1} &0.05 &	0.05 &	0.40 &	0.06 &	0.04 &	0.84 &	0.85 &	0.26 &	0.82 &	0.84\\ 
             &\multicolumn{1}{l}{0.2} &0.05 &	0.04 &	0.17 &	0.05 &	0.04 	&0.84 &	0.84 &	0.46 &	0.83 	&0.84\\
             &\multicolumn{1}{l}{0.3} &0.05 &	0.04 &	0.09 &	0.05 &	0.04 &	0.84 &	0.84 &	0.65 &	0.83 	&0.84 \\
             & \multicolumn{1}{l}{0.4} &0.05 &	0.04 &	0.07 &	0.05 &	0.04 &	0.84 &	0.85 &	0.75 &	0.83 &	0.84 \\
             \hline \multirow{4}{*}{(iii)}& \multicolumn{1}{l}{0.1} &0.03 	&0.03 	&0.34 &	0.03 &	0.03 &	0.84 &	0.86 &	0.18 &	0.87 &	0.85 \\ 
             &\multicolumn{1}{l}{0.2} &0.03& 	0.03 &	0.13 &	0.03 &	0.03 &	0.84 &	0.85 &	0.44 &	0.86 	&0.85 \\
             &\multicolumn{1}{l}{0.3} &0.03& 	0.03& 	0.06 &	0.03 &	0.03 &	0.84 &	0.85 &	0.68 &	0.86 &	0.85 \\
             & \multicolumn{1}{l}{0.4} &0.03 &	0.03 &	0.04 &	0.03 &	0.03 &	0.84 &	0.85 &	0.79 &	0.86 &	0.85 \\
             \hline \multirow{4}{*}{(iv)}& \multicolumn{1}{l}{0.1} &0.02 &	0.02 &	0.29 &	0.02& 	0.02 &	0.86 &	0.88 &	0.12 &	0.86 &	0.86 \\ 
             &\multicolumn{1}{l}{0.2} &0.02 &	0.02 &	0.09 &	0.02 &	0.02 	&0.86 	&0.88 &	0.41 &	0.86 &	0.86 \\
             &\multicolumn{1}{l}{0.3} &0.02 &	0.02 &	0.04 &	0.02 &	0.02 &	0.86 &	0.88 &	0.69 &	0.86 &	0.86 \\
             & \multicolumn{1}{l}{0.4} &0.02 &	0.02 &	0.02 &	0.02 &	0.02 &	0.86 &	0.88 &	0.82 &	0.86 &	0.86 \\
             \hline 
		\end{tabular}
		\begin{tablenotes}
			\footnotesize
			\item[]Note: $\beta_{T,t} = \mu_{T,t} + T^{-1/2} \sum_{i=1}^{t}v_i$, where $\mu_{T,t} = 0$ for $t\leq 0.5T$ and $\mu_{T,t}=2/T^{\alpha}$ for $t>0.5T$, and $v_i \sim \mathrm{i.i.d.} \ N(0,1)$. $\beta_{T,t}$ with $t=\tau T$ is estimated using bandwidth $h = T^{-0.5}$. The sample size is $T=100$ for case (i), $T=200$ for case (ii) $T=400$ for case (iii), and $T=800$ for case (iv).
		\end{tablenotes}
	\end{threeparttable} \label{tab:sb_ignore_garch}
\end{table}

\clearpage
\begin{table} \caption{Mean MSE and coverage rate (CR) for the case of a smooth function with a jump}
	\centering
	\begin{threeparttable}
		\renewcommand{\arraystretch}{1.3}	\begin{tabular}{ccc@{\hskip 15pt}c@{\hskip 15pt}c@{\hskip 15pt}c@{\hskip 28pt}c@{\hskip 15pt}c@{\hskip 15pt}c@{\hskip 15pt}c}
			\hline & $\gamma$ & \multicolumn{3}{c}{MSE} & & \multicolumn{3}{r}{CR ($t=0.5T$)} & \\
			\hline  &  & $T$ & \multicolumn{2}{}{} & & $T$ & \multicolumn{2}{}{} & \\
             \cline{3-10} & & 100 & 200 & 400 & 800 & 100 & 200 & 400 & 800 \\
             \hline \multirow{5}{*}{i.i.d error} & \multicolumn{1}{l}{-0.2} &0.025 & 0.014 & 0.008 &	0.004 &	0.758 &	0.751 &	0.791 &	0.785   \\ 
             &\multicolumn{1}{l}{-0.33} &0.029 &	0.017 &	0.010 &	0.006 &	0.817 &	0.833 &	0.862 &	0.878  \\
             &\multicolumn{1}{l}{-0.5} &0.059 &	0.038 &	0.026 &	0.018 &	0.830 &	0.853 &	0.882 &	0.899  \\
             & \multicolumn{1}{l}{CV} &0.034 &	0.018 &	0.010 &	0.006 &	0.774 &	0.765 &	0.807 &	0.809   \\
             &\multicolumn{1}{l}{Boot} &0.059 &	0.038 &	0.025 &	0.014 &	0.830 &	0.853 &	0.883 &	0.896   \\
             \hline \multirow{5}{*}{GARCH error} & \multicolumn{1}{l}{-0.2} &0.025 &	0.014 &	0.007 &	0.004 &	0.709 &	0.720 &	0.770 &	0.770   \\ 
             &\multicolumn{1}{l}{-0.33} &0.029 &	0.017 &	0.010 &	0.006 &	0.774 &	0.801 &	0.844 &	0.861   \\ 
             &\multicolumn{1}{l}{-0.5} &0.060 &	0.038 &	0.025 &	0.017 &	0.801 &	0.825 &	0.874 &	0.884  \\
             & \multicolumn{1}{l}{CV} &0.035 &	0.020 &	0.011 &	0.006 &	0.731 &	0.743 &	0.798 &	0.804    \\
             &\multicolumn{1}{l}{Boot} &0.060 &	0.038 &	0.025 &	0.014 &	0.801 &	0.826 &	0.874 &	0.882    \\
             \hline 
		\end{tabular}
		\begin{tablenotes}
			\footnotesize
			\item[]Note: $\beta_{T,t} = \beta(t/T)$, where $\beta(x) = x + \mu_T(x)$ with $\mu_T(x) = 0$ for $x\leq 0.5$ and $\mu_T(x) = 1.5/T^{0.4}$ for $x>0.5$. $\hat{\beta}_t$ is calculated using bandwidth parameter $h=T^{\gamma}$. The rows labeled ``CV" and ``Boot" signify the results for cross-validation-based and bootstrap-based selections, respectively 
		\end{tablenotes}
	\end{threeparttable} \label{tab:smooth_break_dd}
\end{table} 

\clearpage
\begin{table} \caption{Mean empirical coverage rates of 95\% bootstrap confidence intervals for V-G}
	\centering
	\begin{threeparttable}
		\renewcommand{\arraystretch}{1.3}	\begin{tabular}{ccc@{\hskip 21pt}c@{\hskip 21pt}c@{\hskip 21pt}c}
			\hline  &  & \multicolumn{4}{c}{$\gamma_2$} \\
             \cline{3-6}& & -0.2 & -0.33 & -0.4 & -0.5 \\
             \hline \multirow{5}{*}{$\gamma_1$}& \multicolumn{1}{l}{-0.2} &0.861&	0.936&	0.937&	0.929 \\ 
             &\multicolumn{1}{l}{-0.33} &-	&0.917	&0.923	&0.917 \\
             &\multicolumn{1}{l}{-0.4} &-	&-	&0.902	&0.918\\
             & \multicolumn{1}{l}{-0.5} &-	&-	&-	&0.889\\
             \hline 
		\end{tabular}
		\begin{tablenotes}
			\footnotesize
			\item[]Note: Each entry denotes the mean empirical coverage rate of the 95\% bootstrap confidence intervals for $(\hat{\alpha}_{j,t}(\gamma_1), \hat{\beta}_{j,t}(\gamma_1))$ based on $(\hat{\alpha}_{j,t}^*(\gamma_1,\gamma_2), \hat{\beta}_{j,t}^*(\gamma_1,\gamma_2))$ taken over $t=1,\ldots,T$: $\overline{\mathrm{CR}}(\gamma_1,\gamma_2)= T^{-1}\sum_{t=1}^T\mathrm{CR}_t(\gamma_1,\gamma_2)$.
		\end{tablenotes}
	\end{threeparttable} \label{tab:gamma_selection_cr_vmg}
\end{table}

\clearpage
\begin{figure}
    \centering
    \includegraphics[width=0.85\textwidth]{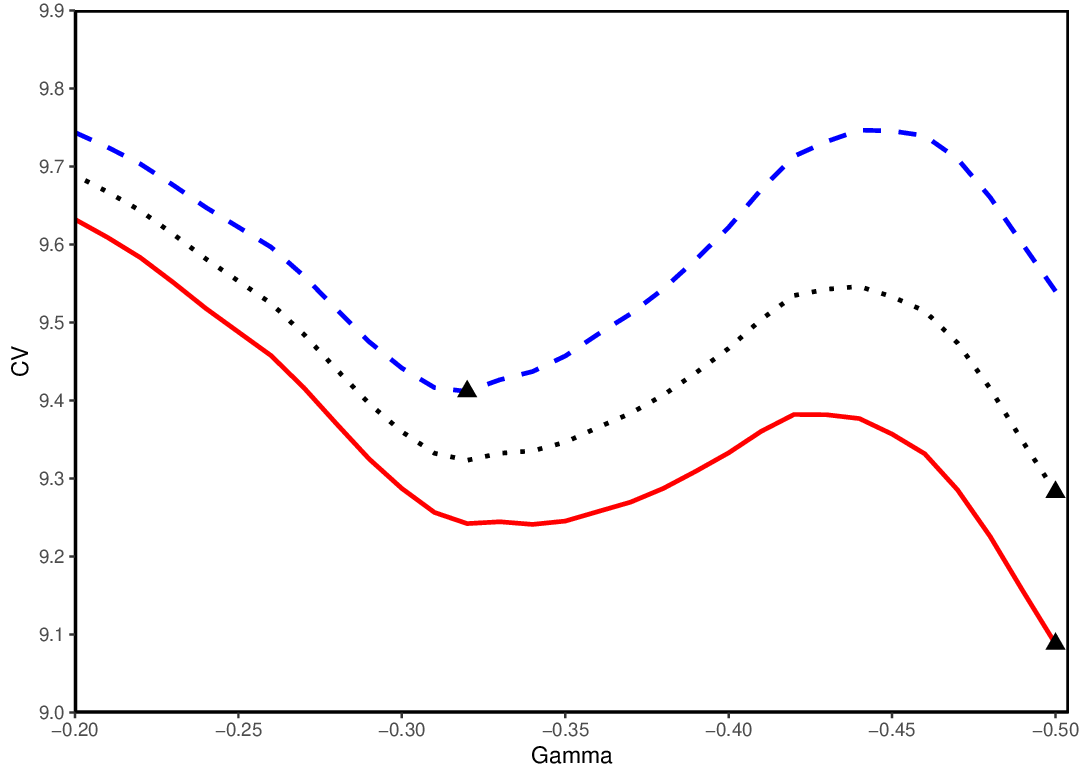}
    \caption{Cross-validation criteria calculated using leave-$(2m+1)$-out estimators with $h=T^{\gamma}$, for V-G}
    \caption*{\color{red}\full\color{black}: $m=0$, \quad \color{black}\dotted\color{black}: $m=1$, \quad \color{blue}\dashed\color{black}: $m=2$, \quad $\blacktriangle$: Minimum}
    \label{fig:CV_gamma_vmg}
\end{figure}

\clearpage
\begin{figure}
    \centering
    \begin{subfigure}{0.52\textheight}
        \includegraphics[width=\columnwidth]{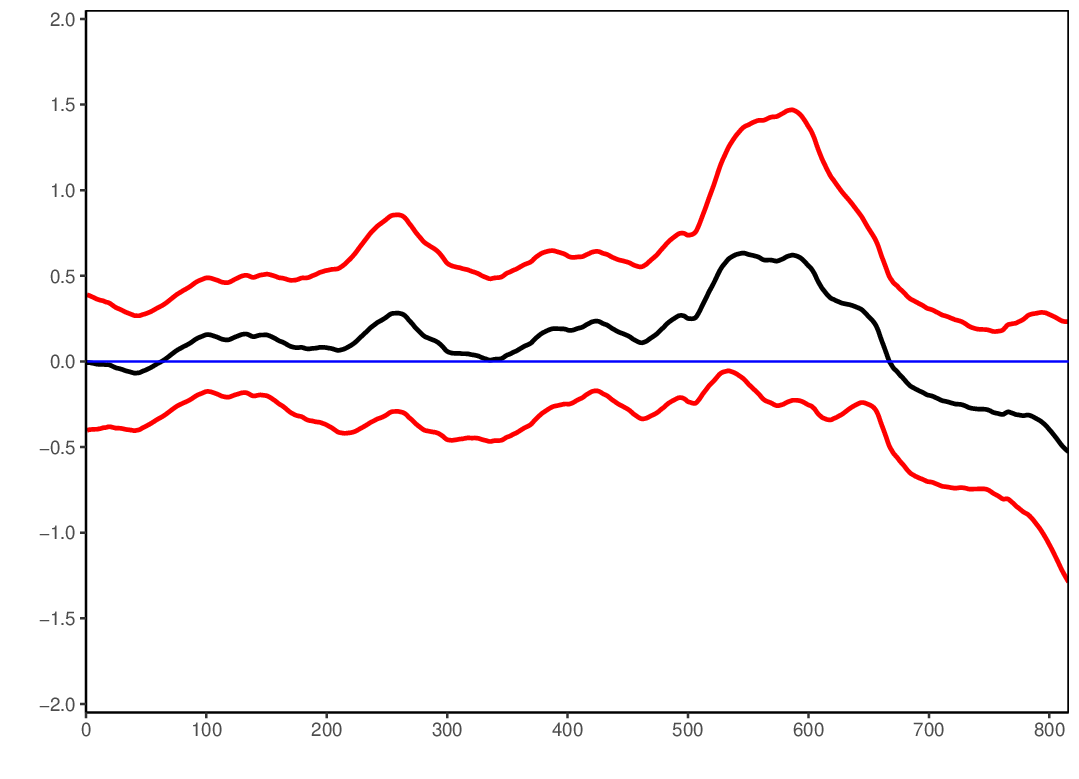}
    \caption{Plot of the time-varying alpha}
	\label{fig:vmg_int_kernel}
    \end{subfigure}
    \begin{subfigure}{0.52\textheight}
        \includegraphics[width=\columnwidth]{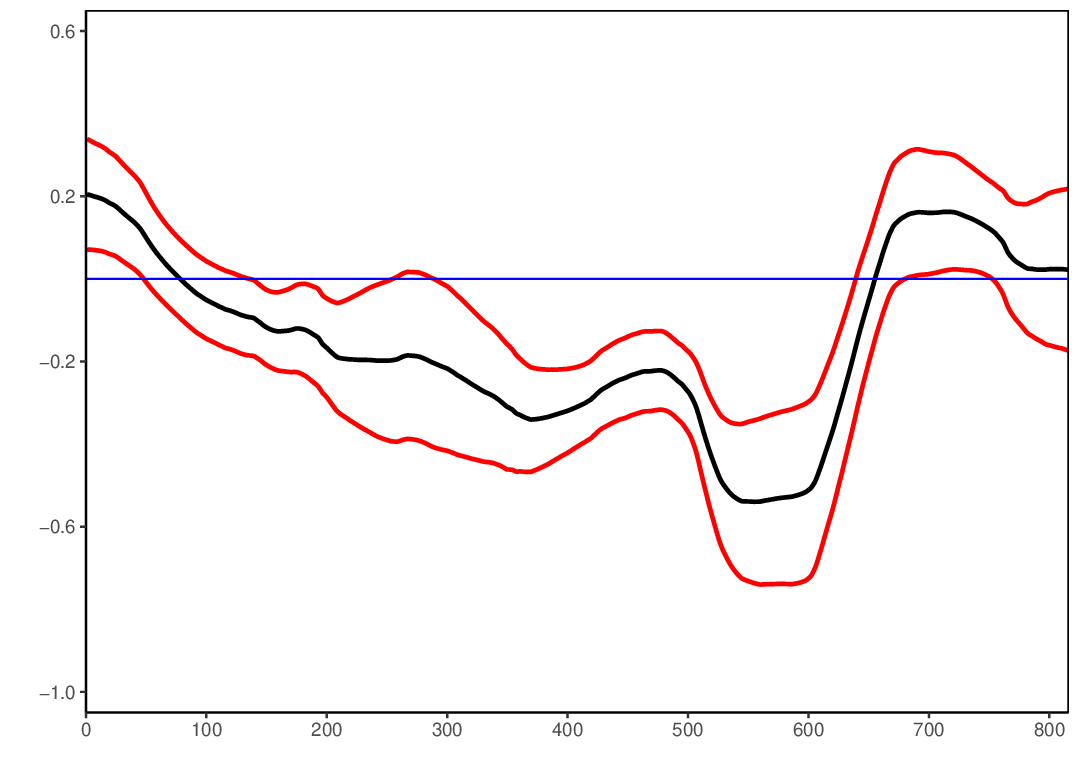}
    \caption{Plot of the time-varying beta}
	\label{fig:vmg_slope_kernel}
    \end{subfigure}
    \caption{Estimates and 95\% confidence band from the kernel-based method ($h=\hat{c}T^{-1/3}$) for V-G}
    (Horizontal lines in (a) and (b) indicate the value zero.)
    \label{fig:vmg}
\end{figure}

\clearpage
\begin{figure}
    \centering
    \begin{subfigure}{0.56\textheight}
        \includegraphics[width=\columnwidth]{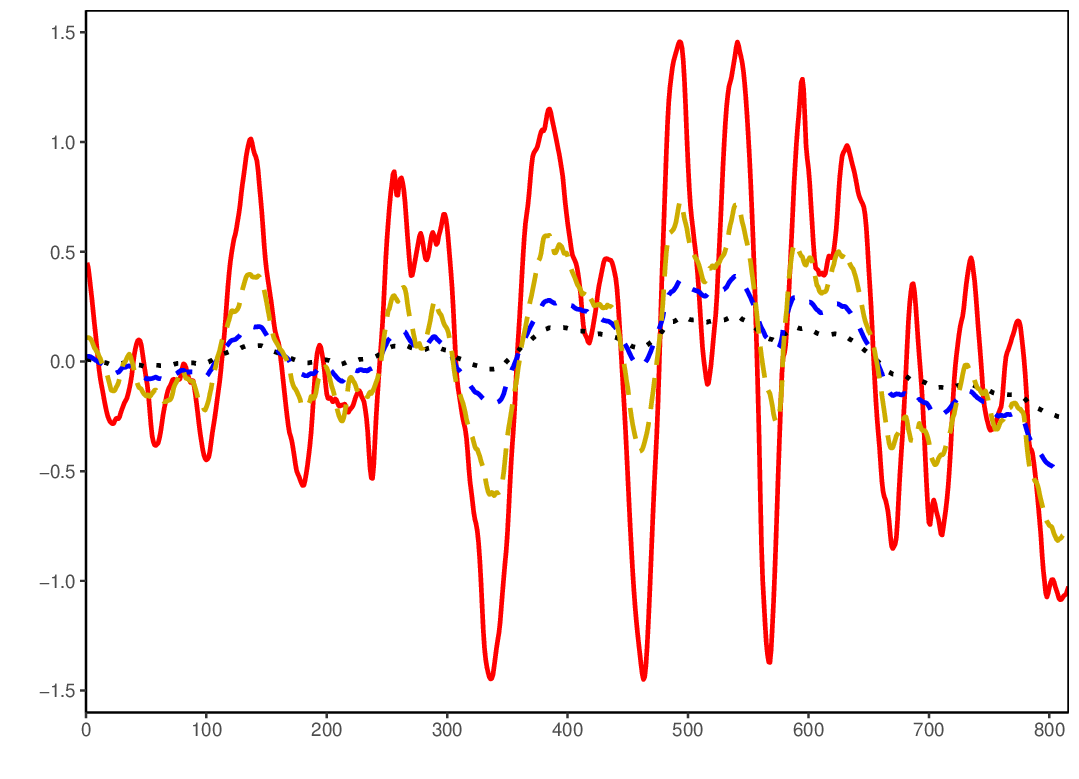}
    \caption{Plot of the time-varying alpha}
	\label{fig:vmg_int_comp}
    \end{subfigure}
    \begin{subfigure}{0.56\textheight}
        \includegraphics[width=\columnwidth]{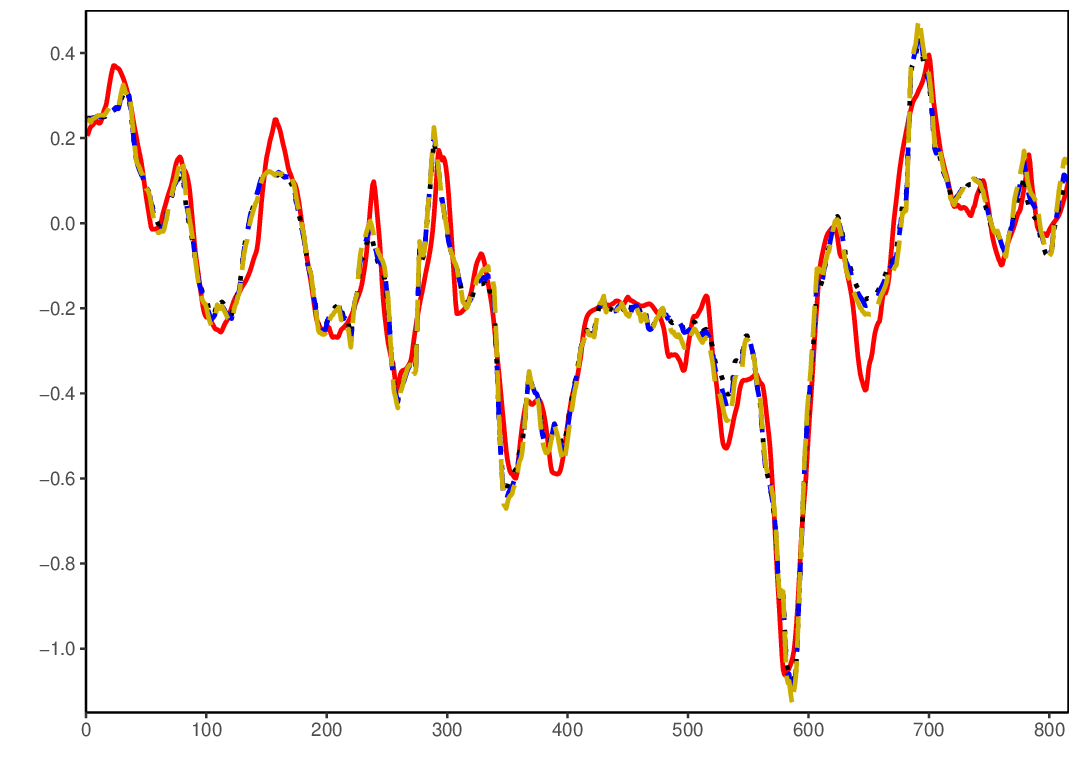}
    \caption{Plot of the time-varying beta}
	\label{fig:vmg_slope_comp}
    \end{subfigure}
    \caption{Estimates from the kernel method with $h=\hat{c}T^{-1/2}$ and Bayesian method (posterior means) for V-G}
    \label{fig:vmg_comp}
    \caption*{\color{red}\full\color{black}: Kernel, \quad \color{black}\dotted\color{black}: Bayesian ($v_1=1$), \quad \color{blue}\dashed\color{black}: Bayesian ($v_1=2$), \quad \color{gold(metallic)}\longdash\color{black}: Bayesian ($v_1=4$)}
\end{figure}

\clearpage

\titleformat*{\section}{\Large\centering}
\section*{Appendix to ``Estimating Time-Varying Parameters of Various Smoothness in Linear Models via Kernel Regression" by M. Nishi}

Throughout the Appendix, $C>0$ is a generic constant that may vary across lines.

\titleformat*{\section}{\Large\bfseries\centering}

\section*{Appendix A: Proofs of the Main Results}
\setcounter{equation}{0}
\renewcommand{\theequation}{A.\arabic{equation}}

\begin{lemapp}
	Under Assumptions \ref{asm:kernel} and \ref{asm:dgp}, for each $t=\lfloor Tr \rfloor, \ r\in(0,1)$, we have $(1/Th)\sum_{i=1}^{T}K((t-i)/Th)x_ix_i' \stackrel{p}{\to} \Omega(r)$, where $\Omega(r) = \lim_{T\to\infty}(1/Th)\sum_{i=1}^{T}K((t-i)/Th)E[x_ix_i']$.
	\label{lemapp:variance_plim}
\end{lemapp}
\noindent \begin{proof}
    Decompose $(1/Th)\sum_{i=1}^{T}K((t-i)/Th)x_ix_i'$ as
	\begin{align}
		\frac{1}{Th}\sum_{i=1}^{T}K\Bigl(\frac{t-i}{Th}\Bigr)x_ix_i' &= \frac{1}{Th}\sum_{i=1}^{T}K\Bigl(\frac{t-i}{Th}\Bigr)E[x_ix_i'] + \frac{1}{Th}\sum_{i=1}^{T}K\Bigl(\frac{t-i}{Th}\Bigr)\bigl(x_ix_i' - E[x_ix_i']\bigr) \\
		&\eqqcolon A_{T,1} + A_{T,2}.
	\end{align}
	Because $A_{T,1} \to \Omega(r)$ by Assumption \ref{asm:dgp}(c), it suffices to show $A_{T,2} = o_p(1)$. Following the argument of Example 17.17 of \citetapp{davidsonStochasticLimitTheory1994}, we can show that $\{K((t-i)/Th)(x_ix_i' - E[x_ix_i'])\}_i$ is an $L_{r}$-bounded ($r>2$), mean-zero $L_2$-NED triangular array under Assumptions \ref{asm:kernel} and \ref{asm:dgp}(a), and thus it is a uniformly integrable $L_2$-mixingale \citepapp[see][]{andrewsLawsLargeNumbers1988}. This result allows us to apply the law of large numbers (see \citetapp{andrewsLawsLargeNumbers1988}, p.464) and obtain
	\begin{align}
		A_{T,2} &= \frac{1}{Th}\sum_{i=0}^{\lfloor Th \rfloor}K\Bigl(\frac{i}{Th}\Bigr)\bigl(x_{t-i}x_{t-i}' - E[x_{t-i}x_{t-i}']\bigr) + \frac{1}{Th}\sum_{i=1}^{\lfloor Th \rfloor}K\Bigl(\frac{-i}{Th}\Bigr)\bigl(x_{t+i}x_{t+i}' - E[x_{t+i}x_{t+i}']\bigr) \\
		&\stackrel{p}{\to}0,
	\end{align}
	which, together with $A_{T,1}\to\Omega(r)$, shows that $(1/Th)\sum_{i=1}^{T}K((t-i)/Th)x_ix_i' \stackrel{p}{\to} \Omega(r)$.
\end{proof}

\begin{lemapp}
	Under Assumptions \ref{asm:kernel} and \ref{asm:dgp}, for each $t=\lfloor Tr \rfloor, \ r\in(0,1)$, we have
	\begin{align}
		\left\|\sum_{i=1}^{T}K\Bigl(\frac{t-i}{Th}\Bigr)x_ix_i'(\beta_{T,i} - \beta_{T,t}) \right\| = \begin{cases}
			O_p(Th^{1+\alpha}) & \text{if $\beta_{T,t}$ satisfies Definition \ref{def:tvp_alpha_class}(a)} \\
			O_p(T^{1-\alpha}h) & \text{if $\beta_{T,t}$ satisfies Definition \ref{def:tvp_alpha_class}(b)}
		\end{cases}.
	\end{align}
	\label{lemapp:bias_order}
\end{lemapp}

\noindent \begin{proof}
	First, we have
	\begin{align}
		\left\|\sum_{i=1}^{T}K\Bigl(\frac{t-i}{Th}\Bigr)x_ix_i'(\beta_{T,i} - \beta_{T,t}) \right\| &= \left\|\sum_{i=t-\lfloor Th \rfloor}^{t+\lfloor Th \rfloor}K\Bigl(\frac{t-i}{Th}\Bigr)x_ix_i'(\beta_{T,i} - \beta_{T,t}) \right\| \\
		&\leq \max_{t-\lfloor Th \rfloor\leq i \leq t+\lfloor Th \rfloor}\|\beta_{T,i} - \beta_{T,t}\| \\ 
        & \qquad \times \sum_{i = t-\lfloor Th \rfloor}^{t+\lfloor Th \rfloor}K\Bigl(\frac{t-i}{Th}\Bigr)\|x_ix_i'\|, \label{eqn:bias_upper}
	\end{align}
        because the support of $K$ is $[-1,1]$ under Assumption \ref{asm:kernel}. Note that
	\begin{align}
		\sum_{i = t-\lfloor Th \rfloor}^{t+\lfloor Th \rfloor}K\Bigl(\frac{t-i}{Th}\Bigr)\|x_ix_i'\| &= O_p(Th), \label{eqn:bias_1} \\
		\intertext{because}
		E\biggl[\sum_{i = t-\lfloor Th \rfloor}^{t+\lfloor Th \rfloor}K\Bigl(\frac{t-i}{Th}\Bigr)\|x_ix_i'\|\biggr] &\leq \max_iE\bigl[\|x_ix_i'\|\bigr] \sum_{i = t-\lfloor Th \rfloor}^{t+\lfloor Th \rfloor}K\Bigl(\frac{t-i}{Th}\Bigr) \\
		&\leq \sup_tE[\|x_t\|^2] \times Th\times\frac{1}{Th}\sum_{i=-\lfloor Th \rfloor}^{\lfloor Th \rfloor}K\Bigl(\frac{i}{Th}\Bigr) \\
		&= O(1) \times Th\sum_{i=-\lfloor Th \rfloor}^{\lfloor Th \rfloor}\int_{(i-1)/Th}^{i/Th}K\Bigl(\frac{i}{Th}\Bigr)dr \\
		&= O(1) \times Th\sum_{i=-\lfloor Th \rfloor}^{\lfloor Th \rfloor}\int_{(i-1)/Th}^{i/Th}\biggl\{K\Bigl(\frac{i}{Th}\Bigr) - K(r)+K(r)\biggr\}dr \\
		&= O(1) \times Th\Biggl(\int_{-\lfloor Th \rfloor/Th}^{\lfloor Th \rfloor/Th}K(r)dr + O(1/Th)\Biggr) = O(Th),
	\end{align}
	because $\sup_tE[\|x_t\|^2] < \infty$ under Assumption \ref{asm:dgp}, $K$ is Lipschitz continuous, and $\int_{-1}^{1}K(x)dx=1$ under Assumption \ref{asm:kernel}. We also have
	\begin{align}
		\max_{t-\lfloor Th \rfloor\leq i \leq t+\lfloor Th \rfloor}\|\beta_{T,i} - \beta_{T,t}\|= \begin{cases}
			O_p(h^{\alpha}) & \text{if $\beta_{T,t}$ satisfies Definition \ref{def:tvp_alpha_class}(a)} \\
			O_p(T^{-\alpha}) & \text{if $\beta_{T,t}$ satisfies Definition \ref{def:tvp_alpha_class}(b)} \\
		\end{cases}.
	\label{eqn:bias_2}
	\end{align}
	Substituting \eqref{eqn:bias_1} and \eqref{eqn:bias_2} into \eqref{eqn:bias_upper}, we deduce
	\begin{align}
		\left\|\sum_{i=1}^{T}K\Bigl(\frac{t-i}{Th}\Bigr)x_ix_i'(\beta_{T,i} - \beta_{T,t}) \right\|  = \begin{cases}
			O_p(Th^{1+\alpha}) & \text{if $\beta_{T,t}$ satisfies Definition \ref{def:tvp_alpha_class}(a)} \\
			O_p(T^{1-\alpha}h) & \text{if $\beta_{T,t}$ satisfies Definition \ref{def:tvp_alpha_class}(b)}
		\end{cases}.
	\end{align}
\end{proof}

\begin{lemapp}
	Under Assumptions \ref{asm:kernel} and \ref{asm:dgp}, for each $t=\lfloor Tr \rfloor, \ r\in(0,1)$, we have
	\begin{align}
		\frac{1}{\sqrt{Th}}\sum_{i=1}^{T}K\Bigl(\frac{t-i}{Th}\Bigr)x_i\varepsilon_i \stackrel{d}{\to}N(0,\Sigma(r)).
	\end{align}
	\label{lemapp:weak_convergence}
\end{lemapp}

\noindent \begin{proof}
	To prove this result, we use the Cramer-Wold device. Define $z_{T,i}^* \coloneqq \lambda'K((t-i)/Th)x_i\varepsilon_i$, where $\lambda\in\mathbb{R}^p$ is any vector such that $\lambda'\lambda=1$, $\sigma_T^2 \coloneqq \mathrm{Var}\bigl(\sum_{i=1}^{T}z_{T,i}^*\bigr)$, and $z_{T,i} \coloneqq z_{T,i}^*/\sigma_T$. Note that $\sigma_T^2/Th\to \lambda'\Sigma(r)\lambda>0$ by Assumption \ref{asm:dgp}(c). Moreover, define positive constant array $\{c_{T,i}\}$ as 
    \begin{align}
        c_{T,i} = \begin{cases}
            \max\Bigl\{\sqrt{\mathrm{Var}(z_{T,i}^*)}, 1\Bigr\}/\sigma_T & \mathrm{for} \ i \in [t-\lfloor Th \rfloor, t + \lfloor Th \rfloor] \\
            1/\sqrt{T} & \mathrm{otherwise}
        \end{cases}.
    \end{align}
    To show Lemma \ref{lemapp:weak_convergence}, we rely on Theorem 2 of \citetapp{dejongCentralLimitTheorems1997}, which requires that the following conditions hold for $\{z_{T,i}, c_{T,i}\}$:
    \begin{itemize}
        \item[(i)] $z_{T,i}$ has mean zero, and $\mathrm{Var}(\sum_{i=1}^{T}z_{T,i})=1$.
        \item[(ii)] $z_{T,i}/c_{T,i}$ is $L_r$-bounded for some $r>2$ uniformly in $i$ and $T$.
        \item[(iii)] $z_{T,i}$ is $L_2$-NED of size $-1/2$ on an $\alpha$-mixing array of size $-r/(r-2)$, with respect to some constants $d_{T,i}$. Moreover, $d_{T,i}/c_{T,i}$ is bounded uniformly in $i$ and $T$.
        \item[(iv)] Let $b_T$ be a positive non-decreasing integer-valued sequence such that $b_T\leq T$, $b_T \to \infty$, and $b_T/T \to 0$ as $T\to \infty$. Also let $r_T \coloneqq \lfloor T/b_T \rfloor$. Define $M_{T,j} \coloneqq \max_{(j-1)b_T+1\leq i \leq jb_T} c_{T,i}, \ j=1,\ldots,r_T$, and $M_{T,r_T+1} \coloneqq \max_{r_Tb_T+1\leq i \leq T}c_{T,i}$. Then, we have $\max_{1\leq j \leq r_T+1} M_{T,j} = o(b_T^{-1/2})$ and $\sum_{j=1}^{r_T}M_{T,j}^2 = O(b_T^{-1})$.
    \end{itemize}
    Conditions (i)-(iv) imply that $\sum_{i=1}^Tz_{t,i} \stackrel{d}{\to} N(0,1)$. We show that the above four conditions hold.

    (i) This condition trivially follows from Assumption \ref{asm:dgp}(b) and the definition of $z_{T,i}$.

    (ii) Noting that $z_{T,i}^*=0$ for $i<t-\lfloor Th \rfloor$ and $i>t+\lfloor Th \rfloor$, we have
    \begin{align}
        z_{T,i}/c_{T,i} = \begin{cases}
            z_{T,i}^*/\max\Bigl\{\sqrt{\mathrm{Var}(z_{T,i}^*)},1\Bigr\} & \mathrm{for} \ i\in [t-\lfloor Th \rfloor, t+\lfloor Th \rfloor] \\
            0 & \mathrm{otherwise}
            \label{eqn:z_over_c}
        \end{cases}.
    \end{align}
    Because $(x_i',\varepsilon_i)$ is uniformly $L_{2r}$-bounded for $r>2$ by Assumption \ref{asm:dgp}(a), $z_{T,i}^*$ is $L_r$-bounded uniformly in $i$ and $T$ since kernel $K(\cdot)$ is bounded. This implies that  $z_{T,i}/c_{T,i}$ is also $L_r$-bounded uniformly in $i$ and $T$ in view of \eqref{eqn:z_over_c}.

    (iii) Note that $(x_i',\varepsilon_i)$ is $L_{2r}$-bounded and $L_2$-NED of size $-(r-1)/(r-2)$ on an $\alpha$-mixing sequence of size $-r/(r-2)$. Thus, following the argument of Example 17.17 of \citetapp{davidsonStochasticLimitTheory1994}, we can show that $z_{T,i}$ is $L_2$-NED of size $-1/2$ on the same $\alpha$-mixing sequence, with respect to positive constant array $d_{T,i}$ satisfying
    \begin{align}
        \sup_{t-\lfloor Th \rfloor\leq i \leq t+\lfloor Th \rfloor}d_{T,i} \leq \frac{C}{\sigma_T} = O\Biggl(\frac{1}{\sqrt{Th}}\Biggr),
    \end{align}
    for some positive constant $C<\infty$ independent of $T$, and $d_{T,i}=0$ for $i\notin [t-\lfloor Th \rfloor, t+\lfloor Th \rfloor]$. This follows from the fact that $z_{T,i}^*$ is $L_2$-NED of size $-1/2$ with respect to some positive constant array $d_{T,i}^*$ satisfying $\sup_{T,i}d_{T,i}^*<\infty$ under Assumption \ref{asm:dgp}(a), $\sigma_T^2/Th\to \lambda'\Sigma(r)\lambda>0$ as $T\to\infty$, and $K((t-i)/Th) = 0$ for $i<t-\lfloor Th \rfloor$ and $i>t+\lfloor Th \rfloor$. This implies that $d_{T,i}/c_{T,i}$ is bounded uniformly in $i$ and $T$.

    (iv) Let $b_T = \sqrt{Th}$. Then, by the definition of $c_{T,i}$ and the fact that $\sigma_T/\sqrt{Th}\geq C>0$ for sufficiently large $T$ and $\mathrm{Var}(z_{T,i}^*) < \infty$ uniformly in $i$ and $T$ by Assumption \ref{asm:dgp}(a), we get
    \begin{align}
        \max_{1\leq j \leq r_T+1} M_{T,j} = O((Th)^{-1/2}) = o(b_T^{-1/2}).
    \end{align}
    Furthermore, letting $j_1 \coloneqq \lfloor (t-\lfloor Th \rfloor)/b_T \rfloor$ and $j_2 \coloneqq \lfloor (t+\lfloor Th \rfloor)/b_T \rfloor$, we obtain
    \begin{align}
        \sum_{j=1}^{r_T}M_{T,j}^2 &= \sum_{j=1}^{j_1}M_{T,j}^2 + \sum_{j=j_1+1}^{j_2}M_{T,j}^2 + \sum_{j=j_2+1}^{r_T}M_{T,j}^2 \\
        &= \frac{j_1}{T} + O\Biggl(\frac{j_2-j_1}{Th}\Biggr) + \frac{r_T-j_2}{T} = O(b_T^{-1}).
    \end{align}
     
	Now that conditions (i)-(iv) are seen to hold, we obtain
    \begin{align}
        \sum_{i=1}^{T}z_{T,i} = \lambda'\sum_{i=1}^{T}K\Bigl(\frac{t-i}{Th}\Bigr)x_i\varepsilon_i/\sigma_T \stackrel{d}{\to} N(0,1).
    \end{align}
     Moreover, we have
    \begin{align}
        \frac{1}{Th}\sigma_T^2 = \frac{1}{Th}\lambda'\mathrm{Var}\Biggl(\sum_{i=1}^{T}K\Bigl(\frac{t-i}{Th}\Bigr)x_i\varepsilon_i\Biggr)\lambda \to \lambda'\Sigma(r)\lambda>0,
    \end{align}
    by Assumption \ref{asm:dgp}(c). This implies that
	\begin{align}
		\lambda'\frac{1}{\sqrt{Th}}\sum_{i=1}^{T}K\Bigl(\frac{t-i}{Th}\Bigr)x_i\varepsilon_i \stackrel{d}{\to} N(0,\lambda'\Sigma(r)\lambda).
	\end{align}
	 By the Cramer-Wold device, we deduce
	\begin{align}
		\frac{1}{\sqrt{Th}}\sum_{i=1}^{T}K\Bigl(\frac{t-i}{Th}\Bigr)x_i\varepsilon_i \stackrel{d}{\to}N(0,\Sigma(r)).
	\end{align}
\end{proof}

\noindent \begin{proof}[Proof of Theorem \ref{thm:bias_bandwidth}]
	Since
	\begin{align}
		\hat{\beta}_t &= \Biggl(\frac{1}{Th}\sum_{i=1}^{T}K\Bigl(\frac{t-i}{Th}\Bigr)x_ix_i'\Biggr)^{-1}\frac{1}{Th}\sum_{i=1}^{T}K\Bigl(\frac{t-i}{Th}\Bigr) x_ix_i'\beta_{T,i} \\
		&\qquad + \Biggl(\frac{1}{Th}\sum_{i=1}^{T}K\Bigl(\frac{t-i}{Th}\Bigr)x_ix_i'\Biggr)^{-1}\frac{1}{Th}\sum_{i=1}^{T}K\Bigl(\frac{t-i}{Th}\Bigr) x_i\varepsilon_i, 
	\end{align}
	we have
	\begin{align}
		\sqrt{Th}(\hat{\beta}_t - \beta_{T,t} - R_{T,t}) 
        &=\Biggl(\frac{1}{Th}\sum_{i=1}^{T}K\Bigl(\frac{t-i}{Th}\Bigr)x_ix_i'\Biggr)^{-1}\frac{1}{\sqrt{Th}}\sum_{i=1}^{T}K\Bigl(\frac{t-i}{Th}\Bigr) x_i\varepsilon_i,
        \label{eqn:beta_hat_decompose}
	\end{align}
	where $R_{T,t} \coloneqq \Bigl(\frac{1}{Th}\sum_{i=1}^{T}K\Bigl(\frac{t-i}{Th}\Bigr)x_ix_i'\Bigr)^{-1}\frac{1}{Th}\sum_{i=1}^{T}K\Bigl(\frac{t-i}{Th}\Bigr) x_ix_i'(\beta_{T,i}-\beta_{T,t})$. It follows from Lemmas \ref{lemapp:variance_plim} and \ref{lemapp:weak_convergence} that
	\begin{align}
		\Biggl(\frac{1}{Th}\sum_{i=1}^{T}K\Bigl(\frac{t-i}{Th}\Bigr)x_ix_i'\Biggr)^{-1}\frac{1}{\sqrt{Th}}\sum_{i=1}^{T}K\Bigl(\frac{t-i}{Th}\Bigr) x_i\varepsilon_i &\stackrel{d}{\to} \Omega(r)^{-1} \times N(0,\Sigma(r)) \\ 
		&= N(0,\Omega(r)^{-1}\Sigma(r)\Omega(r)^{-1}).
	\end{align}
	The bias term, $R_{T,t}$, satisfies
	\begin{align}
		R_{T,t} = \begin{cases}
			O_p(h^{\alpha}) & \text{if $\beta_{T,t}$ satisfies Definition \ref{def:tvp_alpha_class}(a)} \\
			O_p(T^{-\alpha}) & \text{if $\beta_{T,t}$ satisfies Definition \ref{def:tvp_alpha_class}(b)}
		\end{cases},
	\end{align}
	by Lemmas \ref{lemapp:variance_plim} and \ref{lemapp:bias_order}.
	
	Set $h=cT^{\gamma}$ for some $c>0$ and $\gamma\in(-1,0)$. Because $\sqrt{Th}R_{T,t} = O_p(T^{1/2+\gamma(1/2+\alpha)})$ for the type-a $\mathrm{TVP}(\alpha)$ case and $\sqrt{Th}R_{T,t} = O_p(T^{1/2-\alpha+\gamma/2})$ for the type-b $\mathrm{TVP}(\alpha)$ case, $\sqrt{Th}R_{T,t} = o_p(1)$ if
	\begin{align}
		\gamma \in \begin{cases}
			(-1,-\frac{1}{2\alpha+1}) & \text{if $\beta_{T,t}$ satisfies Definition \ref{def:tvp_alpha_class}(a)} \\
			(-1,2\alpha -1)\cap (-1,0) & \text{if $\beta_{T,t}$ satisfies Definition \ref{def:tvp_alpha_class}(b)}
		\end{cases},
	\end{align}
	under which choice we obtain
	\begin{align}
		\sqrt{cT^{1+\gamma}}(\hat{\beta}_t - \beta_{T,t}) &= \Biggl(\frac{1}{Th}\sum_{i=1}^{T}K\Bigl(\frac{t-i}{Th}\Bigr)x_ix_i'\Biggr)^{-1}\frac{1}{\sqrt{Th}}\sum_{i=1}^{T}K\Bigl(\frac{t-i}{Th}\Bigr) x_i\varepsilon_i + o_p(1) \\ 
		&\stackrel{d}{\to} N(0,\Omega(r)^{-1}\Sigma(r)\Omega(r)^{-1}).
	\end{align}
\end{proof}

\noindent \begin{proof}[Proof of Corollary \ref{cor:bias_bandwidth}]
    We show that Assumption \ref{asm:dgp}(c) holds with $\Omega(r)=\Omega=E[x_1x_1']$ and $\Sigma(r)=\Sigma=\int_{-1}^{1}K(x)^2dxE[\varepsilon_1^2x_1x_1']$ under Assumptions \ref{asm:kernel} and \ref{asm:dgp}(a)-(b) and covariance-stationarity.

	First, we show $(1/Th)\sum_{i=1}^{T}K((t-i)/Th)E[x_ix_i'] \to \Omega = E[x_1x_1']$. By the covariance-stationarity of $x_i$ and Assumption \ref{asm:kernel}, we have
	\begin{align}
		\frac{1}{Th}\sum_{i=1}^{T}K\left(\frac{t-i}{Th}\right)E[x_ix_i'] &=E[x_1x_1']\frac{1}{Th}\sum_{i=-\lfloor Th \rfloor}^{\lfloor Th \rfloor}K\Bigl(\frac{i}{Th}\Bigr) \\
		&=E[x_1x_1']\int_{-\lfloor Th \rfloor/Th}^{\lfloor Th \rfloor/Th}K(r)dr + O(1/Th) \\
		&\to E[x_1x_1']\int_{-1}^{1}K(r)dr = E[x_1x_1'].
	\end{align}
    Similarly, noting that $x_t\varepsilon_t$ is serially uncorrelated under Assumption \ref{asm:dgp}(b), we have
    \begin{align}
        \mathrm{Var}\Biggl(\frac{1}{\sqrt{Th}}\sum_{i=1}^{T}K\Bigl(\frac{t-i}{Th}\Bigr)x_i\varepsilon_i\Biggr) &= \frac{1}{Th}\sum_{i=1}^{T}K\Bigl(\frac{t-i}{Th}\Bigr)^2E[\varepsilon_1^2x_1x_1'] \\
        &\to \int_{-1}^{1}K(x)^2dxE[\varepsilon_1^2x_1x_1'],
    \end{align}
    since $x_t\varepsilon_t$ is covariance-stationary.
\end{proof}

Define $t_j \coloneqq t+j$ and
\begin{align}
    \Delta_{t_j} \coloneqq \frac{1}{Th}\sum_{i=t_j-\lfloor Th \rfloor}^{t_j+\lfloor Th \rfloor}K\left(\frac{t_j-i}{Th}\right)\left(x_ix_i'-E\left[x_ix_i'\right]\right).
    \label{def:delta}
\end{align}

\begin{lemapp}
    Under Assumptions \ref{asm:kernel}, \ref{asm:dgp_2}, and \ref{asm:min_eigen}, for each $t=\lfloor Tr \rfloor, \ r\in(0,1)$, the following results hold:
    \begin{itemize}
        \item[(i)] There exists a constant $C>0$ such that for sufficiently large $T$,
        \begin{align}
            \min_{-\lfloor Th \rfloor\leq j \leq \lfloor Th \rfloor}\lambda_{\min}\left(\frac{1}{Th}\sum_{i=t_j-\lfloor Th \rfloor}^{t_j+\lfloor Th \rfloor}K\left(\frac{t_j-i}{Th}\right)E\left[x_ix_i'\right]\right) \geq C>0,
        \end{align}
        where $\lambda_{\min}(\cdot)$ denotes the minimum eigenvalue.
        \item[(ii)] $\max_{-\lfloor Th \rfloor\leq j \leq \lfloor Th \rfloor}\left\|\Delta_{t_j}\right\|=o_p(1)$.
        \item[(iii)]
        \begin{align}
            \max_{-\lfloor Th \rfloor\leq j \leq \lfloor Th \rfloor}\left\|\frac{1}{Th}\sum_{i=t_j-\lfloor Th \rfloor}^{t_j+\lfloor Th \rfloor}K\left(\frac{t_j-i}{Th}\right)x_i\varepsilon_i\right\|=o_p(1).
        \end{align}
    \end{itemize}
    \label{lemapp:uniform_order}
\end{lemapp}

\noindent\begin{proof}
    (i) Note that under Assumptions \ref{asm:kernel} and \ref{asm:min_eigen}, for any $\lambda\neq0$, we have
    \begin{align}
        \lambda'\left(\frac{1}{Th}\sum_{i=t_j-\lfloor Th \rfloor}^{t_j+\lfloor Th \rfloor}K\left(\frac{t_j-i}{Th}\right)E\left[x_ix_i'\right]\right)\lambda 
        &\geq \frac{1}{Th}\sum_{i=t_j-\lfloor Th \rfloor}^{t_j+\lfloor Th \rfloor}K\left(\frac{t_j-i}{Th}\right)\rho\left\|\lambda\right\|^2 \\
        &=\rho\left\|\lambda\right\|^2\left(1+O\left(1/Th\right)\right)
    \end{align}
    uniformly in $j$. This implies that there exists some constant $\epsilon\in(0,1)$ such that for $T$ sufficiently large, uniformly in $j$,
    \begin{align}
        \lambda_{\min}\left(\frac{1}{Th}\sum_{i=t_j-\lfloor Th \rfloor}^{t_j+\lfloor Th \rfloor}K\left(\frac{t_j-i}{Th}\right)E\left[x_ix_i'\right]\right) \geq \rho(1-\epsilon)>0.
    \end{align}
    The proof is completed by taking $C=\rho(1-\epsilon)>0$.

    (ii) It suffices to show
    \begin{align}
        \max_{-\lfloor Th \rfloor\leq j \leq \lfloor Th \rfloor}\left\|\Delta_{t_j}\right\|
        &\leq \max_{-\lfloor Th \rfloor\leq j \leq \lfloor Th \rfloor}\left\|\frac{1}{Th}\sum_{i=t-2\lfloor Th \rfloor}^{t+2\lfloor Th \rfloor}K\left(\frac{t_j-i}{Th}\right)\left(x_ix_i'-E\left[x_ix_i'\right]\right)\right\| \\
        &\quad + \max_{-\lfloor Th \rfloor\leq j \leq \lfloor Th \rfloor}\left\|\frac{1}{Th}\sum_{i=t-2\lfloor Th \rfloor}^{t_j-\lfloor Th \rfloor-1}K\left(\frac{t_j-i}{Th}\right)\left(x_ix_i'-E\left[x_ix_i'\right]\right)\right\| \\
        &\quad + \max_{-\lfloor Th \rfloor\leq j \leq \lfloor Th \rfloor}\left\|\frac{1}{Th}\sum_{i=t_j+\lfloor Th \rfloor+1}^{t+2\lfloor Th \rfloor}K\left(\frac{t_j-i}{Th}\right)\left(x_ix_i'-E\left[x_ix_i'\right]\right)\right\| \\
        &= o_p(1). \label{bound:delta}
    \end{align}

    Following the argument of Example 17.17 of \citetapp{davidsonStochasticLimitTheory1994}, under Assumption \ref{asm:dgp_2}, $\{K((t_j-i)/Th)(x_ix_i'-E[x_ix_i'])/Th\}$ is an $L_r$-bounded ($r>2$), zero-mean $L_2$-NED triangular array of size $-1$ on an $\alpha$-mixing sequence of size $-2r/(r-2)$. By Theorem 17.5 of \citetapp{davidsonStochasticLimitTheory1994}, this array is an $L_2$-mixingale of size $-1$ with constants $c_i\leq C\max\{\sup_iE[\|x_i\|^{2r}]^{1/r}, \sup_i d_i\}/Th = O(1/Th)$ uniformly in $i$. Therefore, Lemma 2 of \citetapp{hansenStrongLawsDependent1991} can be applied to obtain
    \begin{align}
        &E\left[\max_{-\lfloor Th \rfloor\leq j \leq \lfloor Th \rfloor}\left\|\frac{1}{Th}\sum_{i=t-2\lfloor Th \rfloor}^{t+2\lfloor Th \rfloor}K\left(\frac{t_j-i}{Th}\right)\left(x_ix_i'-E\left[x_ix_i'\right]\right)\right\|^2\right] \\
        &\qquad \leq CE\left[\max_{-\lfloor Th \rfloor\leq j \leq \lfloor Th \rfloor}\left\|\frac{1}{Th}\sum_{i=t-2\lfloor Th \rfloor}^{t+2\lfloor Th \rfloor}\left(x_ix_i'-E\left[x_ix_i'\right]\right)\right\|^2\right] = O\left(\frac{1}{Th}\right),
    \end{align}
    \begin{align}
        E\left[\max_{-\lfloor Th \rfloor\leq j \leq \lfloor Th \rfloor}\left\|\frac{1}{Th}\sum_{i=t-2\lfloor Th \rfloor}^{t_j-\lfloor Th \rfloor}K\left(\frac{t_j-i}{Th}\right)\left(x_ix_i'-E\left[x_ix_i'\right]\right)\right\|^2\right] = O\left(\frac{1}{Th}\right),
        \intertext{and}
        E\left[\max_{-\lfloor Th \rfloor\leq j \leq \lfloor Th \rfloor}\left\|\frac{1}{Th}\sum_{i=t_j+\lfloor Th \rfloor}^{t+2\lfloor Th \rfloor}K\left(\frac{t_j-i}{Th}\right)\left(x_ix_i'-E\left[x_ix_i'\right]\right)\right\|^2\right] = O\left(\frac{1}{Th}\right),
    \end{align}
    which, in conjunction with the Markov inequality, proves \eqref{bound:delta}.

    (iii) Following the same argument used to prove \eqref{bound:delta}, part (iii) follows since $\{x_i\varepsilon_i\}$ is a zero-mean process that shares the same NED properties with $\{x_ix_i'-E[x_ix_i']\}$.
\end{proof}

\begin{lemapp}
    Under Assumptions \ref{asm:kernel} and \ref{asm:dgp_2}-\ref{asm:uniform_negligible}, for each $t=\lfloor Tr \rfloor, \ r\in(0,1)$, we have
    \begin{align}
        \max_{-\lfloor Th \rfloor \leq j \leq \lfloor Th \rfloor}\left\|\hat{\beta}_{t_j}-\beta_{T,t_j}\right\| = o_p(1).
    \end{align}
    \label{lemapp:beta_hat_uniform}
\end{lemapp}

\noindent\begin{proof}
    From the decomposition given in \eqref{eqn:beta_hat_decompose}, we have
    \begin{align}
        &\max_{-\lfloor Th \rfloor \leq j \leq \lfloor Th \rfloor}\left\|\hat{\beta}_{t_j}-\beta_{T,t_j}\right\| \\
        &\leq \max_{-\lfloor Th \rfloor \leq j \leq \lfloor Th \rfloor}\left\|\Biggl(\frac{1}{Th}\sum_{i=t_j-\lfloor Th \rfloor}^{t_j+\lfloor Th \rfloor}K\left(\frac{t_j-i}{Th}\right)x_ix_i'\Biggr)^{-1}\right\| \\
        &\quad \times \left(\max_{-\lfloor Th \rfloor \leq j \leq \lfloor Th \rfloor}\left\|\frac{1}{Th}\sum_{i=t_j-\lfloor Th \rfloor}^{t_j+\lfloor Th \rfloor}K\left(\frac{t_j-i}{Th}\right) x_i\varepsilon_i\right\| \right. \\
        &\left. \hspace{2cm} + \max_{-\lfloor Th \rfloor \leq j \leq \lfloor Th \rfloor}\left\|\frac{1}{Th}\sum_{i=t_j-\lfloor Th \rfloor}^{t_j+\lfloor Th \rfloor}K\left(\frac{t_j-i}{Th}\right)x_ix_i'\left(\beta_{T,i}-\beta_{T,t_j}\right)\right\|\right). \label{bound:betahat_error}
    \end{align}
    The first term in the parentheses is $o_p(1)$ by Lemma \ref{lemapp:uniform_order}(iii), and the second term is also $o_p(1)$ by Assumption \ref{asm:uniform_negligible}. Therefore, it suffices to show
    \begin{align}
        \max_{-\lfloor Th \rfloor \leq j \leq \lfloor Th \rfloor}\left\|\Biggl(\frac{1}{Th}\sum_{i=t_j-\lfloor Th \rfloor}^{t_j+\lfloor Th \rfloor}K\left(\frac{t_j-i}{Th}\right)x_ix_i'\Biggr)^{-1}\right\|=O_p(1). \label{eqn:uniform_denominator}
    \end{align}
    Letting $\lambda_{\max}(\cdot)$ denote the maximum eigenvalue, and using Lemma \ref{lemapp:uniform_order} and the inequality $\lambda_{\max}(A)\leq \|A\| \leq \sqrt{p}\lambda_{\max}(A)$ for any $p\times p$ symmetric matrix $A\geq0$, we have
    \begin{align}
         &\max_{-\lfloor Th \rfloor \leq j \leq \lfloor Th \rfloor}\left\|\left(\frac{1}{Th}\sum_{i=t_j-\lfloor Th \rfloor}^{t_j+\lfloor Th \rfloor}K\left(\frac{t_j-i}{Th}\right)x_ix_i'\right)^{-1}\right\| \\
        &\leq \sqrt{p}\left\{\min_{-\lfloor Th \rfloor \leq j \leq \lfloor Th \rfloor}\lambda_{\min}\left(\frac{1}{Th}\sum_{i=t_j-\lfloor Th \rfloor}^{t_j+\lfloor Th \rfloor}K\left(\frac{t_j-i}{Th}\right)x_ix_i'\right)\right\}^{-1} \\
        &\leq\sqrt{p}\left\{\min_{-\lfloor Th \rfloor \leq j \leq \lfloor Th \rfloor}\lambda_{\min}\left(\frac{1}{Th}\sum_{i=t_j-\lfloor Th \rfloor}^{t_j+\lfloor Th \rfloor}K\left(\frac{t_j-i}{Th}\right)E[x_ix_i']\right) - \max_{-\lfloor Th \rfloor \leq j \leq \lfloor Th \rfloor}\lambda_{\max}^{1/2}\left(\Delta_{t_j}^2\right)\right\}^{-1} \\
        &\leq\sqrt{p}\left\{\min_{-\lfloor Th \rfloor \leq j \leq \lfloor Th \rfloor}\lambda_{\min}\left(\frac{1}{Th}\sum_{i=t_j-\lfloor Th \rfloor}^{t_j+\lfloor Th \rfloor}K\left(\frac{t_j-i}{Th}\right)E[x_ix_i']\right) - \max_{-\lfloor Th \rfloor \leq j \leq \lfloor Th \rfloor}\left\|\Delta_{t_j}\right\|\right\}^{-1} \\
        &\leq\frac{\sqrt{p}}{C-o_p(1)} = O_p(1).
    \end{align}
    This proves \eqref{eqn:uniform_denominator}.
\end{proof}

\noindent\begin{proof}[Proof of Theorem \ref{thm:variance_consistency}]
    Using $\hat{\varepsilon}_i = y_i-x_i'\hat{\beta}_i=\varepsilon_i-x_i'(\hat{\beta}_i-\beta_{T,i})$, we have
    \begin{align}
        \hat{\Sigma}(r) = B_{T,1} + B_{T,2} + B_{T,3},
    \end{align}
    where
    \begin{align}
        B_{T,1} &\coloneqq \frac{1}{Th}\sum_{i=1}^TK\left(\frac{\lfloor Tr \rfloor-i}{Th}\right)^2\varepsilon_i^2x_ix_i', \\
        B_{T,2} &\coloneqq - \frac{2}{Th}\sum_{i=1}^TK\left(\frac{\lfloor Tr \rfloor-i}{Th}\right)^2\left\{\varepsilon_ix_i'\left(\hat{\beta}_i-\beta_{T,i}\right)\right\}x_ix_i', 
        \intertext{and}
        B_{T,3} &\coloneqq \frac{1}{Th}\sum_{i=1}^TK\left(\frac{\lfloor Tr \rfloor-i}{Th}\right)^2\left(\hat{\beta}_i-\beta_{T,i}\right)'x_ix_i'\left(\hat{\beta}_i-\beta_{T,i}\right)x_ix_i'.
    \end{align}

    For $B_{T,2}$, by the fact that $K(\cdot)$ is bounded on compact support $[-1,1]$ under Assumption \ref{asm:kernel}, we obtain
    \begin{align}
        \left\|B_{T,2}\right\| \leq C \max_{t-\lfloor Th \rfloor \leq i \leq t+\lfloor Th \rfloor}\left\|\hat{\beta}_i-\beta_{T,i}\right\|\frac{2}{Th}\sum_{i=t-\lfloor Th \rfloor}^{t+\lfloor Th \rfloor}\left|\varepsilon_i\right|\left\|x_i\right\|^3.
    \end{align}
    An application of the H\"{o}lder inequality and Assumption \ref{asm:dgp_2}(a') yields
    \begin{align}
        \frac{2}{Th}\sum_{i=t-\lfloor Th \rfloor}^{t+\lfloor Th \rfloor}E\left[\left|\varepsilon_i\right|\left\|x_i\right\|^3\right] \leq \frac{2}{Th}\sum_{i=t-\lfloor Th \rfloor}^{t+\lfloor Th \rfloor}E\left[\left|\varepsilon_i\right|^4\right]^{1/4}E\left[\left\|x_i\right\|^4\right]^{3/4} = O(1).
    \end{align}
    This, in conjunction with the Markov inequality and Lemma \ref{lemapp:beta_hat_uniform}, shows $\left\|B_{T,2}\right\|=o_p(1)$. An analogous argument shows $\left\|B_{T,3}\right\|=o_p(1)$. Finally, decompose $B_{T,1}$ as
    \begin{align}
        B_{T,1} &= \frac{1}{Th}\sum_{i=1}^TK\left(\frac{\lfloor Tr \rfloor-i}{Th}\right)^2E\left[\varepsilon_i^2x_ix_i'\right]
        + \frac{1}{Th}\sum_{i=1}^TK\left(\frac{\lfloor Tr \rfloor-i}{Th}\right)^2\left(\varepsilon_i^2x_ix_i'-E\left[\varepsilon^2x_ix_i'\right]\right) \\
        &\eqqcolon B_{T,11} + B_{T,12}.
    \end{align}
    By Assumptions \ref{asm:dgp_2}(b) and (c), $B_{T,11}\to\Sigma(r)$. For $B_{T,12}$, note that $\{\varepsilon_ix_i\}$ is $L_2$-NED (of size $-1$) with respect to uniformly bounded constants under Assumption \ref{asm:dgp_2}, which is a direct consequence of Example 17.17 of \citetapp{davidsonStochasticLimitTheory1994}. It follows from Theorem 17.9 of \citetapp{davidsonStochasticLimitTheory1994} that $\{\varepsilon_i^2x_ix_i'\}$ is $L_1$-NED. Therefore, $\{K((\lfloor Tr \rfloor - i)/Th)^2(\varepsilon_i^2x_ix_i'-E[\varepsilon_i^2x_ix_i'])\}$ is an $L_{r'}$-bounded ($r'=r/2>1$), $L_1$-NED triangular array, and thus is a uniformly integrable $L_1$-mixingale \citepapp{andrewsLawsLargeNumbers1988}. Applying the law of large numbers of \citetapp{andrewsLawsLargeNumbers1988}, we deduce $B_{T,12}\stackrel{p}{\to}0$. Collecting above results gives
    \begin{align}
        \hat{\Sigma}(r) = B_{T,11} + o_p(1) \stackrel{p}{\to} \Sigma(r),
    \end{align}
    which completes the proof.
\end{proof}

Set $h_1=T^{\gamma_1}$ and $h_2=T^{\gamma_2}$ with $\gamma_2\leq\gamma_1$.
\begin{lemapp}
    Under Assumptions \ref{asm:kernel} and \ref{asm:dgp_2}, for $t=\lfloor Tr\rfloor, \ r\in(0,1)$, we have
    \begin{align}
        \max_{-\lfloor Th_2\rfloor\leq j \leq\lfloor Th_2\rfloor}\left\|\frac{1}{Th_1}\sum_{i=1}^TK\left(\frac{t+j-i}{Th_1}\right)x_i\varepsilon_i\right\| = \begin{cases}
            O_p\left(1/\sqrt{Th_2}\right) & \mathrm{if} \ \gamma_2=\gamma_1 \\
            o_p\left(1/\sqrt{Th_2}\right) & \mathrm{if} \ \gamma_2<\gamma_1
        \end{cases}.
    \end{align}
    \label{lemapp:product_uniform_order}
\end{lemapp}

\noindent\begin{proof}
    Using the same argument used to prove \eqref{bound:delta} and the Cauchy-Schwarz inequality, we obtain
    \begin{align}
        E&\left[\max_{-\lfloor Th_2\rfloor\leq j \leq\lfloor Th_2\rfloor}\left\|\frac{1}{Th_1}\sum_{i=1}^TK\left(\frac{t+j-i}{Th_1}\right)x_i\varepsilon_i\right\|\right] \\
        &\leq E\left[\max_{-\lfloor Th_1\rfloor\leq j \leq\lfloor Th_1\rfloor}\left\|\frac{1}{Th_1}\sum_{i=1}^TK\left(\frac{t+j-i}{Th_1}\right)x_i\varepsilon_i\right\|^2\right]^{1/2} \\
        &=O\left(1/\sqrt{Th_1}\right),
    \end{align}
    for $\gamma_2\leq\gamma_1$. The result follows from Markov's inequality. 
\end{proof}

\begin{lemapp}
    Under Assumptions \ref{asm:kernel}, \ref{asm:dgp_2}, \ref{asm:min_eigen}, and \ref{asm:uniform_negligible_2}, for $t=\lfloor Tr\rfloor, \ r\in(0,1)$, we have
    \begin{align}
        \sqrt{Th_2}\max_{-\lfloor Th_2\rfloor\leq i \leq \lfloor Th_2\rfloor}\left\|\hat{\beta}_{t+i}(\gamma_1)-\beta_{T,t+i}\right\| = \begin{cases}
            O_p(1) & \mathrm{if} \ \gamma_2=\gamma_1 \\
            o_p(1) & \mathrm{if} \ \gamma_2<\gamma_1
        \end{cases}.
    \end{align}
    \label{lemapp:beta_hat_unifrom_strong}
\end{lemapp}

\noindent\begin{proof}
    From \eqref{bound:betahat_error}, we have
    \begin{align}
        &\sqrt{Th_2}\max_{-\lfloor Th_2\rfloor\leq i \leq \lfloor Th_2\rfloor}\left\|\hat{\beta}_{t+i}(\gamma_1)-\beta_{T,t+i}\right\| \\
        &\leq \max_{-\lfloor Th_2 \rfloor\leq i \leq \lfloor Th_2\rfloor}\left\|\Biggl(\frac{1}{Th_1}\sum_{j=t_i-\lfloor Th_1 \rfloor}^{t_i+\lfloor Th_1 \rfloor}K\left(\frac{t_i-j}{Th_1}\right)x_jx_j'\Biggr)^{-1}\right\| \\
        &\quad \times \left(\sqrt{Th_2}\max_{-\lfloor Th_2 \rfloor \leq i \leq \lfloor Th_2 \rfloor}\left\|\frac{1}{Th_1}\sum_{j=t_i-\lfloor Th_1 \rfloor}^{t_i+\lfloor Th_1 \rfloor}K\left(\frac{t_i-j}{Th_1}\right) x_j\varepsilon_j\right\| \right. \\
        &\left. \hspace{2cm} + \sqrt{Th_2} \max_{-\lfloor Th_2 \rfloor \leq i \leq \lfloor Th_2 \rfloor}\left\|\frac{1}{Th_1}\sum_{j=t_i-\lfloor Th_1 \rfloor}^{t_i+\lfloor Th_1 \rfloor}K\left(\frac{t_i-j}{Th_1}\right)x_jx_j'\left(\beta_{T,j}-\beta_{T,t_i}\right)\right\|\right).
    \end{align}
    The desired result now follows from \eqref{eqn:uniform_denominator}, Lemma \ref{lemapp:product_uniform_order}, and Assumption \ref{asm:uniform_negligible_2}.
\end{proof}

In the proof of Theorem \ref{thm:bootstrap} below, we write, for any bootstrap statistic $S_T^*$ and any distribution $D$, $S_T^*\stackrel{d_{p^*}}{\to}D$, in probability, when convergence in distribution under the bootstrap probability measure occurs on a sequence of events with probability approaching one. We also let $E^*[\cdot]$ and $V^*[\cdot]$ denote the expectation and variance under the bootstrap measure, respectively.

\noindent\begin{proof}[Proof of Theorem \ref{thm:bootstrap}]
    Decompose $\hat{\beta}_t^*(\gamma_1,\gamma_2)-\hat{\beta}_t(\gamma_1)$ as
    \begin{align}
        &\hat{\beta}_t^*(\gamma_1,\gamma_2)-\hat{\beta}_t(\gamma_1) \\
        &=\left(\sum_{i=t-\lfloor Th_2\rfloor}^{t+\lfloor Th_2\rfloor}K\left(\frac{t-i}{Th_2}\right)x_ix_i'\right)^{-1} \sum_{i=t-\lfloor Th_2\rfloor}^{t+\lfloor Th_2\rfloor}K\left(\frac{t-i}{Th_2}\right)x_ix_i'\left(\hat{\beta}_i(\gamma_1)-\hat{\beta}_t(\gamma_1)\right) \\
        & \qquad +\left(\sum_{i=t-\lfloor Th_2\rfloor}^{t+\lfloor Th_2\rfloor}K\left(\frac{t-i}{Th_2}\right)x_ix_i'\right)^{-1} \sum_{i=t-\lfloor Th_2\rfloor}^{t+\lfloor Th_2\rfloor}K\left(\frac{t-i}{Th_2}\right)x_i\varepsilon_i^*(\gamma_1) \\
        &\eqqcolon C_{T,1} + C_{T,2}. \label{eqn:decomposition_bootstrap}
    \end{align}

    We first show
    \begin{align}
        C_{T,1} = \begin{cases}
            O_p\left(1/\sqrt{Th_2}\right) & \mathrm{if} \ \gamma_2=\gamma_1 \\
            o_p\left(1/\sqrt{Th_2}\right) & \mathrm{if} \ \gamma_2<\gamma_1
        \end{cases}.
        \label{eqn:order_ct1}
    \end{align}
    A straightforward calculation shows
    \begin{align}
        &\left\|\sum_{i=t-\lfloor Th_2\rfloor}^{t+\lfloor Th_2\rfloor}K\left(\frac{t-i}{Th_2}\right)x_ix_i'\left(\hat{\beta}_i(\gamma_1)-\hat{\beta}_t(\gamma_1)\right)\right\| \\
        &\leq \left(\max_{t-\lfloor Th_2\rfloor\leq i\leq t+\lfloor Th_2\rfloor}\left\|\beta_{T,i}-\beta_{T,t}\right\| + 2\max_{t-\lfloor Th_2\rfloor\leq i\leq t+\lfloor Th_2\rfloor}\left\|\hat{\beta}_{i}(\gamma_1)-\beta_{T,i}\right\|\right) \\
        &\qquad \times\frac{C}{Th_2}\sum_{i=t-\lfloor Th_2\rfloor}^{t+\lfloor Th_2\rfloor}\left\|x_i\right\|^2 \\
        &=\begin{cases}
            O_p\left(h_2^{\alpha}\right) + O_p\left(1/\sqrt{Th_2}\right) & \mathrm{if} \ \gamma_2=\gamma_1 \\
            O_p\left(h_2^{\alpha}\right) + o_p\left(1/\sqrt{Th_2}\right) & \mathrm{if} \ \gamma_2<\gamma_1
        \end{cases},
    \end{align}
    where the last probability order follows from Definition \ref{def:tvp_alpha_class}, Lemma \ref{lemapp:beta_hat_unifrom_strong}, and Assumption \ref{asm:dgp_2}. Since $\gamma_2\leq\gamma_1<-(2\alpha+1)^{-1}$ by assumption, we have $\sqrt{Th_2}h_2^{\alpha}=O(T^{(1+(2\alpha+1)\gamma_2)/2})=o(1)$, which implies $O_p(h_2^\alpha)=o_p(1/\sqrt{Th_2})$. In view of the fact that $(Th_2)^{-1}\sum_{i=1}^TK((t-i)/Th_2)x_ix_i'\stackrel{p}{\to}\Omega(r)>0$ by Lemma \ref{lemapp:variance_plim}, this proves \eqref{eqn:order_ct1}.

    Next, we consider $C_{T,2}$, whose numerator can be decomposed as
    \begin{align}
        \frac{1}{Th_2}\sum_{i=t-\lfloor Th_2\rfloor}^{t+\lfloor Th_2\rfloor}K\left(\frac{t-i}{Th_2}\right)x_i\varepsilon_i^*(\gamma_1)=\frac{1}{Th_2}\sum_{i=t-\lfloor Th_2\rfloor}^{t+\lfloor Th_2\rfloor}K\left(\frac{t-i}{Th_2}\right)x_i\varepsilon_i\eta_i - r_{T,t},
    \end{align}
    where $r_{T,t}\coloneqq (Th_2)^{-1}\sum_{i=t-\lfloor Th_2\rfloor}^{t+\lfloor Th_2\rfloor}K((t-i)/(Th_2))x_ix_i'\eta_i(\hat{\beta}_i(\gamma_1)-\beta_{T,i})$. $r_{T,t}$ is bounded by
    \begin{align}
        \left\|r_{T,t}\right\|\leq \max_{t-\lfloor Th_2\rfloor\leq i\leq t+\lfloor Th_2\rfloor}\left\|\hat{\beta}_i(\gamma_1)-\beta_{T,i}\right\|\times\frac{C}{Th_2}\sum_{i=t-\lfloor Th_2\rfloor}^{t+\lfloor Th_2\rfloor}\|x_i\|^2|\eta_i|.
    \end{align}
    The second term satisfies
    \begin{align}
        E^*\left[\frac{C}{Th_2}\sum_{i=t-\lfloor Th_2\rfloor}^{t+\lfloor Th_2\rfloor}\|x_i\|^2|\eta_i|\right] 
        &= \frac{C}{Th_2}\sum_{i=t-\lfloor Th_2\rfloor}^{t+\lfloor Th_2\rfloor}\|x_i\|^2E^*[|\eta_i|]=O_p(1).        
    \end{align}
    This implies that for any $\epsilon>0$, there exists a (large) $T_1$ such that $(C/Th_2)\sum_{i=t-\lfloor Th_2\rfloor}^{t+\lfloor Th_2\rfloor}\|x_i\|^2|\eta_i| = O_{p^*}(1)$ with probability at least $1-\epsilon$ for all $T\geq T_1$. This, in conjunction with Lemma \ref{lemapp:beta_hat_unifrom_strong}, yields
    \begin{align}
        r_{T,t} = \begin{cases}
            O_{p^*}\left(1/\sqrt{Th_2}\right) & \mathrm{if} \ \gamma_2=\gamma_1 \\
            o_{p^*}\left(1/\sqrt{Th_2}\right) & \mathrm{if} \ \gamma_2<\gamma_1
        \end{cases},
    \end{align}
    with arbitrarily high probability for $T$ sufficiently large. Consequently, we obtain
    \begin{align}
        C_{T,2} = \left(\frac{1}{Th_2}\sum_{i=t-\lfloor Th_2\rfloor}^{t+\lfloor Th_2\rfloor}K\left(\frac{t-i}{Th_2}\right)x_ix_i'\right)^{-1} \frac{1}{Th_2}\sum_{i=t-\lfloor Th_2\rfloor}^{t+\lfloor Th_2\rfloor}K\left(\frac{t-i}{Th_2}\right)x_i\varepsilon_i\eta_i + r_{T,t}^*, \\
        \label{eqn:ct2_decomposition}
    \end{align}
    where $r_{T,t}^* \coloneqq ((1/Th_2)\sum_{i=t-\lfloor Th_2\rfloor}^{t+\lfloor Th_2\rfloor}K((t-i)/Th_2)x_ix_i')^{-1}\times r_{T,t}$ has the same asymptotic order as $r_{T,t}$.

    Substituting \eqref{eqn:order_ct1} and \eqref{eqn:ct2_decomposition} into \eqref{eqn:decomposition_bootstrap} gives
    \begin{align}
        \sqrt{Th_2}&\left(\hat{\beta}_t^*(\gamma_1,\gamma_2)-\hat{\beta}_t(\gamma_1)-R_{T,t}^*\right) \\
        &= \left(\frac{1}{Th_2}\sum_{i=t-\lfloor Th_2\rfloor}^{t+\lfloor Th_2\rfloor}K\left(\frac{t-i}{Th_2}\right)x_ix_i'\right)^{-1} \frac{1}{\sqrt{Th_2}}\sum_{i=t-\lfloor Th_2\rfloor}^{t+\lfloor Th_2\rfloor}K\left(\frac{t-i}{Th_2}\right)x_i\varepsilon_i\eta_i,
    \end{align}
    where $R_{T,t}^*$ satisfies the condition stated in Theorem \ref{thm:bootstrap}.
    
    Now, if we show
    \begin{align}
        \frac{1}{\sqrt{Th_2}}\sum_{i=t-\lfloor Th_2\rfloor}^{t+\lfloor Th_2\rfloor}K\left(\frac{t-i}{Th_2}\right)x_i\varepsilon_i\eta_i \stackrel{d_{p^*}}{\to} N\left(0,\Sigma(r)\right), \ \mathrm{in}\ \mathrm{probability},
        \label{wc:bootstrap}
    \end{align}
    then the proof is completed by the CMT and Polya's theorem, noting that the normal distribution is everywhere continuous. Take any unit vector $\lambda\in \mathbb{R}^p$, and let $\zeta_{T,i}^*\coloneqq (Th_2)^{-1/2}\lambda'K((t-i)/Th_2)x_i\varepsilon_i\eta_i$. Note that $E^*[\sum_{i=1}^T\zeta_{T,i}^*]=0$, and $V^*[\sum_{i=1}^T\zeta_{T,i}^*]=\lambda'(Th_2)^{-1}\sum_{i=t-\lfloor Th_2\rfloor}^{i=t+\lfloor Th_2\rfloor}K((t-i)/Th_2)^2x_ix_i'\varepsilon_i^2\lambda\stackrel{p}{\to}\lambda'\Sigma(r)\lambda>0$, as shown in the proof of Theorem \ref{thm:variance_consistency}. To show that $\sum_{i=1}^T\zeta_{T,i}^*\stackrel{d_{p^*}}{\to}N(0,\lambda'\Sigma(r)\lambda)$, in probability, we check Liapunov's condition \citep[e.g., Theorem 23.11 of][]{davidsonStochasticLimitTheory1994}. For $\delta>1$, we have
    \begin{align}
        \sum_{i=t-\lfloor Th_2\rfloor}^{t+\lfloor Th_2 \rfloor}E^*\left[\left|\zeta_{T,i}^*\right|^{2\delta}\right] 
        &= \frac{1}{(Th_2)^\delta}\sum_{i=t-\lfloor Th_2\rfloor}^{t+\lfloor Th_2\rfloor}K\left(\frac{t-i}{Th_2}\right)^{2\delta}|\lambda x_i\varepsilon_i|^{2\delta}E^*\left[|\eta_i|^{2\delta}\right] \\
        &=O_p\left((Th_2)^{1-\delta}\right) = o_p(1),
    \end{align}
    since $E\left[\left|\sum_{i=t-\lfloor Th_2\rfloor}^{t+\lfloor Th_2\rfloor}K((t-i)/Th_2)^{2\delta}|\lambda x_i\varepsilon_i|^{2\delta}E^*[|\eta_i|^{2\delta}]\right|\right]\leq CTh_2\sup_iE[\|x_i\|^{4\delta}]^{1/2}E[|\varepsilon_i|^{4\delta}]^{1/2}=O(Th_2)$ under Assumption \ref{asm:dgp_2}. Therefore, \eqref{wc:bootstrap} follows from Liapunov's CLT and the Cramer-Wold device. This completes the proof.
\end{proof}

\titleformat*{\section}{\Large\bfseries\centering}
\section*{Appendix B: MSE-Minimizing Bandwidth in the Case of Rescaled Random Walk Coefficients}

\setcounter{equation}{0}
\renewcommand{\theequation}{B.\arabic{equation}}
\setcounter{prop}{0}
\renewcommand{\theprop}{B.\arabic{prop}}
\setcounter{asm}{0}
\renewcommand{\theasm}{B.\arabic{asm}}
\setcounter{section}{0}
\renewcommand{\thesection}{B.\arabic{section}}
\setcounter{figure}{0}
\renewcommand{\thefigure}{B.\arabic{figure}}
\setcounter{table}{0}
\renewcommand{\thetable}{B.\arabic{table}}

In this appendix, we show that, in the case of random-walk coefficients, the bandwidth that minimizes the MSE of the kernel-based estimator is proportional to $T^{-1/2}$. In what follows, we will assume that $Th$ is an integer for simplicity.

\titleformat*{\section}{\large\bfseries}
\section{A simple case}

To gain some insight, we begin with the following local-level model:
\begin{align}
    y_t = \beta_{T,t} + \varepsilon_t, \label{model:rw_simple}    
\end{align}
where $\beta_{T,t} = T^{-1/2}\sum_{i=1}^tu_i$.

\begin{asm}
    $(\varepsilon_t,u_t)$ is an i.i.d. sequence with mean zero and variance $\Sigma=\mathrm{diag}(\sigma_\varepsilon^2,\sigma_u^2)$. Moreover, $\varepsilon_t$ and $u_t$ are independent.
    \label{asm:rw_mse_simple}
\end{asm}

We estimate $\beta_{T,t}$ using $\hat{\beta}_t$ with $K(\cdot)$ being the uniform kernel, that is, $\hat{\beta}_t = (2Th+1)^{-1}\sum_{i=t-Th}^{t+Th}y_i$. Let $\mathrm{MSE}(h) \coloneqq E[(\hat{\beta_t}-\beta_{T,t})^2]$ denote the MSE of $\hat{\beta}_t$ as a function of bandwidth parameter $h$. 

From model \eqref{model:rw_simple}, $\hat{\beta}_t - \beta_{T,t}$ admits the following decomposition:
\begin{align}
    \hat{\beta}_t - \beta_{T,t} &= \frac{1}{2Th+1}\sum_{i=t-Th}^{t+Th}(\beta_{T,i} - \beta_{T,t}) + \frac{1}{2Th + 1}\sum_{i=t-Th}^{t+Th}\varepsilon_i \\
    &= -\frac{1}{2Th+1}\sum_{i=t-Th}^{t-1}\Bigl(\frac{1}{\sqrt{T}}\sum_{k=i+1}^tu_k\Bigr) + \frac{1}{2Th+1}\sum_{i=t+1}^{t+Th}\Bigl(\frac{1}{\sqrt{T}}\sum_{k=t+1}^iu_k\Bigr) + \frac{1}{2Th + 1}\sum_{i=t-Th}^{t+Th}\varepsilon_i.
\end{align}

Given that $(\varepsilon_t,u_t)$ and $(\varepsilon_s,u_s)$ ($t\neq s$) are independent, and that $\varepsilon_t$ and $u_t$ are independent, we have
\begin{align}
    \mathrm{MSE}(h) &= \Bigl(\frac{1}{2Th+1}\Bigr)^2\biggl\{E\Bigl[\bigl(\sum_{i=t-Th}^{t-1}\frac{1}{\sqrt{T}}\sum_{k=i+1}^tu_k\bigr)^2\Bigr] + E\Bigl[\bigl(\sum_{i=t+1}^{t+Th}\frac{1}{\sqrt{T}}\sum_{k=t+1}^iu_k\bigr)^2\Bigr] + (2Th+1)\sigma_\varepsilon^2\biggr\} \\
    &= \Bigl(\frac{1}{2Th+1}\Bigr)^2\biggl\{\frac{1}{T}E\Bigl[\bigl(\sum_{i=1}^{Th}(Th-i+1)u_{t-i+1}\bigr)^2\Bigr] + \frac{1}{T}E\Bigl[\bigl(\sum_{i=1}^{Th}(Th-i+1)u_{t+i}\bigr)^2\Bigr]\biggr\} + \frac{\sigma_\varepsilon^2}{2Th+1} \\
    &= \frac{2\sigma_u^2}{(2Th+1)^2T}\frac{Th(Th+1)(2Th+1)}{6} + \frac{\sigma_\varepsilon^2}{2Th+1} \\
    &= \frac{\sigma_u^2h(1+o(1))}{6(1+o(1))} + \frac{\sigma_\varepsilon^2}{2Th(1+o(1))}.
\end{align}
Ignoring the $o(1)$ terms, the MSE of $\hat{\beta}_t$ is asymptotically
\begin{align}
    \mathrm{MSE}(h) = \frac{\sigma_u^2}{6}h + \frac{\sigma_\varepsilon^2}{2T}h^{-1}. \label{MSE:simple}
\end{align}
Letting $h_{\mathrm{min}}$ denote the minimizer of \eqref{MSE:simple}, it can be easily shown that
\begin{align}
    h_{\mathrm{min}} = \Biggl(\frac{3\sigma_\varepsilon^2}{\sigma_u^2}\Biggr)^{1/2}T^{-1/2}.
\end{align}
Therefore, the MSE-minimizing bandwidth is proportional to $T^{-1/2}$.

\titleformat*{\section}{\large\bfseries}
\section{A general case}

The argument above can be extended to the multiple regression. Suppose we are interested in the following model: $y_t= x_t'\beta_{T,t} + \varepsilon_t$, where $\beta_{T,t} = T^{-1/2}\sum_{i=1}^tu_i$ is a $p$-dimensional rescaled random walk driven by $u_t=(u_{t,1}, \ldots, u_{t,p})'$. We impose the following assumption.
\begin{asm}
    \begin{itemize}
        \item[(a)] $\{x_t\}_t$ is a $p$-dimensional stationary sequence with $E[x_1x_1'] >0$.

        \item[(b)] $\{(\varepsilon_t,u_t')\}_t$ is a $(p+1)$-dimensional i.i.d. sequence that is independent of $\{x_t\}_t$ and has mean zero and variance $\mathrm{diag}(\sigma_\varepsilon^2, \sigma_u^2I_p)$. Moreover, $\varepsilon_t$ and $u_t$ are independent.

        \item[(c)] There exist nonrandom matrices $\Omega>0$, $\Lambda$, $\bar{\Lambda}$, and $\Xi$ such that $\Lambda \bar{\Lambda} + \bar{\Lambda}\Lambda - 2\Xi > 0$, $(2Th)^{-1}\sum_{i=t-Th}^{t+Th}x_ix_i' \stackrel{p}{\to}\Omega$, $(Th)^{-1}\sum_{i=t-Th}^{t-1}x_ix_i' \stackrel{p}{\to}\Lambda$, $(Th)^{-1}\sum_{i=1}^{Th}\Bigl(\frac{i}{Th}\Bigr)x_{t-i}x_{t-i}' \stackrel{p}{\to}\bar{\Lambda}$, and $(Th)^{-1}\sum_{i=t-Th}^{t-1}x_ix_i'(Th)^{-1}\sum_{j=t-Th}^i\frac{i-j}{Th}x_jx_j' \stackrel{p}{\to} \Xi$ as $T\to\infty$.

        \item[(d)] Matrices $\{(2Th)^{-1}\sum_{i=t-Th}^{t+Th}x_ix_i'\}^{-1}$, $(Th)^{-1}\sum_{i=t-Th}^{t-1}x_ix_i', \ (Th)^{-1}\sum_{i=1}^{Th}\Bigl(\frac{i}{Th}\Bigr)x_{t-i}x_{t-i}'$, $(Th)^{-1}\sum_{i=t-Th}^{t-1}x_ix_i'(Th)^{-1}\sum_{j=t-Th}^i\frac{i-j}{Th}x_jx_j'$, and their products are all uniformly integrable.
    \end{itemize}
    \label{asm:rw_mse_general}
\end{asm}

Assumptions \ref{asm:rw_mse_general}(a)-(b) extend Assumption \ref{asm:rw_mse_simple} to the case of the multiple regression. Assumption \ref{asm:rw_mse_general}(c) will hold if $x_tx_t' - E[x_tx_t']$ satisfies the condition of the law of large numbers. In this case, we will have $\Omega=E[x_1x_1']$, $\Lambda = E[x_1x_1']$, $\bar{\Lambda}=E[x_1x_1']/2$, and $\Xi=E[x_1x_1']^2/6$ under Assumption \ref{asm:rw_mse_general}(a). Assumption \ref{asm:rw_mse_general} (d) holds if all the matrices mentioned are uniformly bounded.

The estimator of $\beta_{T,t}$ is $\hat{\beta}_t = (\sum_{i=t-Th}^{t+Th}x_ix_i')^{-1}\sum_{i=t-Th}^{t+Th}x_iy_i$.

\begin{prop}
    Under Assumption \ref{asm:rw_mse_general}, we have
    \begin{align}
        \mathrm{MSE}(h) = \frac{\sigma_u^2h}{4}\mathrm{tr}&[\Omega^{-1}(\Lambda\bar{\Lambda} + \bar{\Lambda}\Lambda - 2\Xi)\Omega^{-1}](1+o(1)) + \frac{\sigma_\varepsilon^2}{2Th}\mathrm{tr}[\Omega^{-1}](1+o(1)).
    \end{align}
\end{prop}
Checking the first and second order conditions, one can easily verify that the MSE-minimizing $h$ is proportional to $T^{-1/2}$.

\noindent \begin{proof}
    Note that $ \hat{\beta}_t - \beta_{T,t} = \Bigl(\sum_{i=t-Th}^{t+Th}x_ix_i'\Bigr)^{-1}\Biggl\{\sum_{i=t-Th}^{t+Th}x_ix_i'(\beta_{T,i}-\beta_{T,t}) + \sum_{i=t-Th}^{t+Th}x_i\varepsilon_i\Biggr\}$. The conditional MSE given $X_T\coloneqq\{x_t\}_{t=1}^T$ is
    \begin{align}
        E&[\|\hat{\beta}_t - \beta_{T,t}\|^2|X_T]= E[\mathrm{tr}[(\hat{\beta}_t-\beta_{T,t})(\hat{\beta}_t-\beta_{T,t})']|X_T] \\
        &= \mathrm{tr}\Biggl[\Bigl(\sum_{i=t-Th}^{t+Th}x_ix_i'\Bigr)^{-1}E\biggl[\Bigl\{-\sum_{i=t-Th}^{t-1}x_ix_i'\frac{1}{\sqrt{T}}\sum_{k=i+1}^tu_k + \sum_{i=t+1}^{t+Th}x_ix_i'\frac{1}{\sqrt{T}}\sum_{k=t+1}^iu_k + \sum_{i=t-Th}^{t+Th}x_i\varepsilon_i\Bigr\} \\
        &\quad \times\Bigl\{-\sum_{i=t-Th}^{t-1}x_ix_i'\frac{1}{\sqrt{T}}\sum_{k=i+1}^tu_k + \sum_{i=t+1}^{t+Th}x_ix_i'\frac{1}{\sqrt{T}}\sum_{k=t+1}^iu_k
        + \sum_{i=t-Th}^{t+Th}x_i\varepsilon_i\Bigr\}'|X_T\biggr]\Bigl(\sum_{i=t-Th}^{t+Th}x_ix_i'\Bigr)^{-1} \Biggr] \\
        &= \mathrm{tr}\Biggl[\Bigl(\sum_{i=t-Th}^{t+Th}x_ix_i'\Bigr)^{-1}E\biggl[\frac{1}{T}\Bigl(\sum_{i=t-Th}^{t-1}x_ix_i'\sum_{k=i+1}^tu_k\Bigr)\Bigl(\sum_{i=t-Th}^{t-1}x_ix_i'\sum_{k=i+1}^tu_k\Bigr)' \\
        &\qquad + \frac{1}{T}\Bigl(\sum_{i=t+1}^{t + th}x_ix_i'\sum_{k=t+1}^iu_k\Bigr)\Bigl(\sum_{i=t+1}^{t + th}x_ix_i'\sum_{k=t+1}^iu_k\Bigr)' + \Bigl(\sum_{i=t-Th}^{t+Th}x_i\varepsilon_i\Bigr)\Bigl(\sum_{i=t-Th}^{t+Th}x_i\varepsilon_i\Bigr)'|X_T\biggr]\Bigl(\sum_{i=t-Th}^{t+Th}x_ix_i'\Bigr)^{-1}\Biggr], \label{eqn:mse_expand}
    \end{align}
    where the last equality follows from Assumption \ref{asm:rw_mse_general}(b).

    Consider each of the three terms in the conditional expectation in \eqref{eqn:mse_expand}.
    \begin{align}
        E\biggl[&\frac{1}{T}\Bigl(\sum_{i=t-Th}^{t-1}x_ix_i'\sum_{k=i+1}^tu_k\Bigr)\Bigl(\sum_{i=t-Th}^{t-1}x_ix_i'\sum_{k=i+1}^tu_k\Bigr)'|X_T\biggr] \\
        &=\frac{1}{T}\sum_{i=t-Th}^{t-1}\sum_{j=t-Th}^{t-1}x_ix_i'E\biggl[\sum_{k=i+1}^tu_k\sum_{l=j+1}^tu_l'\biggr]x_jx_j' \\
        &=\frac{1}{T}\sum_{i=t-Th}^{t-1}\sum_{j=t-Th}^{i}x_ix_i'\sum_{k=i+1}^tE[u_ku_k']x_jx_j' + \frac{1}{T}\sum_{i=t-Th}^{t-1}\sum_{j=i+1}^{t-1}x_ix_i'\sum_{k=j+1}^tE[u_ku_k']x_jx_j' \\
        &=\frac{\sigma_u^2}{T}\sum_{i=t-Th}^{t-1}x_ix_i'\biggl(\sum_{j=t-Th}^i(t-j+j-i)x_jx_j'+\sum_{j=i+1}^{t-1}(t-j)x_jx_j'\biggr) \\
        &=\frac{\sigma_u^2}{T}\biggl(\sum_{i=t-Th}^{t-1}x_ix_i'\sum_{j=t-Th}^{t-1}(t-j)x_jx_j' - \sum_{i=t-Th}^{t-1}x_ix_i'\sum_{j=t-Th}^i(i-j)x_jx_j'\biggr), \label{conditional_mse_1}
    \end{align}
    where we used the independence between $\{u_t\}$ and $\{x_t\}$, the serial independence of $\{u_t\}$, and $E[u_tu_t']=\sigma_u^2I_p$. Similarly, the second term becomes
    \begin{align}
        E\biggl[&\frac{1}{T}\Bigl(\sum_{i=t+1}^{t+Th}x_ix_i'\sum_{k=t+1}^iu_k\Bigr)\Bigl(\sum_{i=t+1}^{t+Th}x_ix_i'\sum_{k=t+1}^iu_k\Bigr)'|X_T\biggr]=\frac{1}{T}\sum_{i=t+1}^{t+Th}\sum_{j=t+1}^{t+Th}x_ix_i'E\biggl[\sum_{k=t+1}^iu_k\sum_{l=t+1}^ju_l'\biggr]x_jx_j' \\
        &=\frac{1}{T}\sum_{i=t+1}^{t+Th}\biggl(x_ix_i'\sum_{j=t+1}^{i}\sum_{k=t+1}^jE[u_ku_k']x_jx_j' + x_ix_i'\sum_{j=i+1}^{t+Th}\sum_{k=t+1}^iE[u_ku_k']x_jx_j'\biggr) \\
        &=\frac{\sigma_u^2}{T}\sum_{i=t+1}^{t+Th}x_ix_i'\biggl(\sum_{j=t+1}^i(j-i+i-t)x_jx_j'+\sum_{j=i+1}^{t+Th}(i-t)x_jx_j'\biggr) \\
        &=\frac{\sigma_u^2}{T}\biggl(\sum_{i=t+1}^{t+Th}(i-t)x_ix_i'\sum_{j=t+1}^{t+Th}x_jx_j' - \sum_{i=t+1}^{t+Th}x_ix_i'\sum_{j=t+1}^i(i-j)x_jx_j'\biggr). \label{conditional_mse_2}
    \end{align}
    The last term in the conditional expectation in \eqref{eqn:mse_expand} is
    \begin{align}
        E\biggl[\Bigl(\sum_{i=t-Th}^{t+Th}x_i\varepsilon_i\Bigr)\Bigl(\sum_{i=t-Th}^{t+Th}x_i\varepsilon_i\Bigr)'|X_T\biggr] = \sum_{i=t-Th}^{t+Th}\sum_{j=t-Th}^{t+Th}x_ix_j'E[\varepsilon_i\varepsilon_j]=\sigma_\varepsilon^2\sum_{i=t-Th}^{t+Th}x_ix_i'. \label{conditional_mse_3}
    \end{align}
    Substituting \eqref{conditional_mse_1}, \eqref{conditional_mse_2}, and \eqref{conditional_mse_3} into \eqref{eqn:mse_expand} yields
    \begin{align}
        E[&\|\hat{\beta}_t-\beta_{T,t}\|^2|X_T] \\
        &=\mathrm{tr}\biggl[\biggl(\sum_{i=t-Th}^{t+Th}x_ix_i'\biggr)^{-1}\biggl\{\frac{\sigma_u^2}{T}\biggl(\sum_{i=t-Th}^{t-1}x_ix_i'\sum_{j=t-Th}^{t-1}(t-j)x_jx_j' - \sum_{i=t-Th}^{t-1}x_ix_i'\sum_{j=t-Th}^i(i-j)x_jx_j' \\
        & \qquad + \sum_{i=t+1}^{t+Th}(i-t)x_ix_i'\sum_{j=t+1}^{t+Th}x_jx_j' - \sum_{i=t+1}^{t+Th}x_ix_i'\sum_{j=t+1}^i(i-j)x_jx_j'\biggr) + \sigma_\varepsilon^2\sum_{i=t-Th}^{t+Th}x_ix_i'\biggr\}\biggl(\sum_{i=t-Th}^{t+Th}x_ix_i'\biggr)^{-1}\biggr] \\
        &= \mathrm{tr}\biggl[\biggl(\frac{1}{Th}\sum_{i=t-Th}^{t+Th}x_ix_i'\biggr)^{-1}\biggl\{\sigma_u^2h\biggl(\frac{1}{Th}\sum_{i=t-Th}^{t-1}x_ix_i'\frac{1}{Th}\sum_{j=t-Th}^{t-1}\frac{t-j}{Th}x_jx_j' \\
        & \qquad- \frac{1}{Th}\sum_{i=t-Th}^{t-1}x_ix_i'\frac{1}{Th}\sum_{j=t-Th}^i\frac{i-j}{Th}x_jx_j' + \frac{1}{Th}\sum_{i=t+1}^{t+Th}\frac{i-t}{Th}x_ix_i'\frac{1}{Th}\sum_{j=t+1}^{t+Th}x_jx_j' \\
        & \qquad- \frac{1}{Th}\sum_{i=t+1}^{t+Th}x_ix_i'\frac{1}{Th}\sum_{j=t+1}^i\frac{i-j}{Th}x_jx_j'\biggr) + \frac{\sigma_\varepsilon^2}{Th}\frac{1}{Th}\sum_{i=t-Th}^{t+Th}x_ix_i'\biggr\}\biggl(\frac{1}{Th}\sum_{i=t-Th}^{t+Th}x_ix_i'\biggr)^{-1}\biggr] \\
        &= \frac{\sigma_u^2h}{4}\mathrm{tr}\Bigl[\Omega^{-1}\Bigl(\Lambda\bar{\Lambda} + \bar{\Lambda}\Lambda - 2\Xi\Bigr)\Omega^{-1}\Bigr](1+o_p(1)) +\frac{\sigma_\varepsilon^2}{2Th}\mathrm{tr}\Bigl[\Omega^{-1}\Bigr](1+o_p(1)).
    \end{align}
    Note that we used the stationarity of $x_t$ to derive the final expression. Therefore, the MSE of $\hat{\beta}_t$ satisfies
    \begin{align}
        \mathrm{MSE}(h) &= E[E[\|\hat{\beta}_t-\beta_{T,t}\|^2|X_T]] \\
        &=\frac{\sigma_u^2h}{4}\mathrm{tr}\Bigl[\Omega^{-1}\Bigl(\Lambda\bar{\Lambda} + \bar{\Lambda}\Lambda - 2\Xi\Bigr)\Omega^{-1}\Bigr](1+o(1)) +\frac{\sigma_\varepsilon^2}{2Th}\mathrm{tr}\Bigl[\Omega^{-1}\Bigr](1+o(1)),
    \end{align}
    where we interchanged the order of expectation and plim operator in view of Assumption \ref{asm:rw_mse_general}(d).
\end{proof}

\titleformat*{\section}{\Large\bfseries\centering}
\section*{Appendix C: Sufficient Conditions for Assumption \ref{asm:uniform_negligible}}

\setcounter{equation}{0}
\renewcommand{\theequation}{C.\arabic{equation}}
\setcounter{prop}{0}
\renewcommand{\theprop}{C.\arabic{prop}}
\setcounter{asm}{0}
\renewcommand{\theasm}{C.\arabic{asm}}
\setcounter{section}{0}
\renewcommand{\thesection}{C.\arabic{section}}
\setcounter{figure}{0}
\renewcommand{\thefigure}{C.\arabic{figure}}
\setcounter{table}{0}
\renewcommand{\thetable}{C.\arabic{table}}

\titleformat*{\section}{\large\bfseries}
\section{H\"{o}lder condition}

Under Assumptions \ref{asm:kernel} and \ref{asm:dgp}, the following condition is sufficient for Assumption \ref{asm:uniform_negligible} to hold.
\renewcommand*{\thecon}{H}
\begin{con}
    There exist some constants $C>0$ and $\alpha>0$ such that for all $i,j=1,\ldots,T$,
    \begin{align}
        \left\|\beta_{T,i}-\beta_{T,j}\right\| \leq C\left(\frac{\left|i-j\right|}{T}\right)^{\alpha}.
    \end{align}
    \label{con:holder}
\end{con}

Condition \ref{con:holder} is essentially the H\"{o}lder condition, and so it accommodates time-varying parameters $\beta_{T,t}=\beta(t/T)$ with $\beta(\cdot)$ continuously differentiable on $[0,1]$. Moreover, it accommodates models where $\beta_{T,t}$ experiences abrupt structural breaks and/or threshold effects of size $1/T^{\alpha}$. To see that Condition \ref{con:holder} implies Assumption \ref{asm:uniform_negligible}, bound the quantity that appears in Assumption \ref{asm:uniform_negligible} as follows:
\begin{align}
    &\max_{-\lfloor Th \rfloor\leq j \leq \lfloor Th \rfloor}\left\|\frac{1}{Th}\sum_{i=1}^TK\left(\frac{t_j-i}{Th}\right)x_ix_i'\left(\beta_{T,i}-\beta_{T,t_j}\right)\right\| \\
    &\leq C\max_{-\lfloor Th \rfloor\leq j \leq \lfloor Th \rfloor}\max_{t_j-\lfloor Th \rfloor \leq i \leq t_j + \lfloor Th \rfloor}\left\|\beta_{T,i}-\beta_{T,t_j}\right\|\max_{-\lfloor Th \rfloor\leq j \leq \lfloor Th \rfloor}\frac{1}{Th}\sum_{i=t_j-\lfloor Th \rfloor}^{t_j+\lfloor Th \rfloor}\left\|x_i\right\|^2 \\
    &\leq Ch^{\alpha}\frac{1}{Th}\sum_{i=t-2\lfloor Th \rfloor}^{t+2\lfloor Th \rfloor}\left\|x_i\right\|^2 = o_p(1),
\end{align}
where $t_j=t+j$, the first inequality holds because $K(\cdot)$ is bounded on compact support $[-1,1]$ under Assumption \ref{asm:kernel}, the second inequality follows from Condition \ref{con:holder}, and the last equality follows from Assumption \ref{asm:dgp}(a) and the assumed condition that $h\to0$ and $\alpha>0$.

\section{Random walk condition}

When $\beta_{T,t}$ follows the rescaled random walk as in Example \ref{exm:rw}, Assumption \ref{asm:uniform_negligible} holds under a set of conditions that are similar to Assumptions \ref{asm:dgp} and \ref{asm:min_eigen}. The following condition, which is attributed to \citetapp{giraitisTimevaryingInstrumentalVariable2021}, is sufficient for Assumption 5.
\renewcommand*{\thecon}{RW}
\begin{con}

    \begin{itemize}
        \item[(a)] $\{(x_t', \varepsilon_t)\}_t$ is $\alpha$-mixing (but not necessarily stationary) with mixing coefficients $b_k$ such that for some $c>0$ and $0<\phi<1$,
        \begin{align}
            b_k\leq c\phi^k, \quad k\geq 1.
        \end{align} Moreover, $\sup_t E[\|x_t\|^{r}] + \sup_t E[|\varepsilon_t|^{r}] < \infty$ for some $r>8$.
	
        \item[(b)] $\{x_t\varepsilon_t\}_t$ has mean zero.

        \item[(c)] There exists some constant $\rho>0$ such that for all $t\in\mathbb{N}$, $\lambda'E[x_tx_t']\lambda\geq \rho\left\|\lambda\right\|^2$ for any $\lambda\neq0$. Furthermore, $\inf_{t\geq1}E[\varepsilon_t^2]>0$.

        \item[(d)] For any element $\beta_{T,t}^{(\ell)}$ in $\beta_{T,t}$, $\ell=1,\ldots,p$, it holds that
        \begin{align}
            \left|\beta_{T,t}^{(\ell)}-\beta_{T,s}^{(\ell)}\right| \leq \left(\frac{|t-s|}{T}\right)^{1/2}r_{ts}^{(\ell)}
        \end{align}
        for some random variable $r_{ts}^{(\ell)}$, and the distribution of $X=\beta_{T,t}^{(\ell)}, \hspace{0.1cm} r_{ts}^{(\ell)}$ has a thin tail:
        \begin{align}
            P\left(|X|\geq \omega\right) \leq\exp(-c_0|\omega|^{a}), \ \omega>0
        \end{align}
        for some $c_0>0$ and $a>0$ that do not depend on $\ell,t,s$ and $T$.
    \end{itemize}
    \label{con:rw_assumption5}
\end{con}
Part (a) of Condition \ref{con:rw_assumption5} strengthens Assumption \ref{asm:dgp}(a) in two ways. First, we require the variables to be $\alpha$-mixing with mixing coefficients decaying exponentially fast. Second, $x_t$ and $\varepsilon_t$ have an $r$-th moment ($r>8$) that is finite uniformly in $t$. Part (b) is weaker than Assumption \ref{asm:dgp}(b) in that $x_t\varepsilon_t$ may be serially correlated. Part (c) strengthens Assumption \ref{asm:min_eigen} by bounding the variance of $\varepsilon_t$ away from zero uniformly in $t$. Part (d) is satisfied if $\beta_{T,t}=T^{-1/2}\sum_{i=1}^tu_i$ with $u_i$ being weakly serially dependent and having a thin tail. For example, part (d) holds if $u_t$ is i.i.d. normal, or stationary mixing and has a thin tail distribution, as discussed in \citetapp{giraitisTimevaryingInstrumentalVariable2021}. They show that, under Condition \ref{con:rw_assumption5},
\begin{align}
    \max_{-\lfloor Th \rfloor\leq j \leq \lfloor Th \rfloor}\left\|\frac{1}{Th}\sum_{i=1}^TK\left(\frac{t_j-i}{Th}\right)x_ix_i'\left(\beta_{T,i}-\beta_{T,t_j}\right)\right\| = O_p(h^{1/2}\log^{1/a}T),
\end{align}
which is $o_p(1)$ if $h=cT^{\gamma}$ for some $\gamma<0$ and $c>0$. In particular, the choice of $\gamma=-1/2$ ensures that Assumption \ref{asm:uniform_negligible} holds, and hence this assumption is compatible with the optimal bandwidth $h=cT^{-1/2}$ in the case of rescaled random walk coefficients.

\titleformat*{\section}{\Large\bfseries\centering}
\section*{Appendix D: Performance of Structural Break Tests}

\setcounter{equation}{0}
\renewcommand{\theequation}{D.\arabic{equation}}
\setcounter{section}{0}
\renewcommand{\thesection}{D.\arabic{section}}
\setcounter{figure}{0}
\renewcommand{\thefigure}{D.\arabic{figure}}
\setcounter{table}{0}
\renewcommand{\thetable}{D.\arabic{table}}

In this appendix, we investigate the behavior of structural break tests. Our focus is on whether the tests for structural breaks can correctly discover latent discontinuous breaks even when they are mixed with smooth parameter instabilities. We verify this via (limited) Monte Carlo experiments. The data is generated as $y_t = \beta_{T,t}x_t + \varepsilon_{t}, \ t=1,\ldots,T$, where $\varepsilon_t \sim \mathrm{i.i.d.} \ N(0,1)$, and $x_t = 0.5x_{t-1} + \varepsilon_{x,t}$ with $\varepsilon_{x,t} \sim \mathrm{i.i.d.} \ N(0,1)$. $\beta_{T,t}$ is defined as a smooth function or rescaled random walk with two abrupt breaks. Specifically, we let $\beta_{T,t} = \mu_{T,t} + h_{T,t}$, where $\mu_{T,t} = \sum_{i=1}^{3}T^{-\alpha}\mu_i1\{\lfloor \tau_{i-1}T\rfloor + 1 \leq t \leq \lfloor \tau_iT \rfloor\}$ with $\tau_{0} = 0$, $\tau_1 = 0.3$, $\tau_2 = 0.7$, and $\tau_{3} = 1$. $h_{T,t}$ is specified as either a deterministic smooth function $f(t/T)$ or rescaled random walk $g_{T,t}$. The function $f$ is equal to $f(u) = 2u + \exp(-16(u-0.5)^2)$ or $f(u) = \{\sin(\pi u) + \cos(2\pi u) + \sin(3\pi u) + \cos(4\pi u)\}/4$. $g_{T,t}$ is generated as $g_{T,t} = T^{-1/2}\sum_{i=1}^{t}v_i$, where $v_i \sim \mathrm{i.i.d.} \ N(0,1)$ or $v_i \sim \mathrm{i.i.d.} \ \mathrm{log \ normal}$ with parameters $\mu=0$ and $ \sigma = 1$.

When $h_{T,t} = f(t/T)$, $\beta_{T,t}$ evolves smoothly and deterministically over time but experiences two abrupt breaks at the 30\% and 70\% points of the sample period. The magnitude of the breaks is determined by $\mu_i$ and $\alpha$. We let $\mu_1 = 0$, $\mu_2 = 4$, $\mu_3=-2$, and $\alpha \in \{0.1,0.2\}$. When $h_{T,t} = g_{T,t}$, $\beta_{T,t}$ follows a rescaled random walk with two discontinuous jumps.

To identify abrupt breaks, we rely on the comprehensive estimation procedure developed by \citetapp{nguyenMbreaksEstimationInference2023}. In this procedure, the number of breaks and break dates are estimated by the sequential method (SEQ) proposed by \citetapp{baiEstimatingTestingLinear1998}, the BIC suggested by \citetapp{yaoEstimatingNumberChangepoints1988}, the modified SIC (LWZ) of \citetapp{liuSegmentedMultivariateRegression1997} or the modified BIC (KT) of \citetapp{kurozumiModelSelectionCriteria2011} (see \citetapp{nguyenMbreaksEstimationInference2023} for the detailed description of the procedure and the associated R package). We investigate the performance of these four methods through 2000 replications with the sample size being (i) $T=100$, (ii) $T=200$, (iii) $T=400$ and (iv) $T=800$.

We calculate the frequency of particular numbers of breaks (up to 5) being selected and the estimated break date fraction ($\hat{T}_B/T$) being in the $1/25$-neighborhood of the true one.\footnote{We check the behavior of the estimate for the break date fraction, $T_B/T$, rather than break date $T_B$ itself. This is because $T_B/T$ can be consistently estimated but $T_B$ cannot; see \citet{casiniStructuralBreaksTime2018}.}

\begin{table} \caption{Results of structural break tests for $h_{T,t} = f(t/T)$ with $f(u) = 2u + \exp(-16(u-0.5)^2)$}
	\centering
	\begin{threeparttable}
		\begin{tabular}{cc@{\hskip 22pt}c@{\hskip 22pt}c@{\hskip 22pt}c@{\hskip 22pt}c@{\hskip 22pt}c@{\hskip 22pt}c@{\hskip 20pt}c@{\hskip 25pt}c}
			\hline &  & \multicolumn{5}{c}{\# of estimated breaks} & &  \multicolumn{2}{c}{Frequency of $\hat{T}_B/T \in $}  \\
			\hline  &  & 0 & 1 & 2 & 3 & 4 & 5 & $[0.3\pm 1/25]$ & $[0.7\pm 1/25]$ \\
            \hline \multicolumn{10}{c}{$\alpha=0.1$} \\
             \hline \multirow{4}{*}{(i)}& \multicolumn{1}{l@{\hskip 20pt}}{SEQ} &0.126 &	0 &	0.613 &	0.239 &	0.022 &	0 &	0.831 	&0.868  \\ 
             & \multicolumn{1}{l}{BIC} &0 &	0 &	0.870 &	0.126& 	0.005& 	0 &	0.966 &	0.995 \\
             & \multicolumn{1}{l}{LWZ} &0 &	0 	&0.996 &	0.004 &	0 &	0 &	0.979 &	0.995   \\
             & \multicolumn{1}{l}{KT} &0 &	0 &	0.802 &	0.184 &	0.015 &	0 &	0.964 &	0.994 \\ 
             \hline \multirow{4}{*}{(ii)}& \multicolumn{1}{l}{SEQ} &0.004 &	0 &	0.618 &	0.351 	&0.027 &	0 &	0.977 &	0.996  \\ 
             & \multicolumn{1}{l}{BIC} &0 &	0 &	0.789 &	0.198& 	0.013 &	0 &	0.996 &	0.999  \\
             & \multicolumn{1}{l}{LWZ} &0 &	0 	&0.995 &	0.006 &	0 &	0 &	0.998 &	0.999  \\
             & \multicolumn{1}{l}{KT} &0 &	0 &	0.756 &	0.225 &	0.020 &	0 &	0.996 &	0.999 \\
             \hline  \multirow{4}{*}{(iii)}& \multicolumn{1}{l}{SEQ} &0 &	0 &	0.358 &	0.554 &	0.089 &	0 	&0.994& 	1  \\ 
             & \multicolumn{1}{l}{BIC} &0 &	0 &	0.518& 	0.411 &	0.072& 	0 &	1 &	1  \\
             & \multicolumn{1}{l}{LWZ} &0 &	0 &	0.995 &	0.006& 	0 &	0& 	1 &	1  \\
             & \multicolumn{1}{l}{KT} &0 &	0 &	0.515 &	0.430 &	0.056 &	0 &	1 &	1 \\
             \hline  \multirow{4}{*}{(iv)}& \multicolumn{1}{l}{SEQ} &0 &	0 &	0.069 &	0.620 &	0.311 &	0.001 &	0.999 &	1  \\ 
             & \multicolumn{1}{l}{BIC} &0 &	0 &	0.105 &	0.453 &	0.443 &	0 &	1 	&1  \\
             & \multicolumn{1}{l}{LWZ} &0 &	0 	&0.966 &	0.034 &	0 &	0 &	1& 	1  \\
             & \multicolumn{1}{l}{KT} &0 &	0 	&0.127 &	0.502 &	0.371 &	0 &	1 &	1 \\
            \hline \multicolumn{10}{c}{$\alpha=0.2$} \\
             \hline \multirow{4}{*}{(i)}& \multicolumn{1}{l@{\hskip 20pt}}{SEQ} &0.044 &	0 	&0.657& 	0.280 &	0.020 &	0 &	0.841 &	0.935  \\ 
             & \multicolumn{1}{l}{BIC} &0 &	0 &	0.847 &	0.149 &	0.005 &	0 	&0.883 	&0.981  \\
             & \multicolumn{1}{l}{LWZ} &0 &	0 &	0.993 &	0.007 &	0 &	0 &	0.920 &	0.981  \\
             & \multicolumn{1}{l}{KT} &0 &	0 &	0.775 &	0.211 &	0.015 &	0 &	0.876 &	0.980  \\ 
             \hline \multirow{4}{*}{(ii)}& \multicolumn{1}{l}{SEQ} &0 &	0 &	0.588 &	0.392 &	0.021 &	0.001& 	0.923 &	0.995  \\ 
             & \multicolumn{1}{l}{BIC} &0 &	0 &	0.750 &	0.241 &	0.010 &	0 &	0.938 &	0.995  \\
             & \multicolumn{1}{l}{LWZ} &0 &	0 &	0.992 &	0.008 &	0 &	0 &	0.983 &	0.995  \\
             & \multicolumn{1}{l}{KT} &0 &	0 &	0.722 &	0.264 &	0.015 &	0 &	0.936 &	0.995 \\
             \hline  \multirow{4}{*}{(iii)}& \multicolumn{1}{l}{SEQ} &0 &	0 &	0.343 &	0.590 	&0.068 &	0 &	0.968 &	1  \\ 
             & \multicolumn{1}{l}{BIC} &0 &	0 &	0.463 	&0.481 &	0.056 &	0 &	0.949 &	1  \\
             & \multicolumn{1}{l}{LWZ} &0 &	0 &	0.987 &	0.014 &	0 &	0 &	0.996 &	1  \\
             & \multicolumn{1}{l}{KT} &0 &	0 &	0.473 &	0.481 &	0.047 &	0 &	0.950 &	1 \\
             \hline  \multirow{4}{*}{(iv)}& \multicolumn{1}{l}{SEQ} &0 &	0 &	0.067 &	0.682 &	0.251 &	0 	&0.992 &	1  \\ 
             & \multicolumn{1}{l}{BIC} &0 &	0 &	0.091 &	0.556 &	0.354 &	0 &	0.980 &	1  \\
             & \multicolumn{1}{l}{LWZ} &0 &	0 &	0.938& 	0.063& 	0 &	0 &	0.992 	&1  \\
             & \multicolumn{1}{l}{KT} &0 &	0 &	0.119 &	0.583 &	0.298 &	0 &	0.977 &	1 \\
             \hline 
		\end{tabular}
		\begin{tablenotes}
			\footnotesize
			\item[]Note: $\beta_{T,t} = \mu_{T,t} + f(t/T)$, where $\mu_{T,t} = \sum_{i=1}^{3}\mu_iT^{-\alpha}1\{\lfloor \tau_{i-1}T \rfloor + 1 \leq t \leq \lfloor \tau_iT \rfloor\}$ and $f(u) = 2u+\exp(-16(u-0.5)^2)$. The number of replications is 2000. The sample size is $T=100$ for case (i), $T=200$ for case (ii) $T=400$ for case (iii), and $T=800$ for case (iv).
		\end{tablenotes}
	\end{threeparttable} \label{tab:sb_mixed_sf_iid}
\end{table}

Let us start with the case of $h_{T,t} = f(t/T)$ with $f(u) = 2u + \exp(-16(u-0.5)^2)$ (Table \ref{tab:sb_mixed_sf_iid}). When $\alpha=0.1$ and $T=100$ (case (i)), the SEQ method estimates no break with a probability of 13\%, while it overestimates the number of breaks in 26\% of the 2000 replications. The estimate of the break date fraction falls in the $1/25$-neighborhood of the true brake date fraction with a probability of 80\%-85\%. As $T$ gets larger, the frequency of underestimating the number of breaks decreases, and the true break points are detected more frequently, but the number of breaks is more likely to be overestimated. In particular, the estimated number of breaks is more than two in 93\% of the 2000 replications when $T=800$. The same tendency to overestimate the number of breaks is shared by the BIC and KT methods, although they can identify the true break points with a high probability even when $T=100$. This implies that BIC and KT often detect spurious breaks in addition to the true ones. LWZ is the most successful in this case, identifying the true breaks in almost all replications for $T\geq200$ without detecting an additional spurious break. However, LWZ is more likely to overestimate the number of breaks as $T$ grows. When $\alpha=0.2$, the tendency to overestimate the number of breaks is greater for all the four tests than in the case with $\alpha=0.1$, and the probability of the true breaks being identified decreases. The LWZ method still performs well, estimating the number of breaks to be two with a probability of not less than 93\%.

\begin{table} \caption{Results of structural break tests for $h_{T,t} = f(t/T)$ with $f(u) = \{\sin(\pi u) + \cos(2\pi u) + \sin(3\pi u) + \cos(4\pi u)\}/4$}
	\centering
	\begin{threeparttable}
		\begin{tabular}{cc@{\hskip 22pt}c@{\hskip 22pt}c@{\hskip 22pt}c@{\hskip 22pt}c@{\hskip 22pt}c@{\hskip 22pt}c@{\hskip 20pt}c@{\hskip 25pt}c}
			\hline &  & \multicolumn{5}{c}{\# of estimated breaks} & &  \multicolumn{2}{c}{Frequency of $\hat{T}_B/T \in $}  \\
			\hline  &  & 0 & 1 & 2 & 3 & 4 & 5 & $[0.3\pm 1/25]$ & $[0.7\pm 1/25]$ \\
            \hline \multicolumn{10}{c}{$\alpha=0.1$} \\
             \hline \multirow{4}{*}{(i)}& \multicolumn{1}{l@{\hskip 20pt}}{SEQ} &0.001 &	0 &	0.829 &	0.162 &	0.009 &	0 &	0.962 &	0.995  \\ 
             & \multicolumn{1}{l}{BIC} &0 &	0 &	0.972 &	0.028 &	0.001 &	0 &	0.962 &	0.995 \\
             & \multicolumn{1}{l}{LWZ} &0 &	0.002 &	0.999 &	0 &	0 &	0 &	0.961 &	0.995  \\
             & \multicolumn{1}{l}{KT} &0 &	0 &	0.924 &	0.073 &	0.004 &	0 &	0.962 &	0.995 \\ 
             \hline \multirow{4}{*}{(ii)}& \multicolumn{1}{l}{SEQ} &0 &	0 	&0.845 &	0.151 &	0.005 &	0 &	0.997 	&0.999  \\ 
             & \multicolumn{1}{l}{BIC} &0 &	0 	&0.967& 	0.033 &	0 &	0 &	0.997 &	1  \\
             & \multicolumn{1}{l}{LWZ} &0 &	0 &	1& 	0& 	0 &	0 &	0.997 &	1  \\
             & \multicolumn{1}{l}{KT} &0 &	0 &	0.947 	&0.054& 	0 	&0 &	0.997 &	1 \\
             \hline  \multirow{4}{*}{(iii)}& \multicolumn{1}{l}{SEQ} &0 &	0 &	0.773 &	0.219 &	0.009 &	0 &	1 	&1  \\ 
             & \multicolumn{1}{l}{BIC} &0 &	0 &	0.929 	&0.070& 	0.001 &	0 &	1 &	1  \\
             & \multicolumn{1}{l}{LWZ} &0 &	0 &	1 &	0 &	0 &	0 &	1 &	1  \\
             & \multicolumn{1}{l}{KT} &0 &	0 &	0.921 &	0.077 &	0.003 &	0 &	1 &	1 \\
             \hline  \multirow{4}{*}{(iv)}& \multicolumn{1}{l}{SEQ} &0 &	0 &	0.522 &	0.425 	&0.053 &	0.001 &	1 &	1  \\ 
             & \multicolumn{1}{l}{BIC} &0 &	0 	&0.789 	&0.201 &	0.011 &	0 &	1 &	1  \\
             & \multicolumn{1}{l}{LWZ} &0 &	0 &	1 &	0 &	0 &	0 &	1 &	1  \\
             & \multicolumn{1}{l}{KT} &0 &	0 	&0.800 &	0.190 &	0.011 &	0 &	1 &	1 \\
            \hline \multicolumn{10}{c}{$\alpha=0.2$} \\
             \hline \multirow{4}{*}{(i)}& \multicolumn{1}{l@{\hskip 20pt}}{SEQ} &0.025 &	0.047 &	0.782 &	0.139 &	0.009 &	0 &	0.765 	&0.950 \\ 
             & \multicolumn{1}{l}{BIC} &0 &	0.044 	&0.924 &	0.032 &	0.001 &	0 &	0.798 	&0.970  \\
             & \multicolumn{1}{l}{LWZ} &0.078 &	0.303 &	0.620 &	0 &	0 &	0 &	0.529 &	0.898  \\
             & \multicolumn{1}{l}{KT} &0.001 &	0.047 &	0.877 &	0.073 &	0.004 &	0 &	0.793 &	0.971  \\ 
             \hline \multirow{4}{*}{(ii)}& \multicolumn{1}{l}{SEQ} &0 &	0.006 &	0.845 	&0.145 	&0.005 &	0 &	0.912 &	0.994  \\ 
             & \multicolumn{1}{l}{BIC} &0 &	0.007 &	0.957 &	0.036 &	0 &	0 &	0.911 &	0.985  \\
             & \multicolumn{1}{l}{LWZ} &0.009 &	0.248& 	0.744 &	0 &	0 	&0 &	0.694 	&0.980  \\
             & \multicolumn{1}{l}{KT} &0 &	0.011 &	0.928 &	0.061& 	0.001 &	0 &	0.905 &	0.985 \\
             \hline  \multirow{4}{*}{(iii)}& \multicolumn{1}{l}{SEQ} &0& 	0.001 &	0.768 &	0.224 	&0.008 &	0 &	0.966 	&1  \\ 
             & \multicolumn{1}{l}{BIC} &0 &	0.001 &	0.923 &	0.076& 	0.001& 	0 &	0.968& 	0.997  \\
             & \multicolumn{1}{l}{LWZ} &0 &	0.134 &	0.866 &	0& 	0 &	0 &	0.842 &	0.997  \\
             & \multicolumn{1}{l}{KT} &0 &	0.001 &	0.917 &	0.079 &	0.004 &	0 &	0.967 &	0.997 \\
             \hline  \multirow{4}{*}{(iv)}& \multicolumn{1}{l}{SEQ} &0& 	0 &	0.528& 	0.413 &	0.059 &	0.001& 	0.979 &	1  \\ 
             & \multicolumn{1}{l}{BIC} &0 &	0 &	0.767 &	0.220 &	0.013 &	0 	&0.979& 	1  \\
             & \multicolumn{1}{l}{LWZ} &0 &	0.074 &	0.926 &	0 &	0 &	0 &	0.907 &	1  \\
             & \multicolumn{1}{l}{KT} &0 &	0.001 &	0.782 &	0.206 &	0.012 &	0 	&0.979 &	1 \\
             \hline 
		\end{tabular}
		\begin{tablenotes}
			\footnotesize
			\item[]Note: $\beta_{T,t} = \mu_{T,t} + f(t/T)$, where $\mu_{T,t} = \sum_{i=1}^{3}\mu_iT^{-\alpha}1\{\lfloor \tau_{i-1}T \rfloor + 1 \leq t \leq \lfloor \tau_iT \rfloor\}$ and $f(u) = \{\sin(\pi u) + \cos(2\pi u) + \sin(3\pi u) + \cos(4\pi u)\}/4$. The number of replications is 2000. The sample size is $T=100$ for case (i), $T=200$ for case (ii) $T=400$ for case (iii), and $T=800$ for case (iv).
		\end{tablenotes}
	\end{threeparttable} \label{tab:sb_mixed_sf_iid_2}
\end{table}

Next, we consider the case of $h_{T,t} = f(t/T)$ with $f(u) = \{\sin(\pi u) + \cos(2\pi u) + \sin(3\pi u) + \cos(4\pi u)\}/4$ (Table \ref{tab:sb_mixed_sf_iid_2}). When $\alpha=0.1$, the behaviors of the four methods are similar to those in the case of $f(u) = 2u+\exp(-16(u-0.5)^2)$ with $\alpha=0.1$, but the SEQ procedure estimates the number of breaks to be not less than two in almost all replications and detects the true breaks with a high probability even when $T=100$. In this case, the LWZ procedure is the most successful one, identifying the true breaks without detecting a spurious one in all replications for $T=400,800$. When $\alpha=0.2$, there are several differences. First, LWZ is more likely to underestimate the number of breaks than the other tests. For example, it estimates the number of breaks to be less than two with probabilities of 38\% and 26\% for $T=100$ and $T=200$, respectively. The probability of the underestimation is still nonnegligible even when $T=400,800$, under which sample size the other tests estimate the number of breaks to be not less than two in almost all replications. This causes the true breaks (in particular, the first one) to be overlooked by LWZ. For the other tests (SEQ, BIC, KT), the tendency to overestimate the number of breaks becomes stronger as $T$ increases. These tests are the most successful procedures in terms of identifying the true breaks.

\begin{table} \caption{Results of structural break tests for $h_{T,t} = g_{T,t} = T^{-1/2}\sum_{i=1}^{t}v_i$ with $v_i \sim \mathrm{i.i.d.} \ N(0,1)$}
	\centering
	\begin{threeparttable}
	   	\begin{tabular}{cc@{\hskip 22pt}c@{\hskip 22pt}c@{\hskip 22pt}c@{\hskip 22pt}c@{\hskip 22pt}c@{\hskip 22pt}c@{\hskip 20pt}c@{\hskip 25pt}c}
			\hline &  & \multicolumn{5}{c}{\# of estimated breaks} & &  \multicolumn{2}{c}{Frequency of $\hat{T}_B/T \in $}  \\
			\hline  &  & 0 & 1 & 2 & 3 & 4 & 5 & $[0.3\pm 1/25]$ & $[0.7\pm 1/25]$ \\
            \hline \multicolumn{10}{c}{$\alpha=0.1$} \\
             \hline \multirow{4}{*}{(i)}& \multicolumn{1}{l@{\hskip 20pt}}{SEQ} &0.025 &	0.002 	&0.737 &	0.223 &	0.014 &	0 	&0.932 	&0.967  \\ 
             & \multicolumn{1}{l}{BIC} &0 &	0.001 &	0.897 &	0.100 &	0.003 &	0 &	0.957 &	0.994 \\
             & \multicolumn{1}{l}{LWZ} &0 &	0.007 	&0.984 &	0.010& 	0 &	0 &	0.953 &	0.994   \\
             & \multicolumn{1}{l}{KT} &0 &	0.001 &	0.838 &	0.152 &	0.010 &	0 &	0.957 &	0.994 \\ 
             \hline \multirow{4}{*}{(ii)}& \multicolumn{1}{l}{SEQ} &0.001 &	0.001 &	0.718 &	0.262 &	0.019 &	0.001 &	0.990 &	0.998  \\ 
             & \multicolumn{1}{l}{BIC} &0 &	0 &	0.812 &	0.181 &	0.008 &	0 	&0.996 &	1  \\
             & \multicolumn{1}{l}{LWZ} &0 &	0.003 &	0.979 &	0.019 &	0 &	0 &	0.991 &	1  \\
             & \multicolumn{1}{l}{KT} &0 &	0 &	0.779 &	0.214 &	0.008 &	0 	&0.996 &	1 \\
             \hline  \multirow{4}{*}{(iii)}& \multicolumn{1}{l}{SEQ} &0 &	0 &	0.581 &	0.369 &	0.050 &	0.001 &	0.999 &	1  \\ 
             & \multicolumn{1}{l}{BIC} &0 &	0 &	0.682 &	0.289 &	0.029 &	0.001 &	0.999 &	1  \\
             & \multicolumn{1}{l}{LWZ} &0 &	0.001 &	0.959 &	0.040 &	0.001& 	0 &	0.998 &	1  \\
             & \multicolumn{1}{l}{KT} &0 &	0 &	0.677 &	0.294 &	0.029 &	0.002 &	0.999 &	1 \\ 
             \hline  \multirow{4}{*}{(iv)}& \multicolumn{1}{l}{SEQ} &0 &	0.001 	&0.378 &	0.501 &	0.116 &	0.006 &	0.999 &	1  \\ 
             & \multicolumn{1}{l}{BIC} &0 &	0 	&0.470 	&0.423 &	0.103 &	0.005 &	1 	&1  \\
             & \multicolumn{1}{l}{LWZ} &0 &	0.001& 	0.881 &	0.116 &	0.003 &	0 &	0.999 &	1  \\
             & \multicolumn{1}{l}{KT} &0 &	0 	&0.487 &	0.416 &	0.093 &	0.005& 	1 &	1 \\
            \hline \multicolumn{10}{c}{$\alpha=0.2$} \\ 
            \hline \multirow{4}{*}{(i)}& \multicolumn{1}{l@{\hskip 20pt}}{SEQ} &0.056 &	0.065 	&0.685 &	0.183 &	0.011& 	0 &	0.730 &	0.896  \\ 
             & \multicolumn{1}{l}{BIC} &0.001 &	0.059 &	0.849 &	0.089 &	0.002 &	0 &	0.792 &	0.969 \\
             & \multicolumn{1}{l}{LWZ} &0.025 &	0.187 &	0.784 	&0.004 &	0 	&0 &	0.680 	&0.938   \\
             & \multicolumn{1}{l}{KT} &0.001& 	0.064 &	0.790& 	0.136 &	0.010 &	0 &	0.789 &	0.967 \\ 
             \hline \multirow{4}{*}{(ii)}& \multicolumn{1}{l}{SEQ} &0.006 &	0.050 &	0.718 &	0.214& 	0.012& 	0.001 &	0.845 &	0.969  \\ 
             & \multicolumn{1}{l}{BIC} &0 &	0.046 &	0.792 &	0.157 &	0.006 &	0 &	0.867 	&0.985  \\
             & \multicolumn{1}{l}{LWZ} &0.010 &	0.170 &	0.812 &	0.009 &	0 &	0 &	0.762 &	0.968  \\
             & \multicolumn{1}{l}{KT} &0 &	0.044 &	0.760 &	0.188 &	0.009 &	0 	&0.865 &	0.984 \\
             \hline  \multirow{4}{*}{(iii)}& \multicolumn{1}{l}{SEQ} &0 &	0.047 &	0.605 &	0.317 &	0.031 &	0 &	0.863 &	0.982  \\ 
             & \multicolumn{1}{l}{BIC} &0 &	0.041 &	0.671 &	0.265 &	0.024 &	0.001 &	0.895 &	0.996  \\
             & \multicolumn{1}{l}{LWZ} &0.003 &	0.169 &	0.804 &	0.025 &	0 &	0 &	0.777 	&0.982  \\
             & \multicolumn{1}{l}{KT} &0 &	0.045 &	0.671 &	0.259 &	0.025 &	0.001 &	0.888 	&0.995 \\
             \hline  \multirow{4}{*}{(iv)}& \multicolumn{1}{l}{SEQ} &0.001 &	0.027 	&0.431 &	0.464 &	0.076 &	0.003 &	0.870 &	0.979  \\ 
             & \multicolumn{1}{l}{BIC} &0 &	0.025 &	0.482 	&0.402 	&0.088 &	0.004 &	0.906 &	0.998  \\
             & \multicolumn{1}{l}{LWZ} &0.005 &	0.163 &	0.761 	&0.069 &	0.003& 	0 &	0.775 	&0.980  \\
             & \multicolumn{1}{l}{KT} &0 &	0.031 &	0.495 &	0.388 &	0.084 &	0.003 &	0.902 	&0.997 \\
             \hline 
		\end{tabular}
		\begin{tablenotes}
			\footnotesize
			\item[]Note: $\beta_{T,t} = \mu_{T,t} + g_{T,t}$, where $\mu_{T,t} = \sum_{i=1}^{3}\mu_iT^{-\alpha}1\{\lfloor \tau_{i-1}T \rfloor + 1 \leq t \leq \lfloor \tau_iT \rfloor\}$ and $g_{T,t} = T^{-1/2}\sum_{i=1}^{t}v_i$ with $v_i \sim \mathrm{i.i.d.} \ N(0,1)$. The sample size is $T=100$ for case (i), $T=200$ for case (ii) $T=400$ for case (iii), and $T=800$ for case (iv).
		\end{tablenotes}
	\end{threeparttable} \label{tab:sb_mixed_rw_iid}
\end{table} 

\begin{table} \caption{Results of structural break tests for $h_{T,t} = g_{T,t} = T^{-1/2}\sum_{i=1}^{t}v_i$ with $v_i \sim \mathrm{i.i.d.} \ \mathrm{log \ normal}$}
	\centering
	\begin{threeparttable}
	   	\begin{tabular}{cc@{\hskip 22pt}c@{\hskip 22pt}c@{\hskip 22pt}c@{\hskip 22pt}c@{\hskip 22pt}c@{\hskip 22pt}c@{\hskip 20pt}c@{\hskip 25pt}c}
			\hline &  & \multicolumn{5}{c}{\# of estimated breaks} & &  \multicolumn{2}{c}{Frequency of $\hat{T}_B/T \in $}  \\
			\hline  &  & 0 & 1 & 2 & 3 & 4 & 5 & $[0.3\pm 1/25]$ & $[0.7\pm 1/25]$ \\
            \hline \multicolumn{10}{c}{$\alpha=0.1$} \\
             \hline \multirow{4}{*}{(i)}& \multicolumn{1}{l@{\hskip 20pt}}{SEQ} &0.026 &	0.003 &	0.739 &	0.216 &	0.017 &	0.001& 	0.930 &	0.963  \\ 
             & \multicolumn{1}{l}{BIC} &0 &	0.002 &	0.899 &	0.095 &	0.005 &	0 &	0.958 &	0.994 \\
             & \multicolumn{1}{l}{LWZ} &0 &	0.008 &	0.986 &	0.007 &	0& 	0 &	0.954 &	0.995   \\
             & \multicolumn{1}{l}{KT} &0 &	0.002 &	0.864 &	0.126 &	0.009 &	0 &	0.956 	&0.994 \\ 
             \hline \multirow{4}{*}{(ii)}& \multicolumn{1}{l}{SEQ} &0 &	0.001 &	0.728 &	0.259 &	0.014 &	0 &	0.992 &	0.997  \\ 
             & \multicolumn{1}{l}{BIC} &0 &	0 &	0.830 &	0.162& 	0.009 &	0 	&0.992 	&0.998  \\
             & \multicolumn{1}{l}{LWZ} &0 &	0.004 &	0.978 &	0.019& 	0 &	0 &	0.991 &	0.998  \\
             & \multicolumn{1}{l}{KT} &0 &	0 &	0.813 &	0.177 &	0.010 &	0 &	0.992 &	0.999 \\
             \hline  \multirow{4}{*}{(iii)}& \multicolumn{1}{l}{SEQ} &0& 	0 &	0.592 &	0.358 &	0.050 &	0.001 &	0.999 &	1  \\ 
             & \multicolumn{1}{l}{BIC} &0 &	0 &	0.703 &	0.266 &	0.031 &	0.001 &	1 &	1  \\
             & \multicolumn{1}{l}{LWZ} &0 &	0.002 &	0.954& 	0.044 &	0.001 &	0 &	0.998 &	1  \\
             & \multicolumn{1}{l}{KT} &0 &	0 	&0.705 &	0.261 &	0.033 &	0.001 &	1 &	1 \\ 
             \hline  \multirow{4}{*}{(iv)}& \multicolumn{1}{l}{SEQ} &0 &	0 &	0.378 &	0.495 &	0.124 &	0.003 &	0.999 &	1  \\ 
             & \multicolumn{1}{l}{BIC} &0 &	0 &	0.475 &	0.420 &	0.100 &	0.006& 	1 &	1 \\
             & \multicolumn{1}{l}{LWZ} &0 &	0.001 &	0.878 &	0.118 &	0.004 &	0 &	0.999 &	1  \\
             & \multicolumn{1}{l}{KT} &0 &	0 	&0.492 &	0.411 &	0.092 &	0.006 &	1 &	1 \\
            \hline \multicolumn{10}{c}{$\alpha=0.2$} \\ 
            \hline \multirow{4}{*}{(i)}& \multicolumn{1}{l@{\hskip 20pt}}{SEQ} &0.057 &	0.078 	&0.680 &	0.175 &	0.011 &	0 &	0.719 &	0.898  \\ 
             & \multicolumn{1}{l}{BIC} &0 &	0.067 &	0.843 &	0.088 &	0.003 &	0 &	0.780 &	0.968 \\
             & \multicolumn{1}{l}{LWZ} &0.026 &	0.198& 	0.772& 	0.004& 	0 &	0 &	0.665 	&0.935   \\
             & \multicolumn{1}{l}{KT} &0 &	0.070 &	0.800 &	0.122 &	0.009 &	0 &	0.775 	&0.966 \\ 
             \hline \multirow{4}{*}{(ii)}& \multicolumn{1}{l}{SEQ} &0.004 &	0.059& 	0.710 &	0.219 &	0.010 &	0& 	0.846 &	0.966  \\ 
             & \multicolumn{1}{l}{BIC} &0 &	0.053 &	0.799 &	0.142& 	0.007 &	0 &	0.856& 	0.985  \\
             & \multicolumn{1}{l}{LWZ} &0.006 &	0.189 &	0.795 &	0.011 &	0 &	0 &	0.742 &	0.970  \\
             & \multicolumn{1}{l}{KT} &0 &	0.057 &	0.784 &	0.152 &	0.008 &	0 &	0.856 &	0.985 \\
             \hline  \multirow{4}{*}{(iii)}& \multicolumn{1}{l}{SEQ} &0.001 &	0.043 	&0.617 &	0.306 &	0.033 &	0.001 &	0.870 &	0.978  \\ 
             & \multicolumn{1}{l}{BIC} &0 	&0.042 &	0.688 &	0.247 &	0.023 &	0.001 &	0.892 &	0.994 \\
             & \multicolumn{1}{l}{LWZ} &0.004 &	0.170 &	0.803 &	0.023 &	0.001 &	0 &	0.776 &	0.978  \\
             & \multicolumn{1}{l}{KT} &0 &	0.045 &	0.688 	&0.240 &	0.027& 	0.001 &	0.889 &	0.993 \\
             \hline  \multirow{4}{*}{(iv)}& \multicolumn{1}{l}{SEQ} &0 &	0.034 &	0.428& 	0.448 &	0.089 &	0.002& 	0.872 &	0.986  \\ 
             & \multicolumn{1}{l}{BIC} &0 &	0.032 &	0.476& 	0.404 &	0.086& 	0.004 &	0.908 &	0.998  \\
             & \multicolumn{1}{l}{LWZ} &0.003 &	0.144 &	0.763 	&0.090 &	0.002 &	0 &	0.791& 	0.981  \\
             & \multicolumn{1}{l}{KT} &0 &	0.037& 	0.487 &	0.395 &	0.078 &	0.004 &	0.903 &	0.997 \\
             \hline 
		\end{tabular}
		\begin{tablenotes}
			\footnotesize
			\item[]Note: $\beta_{T,t} = \mu_{T,t} + g_{T,t}$, where $\mu_{T,t} = \sum_{i=1}^{3}\mu_iT^{-\alpha}1\{\lfloor \tau_{i-1}T \rfloor + 1 \leq t \leq \lfloor \tau_iT \rfloor\}$ and $g_{T,t} = T^{-1/2}\sum_{i=1}^{t}v_i$ with $v_i \sim \mathrm{i.i.d.} \ \mathrm{log \ normal}$ with parameters $\mu=0$ and $ \sigma = 1$. The sample size is $T=100$ for case (i), $T=200$ for case (ii) $T=400$ for case (iii), and $T=800$ for case (iv).
		\end{tablenotes}
	\end{threeparttable} \label{tab:sb_mixed_rw_iid_2}
\end{table}

We turn to the case with $h_{T,t} = g_{T,t} = T^{-1/2}\sum_{i=1}^{t}v_i$ where $v_i \sim \mathrm{i.i.d.} \ N(0,1)$ (Table \ref{tab:sb_mixed_rw_iid}). When $\alpha=0.1$, the behaviors of the four procedures are similar to those in the preceding cases: SEQ, BIC and KT not only identify true breaks but also detect spurious ones, with this tendency being greater for larger $T$, while LWZ identifies true breaks without detecting spurious one with a large probability. However, the probability of overestimating the number of breaks is nonnegligible for LWZ, and this probability gets larger as $T$ increases. When $\alpha=0.2$, SEQ and LWZ are more likely to underestimate the number of breaks than the other two methods. In particular, LWZ underestimates the number of breaks with a nonnegligible probability even when $T=800$ and thus is more likely to overlook the latent breaks than the other tests. SEQ, BIC and KT can identify latent breaks with a high probability but tend to detect spurious breaks. This tendency is stronger for larger $T$, as in the preceding cases. The results for the case with $v_i \sim \mathrm{i.i.d.} \ \mathrm{log \ normal}$ are similar, so the same comment applies.

In general, the tests for structural breaks can identify latent breaks in the presence of another source of parameter instability but tend to detect additional spurious breaks. This tendency is more evident for larger $T$. Investigating the behavior of each test, LWZ identifies latent breaks without estimating spurious breaks in some situations, but it underestimates the number of breaks and overlooks latent breaks in other situations. SEQ is likely to both underestimate and overestimate the number of breaks. BIC and KT can identify true breaks irrespective of the DGP, but they tend to detect additional spurious breaks.

\section*{Appendix E: Additional Results for the Empirical Application}

\setcounter{equation}{0}
\renewcommand{\theequation}{E.\arabic{equation}}
\setcounter{section}{0}
\renewcommand{\thesection}{E.\arabic{section}}
\setcounter{figure}{0}
\renewcommand{\thefigure}{E.\arabic{figure}}
\setcounter{table}{0}
\renewcommand{\thetable}{E.\arabic{table}}

In this appendix, we discuss the estimation results for portfolios G and V (see Section \ref{sec:empirics} for details).

\titleformat*{\section}{\large\bfseries}
\section{Estimates for G}
Table \ref{tab:gamma_selection_cr_g} shows the empirical coverage rates of the bootstrap 95\% confidence intervals. Noting that $\gamma_1=-0.33,-0.4,-0.5$ satisfies the 90\% criterion, $\gamma=\hat{\gamma}=-0.33$ is selected. Figure \ref{fig:CV_gamma_g} also supports this result ($\mathrm{CV}(\gamma)$ is minimized at $\gamma=-0.35$ for $m=2$).

In Figure \ref{fig:g_int_kernel}, we plot the estimated time-varying alpha and its confidence band obtained from the kernel method with $h=\hat{c}T^{-1/3}$. The estimated alpha stays around zero as a whole, but there are troughs at $t=400$ and $t=530$, around which the confidence band excludes the value zero. Figure \ref{fig:g_slope_kernel} shows the estimates for the time-varying beta. It starts with a value of 0.8 and begins to increase soon later. From $t=100$ until the end of the sample, it stays between 1.2 and 1.5. The confidence band does not include the value zero throughout the sample period.

\begin{figure}
    \centering
    \includegraphics[width=0.85\textwidth]{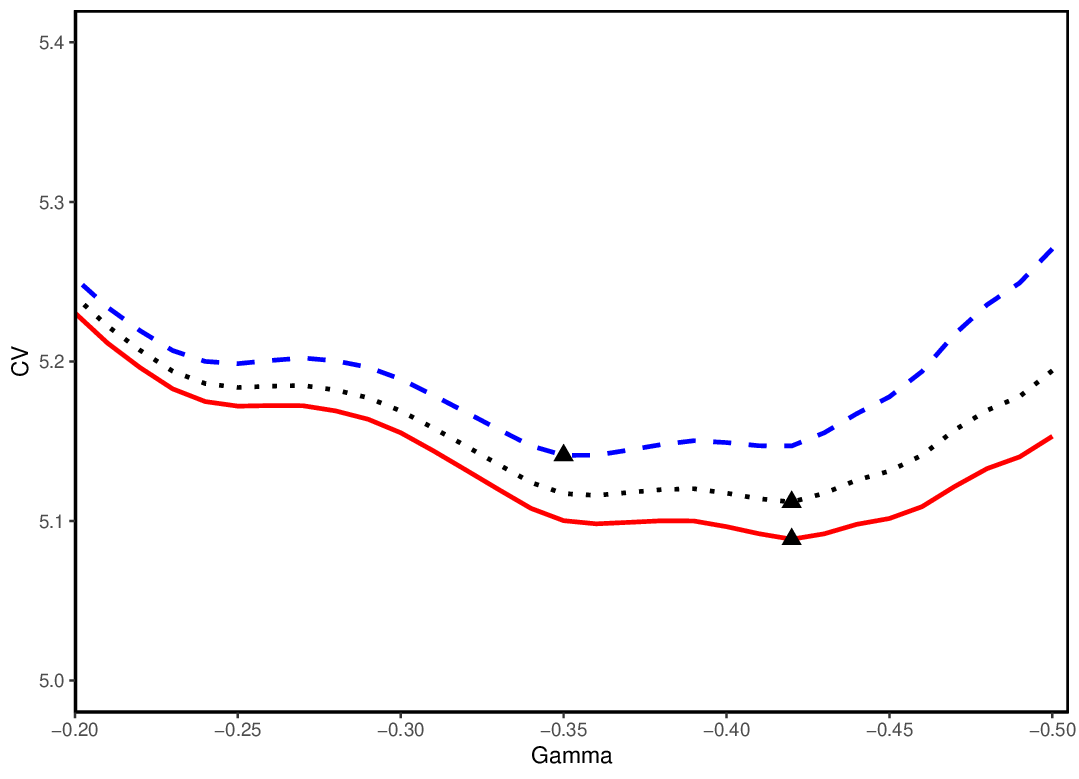}
    \caption{Cross-validation criteria calculated using leave-$(2m+1)$-out estimators with $h=T^{\gamma}$, for G}
    \caption*{\color{red}\full\color{black}: $m=0$, \quad \color{black}\dotted\color{black}: $m=1$, \quad \color{blue}\dashed\color{black}: $m=2$, \quad $\blacktriangle$: Minimum}
    \label{fig:CV_gamma_g}
\end{figure}

\begin{table} \caption{Mean empirical coverage rates of 95\% bootstrap confidence intervals for G}
	\centering
	\begin{threeparttable}
		\renewcommand{\arraystretch}{1.3}	\begin{tabular}{ccc@{\hskip 21pt}c@{\hskip 21pt}c@{\hskip 21pt}c}
			\hline  &  & \multicolumn{4}{c}{$\gamma_2$} \\
             \cline{3-6}& & -0.2 & -0.33 & -0.4 & -0.5 \\
             \hline \multirow{5}{*}{$\gamma_1$}& \multicolumn{1}{l}{-0.2} &0.894	&0.937	&0.936	&0.931 \\ 
             &\multicolumn{1}{l}{-0.33} &-	&0.918&	0.932&	0.930 \\
             &\multicolumn{1}{l}{-0.4} &-	&-	&0.911&	0.927 \\
             & \multicolumn{1}{l}{-0.5} &-	&-	&-	&0.901 \\
             \hline 
		\end{tabular}
		\begin{tablenotes}
			\footnotesize
			\item[]Note: Each entry denotes the mean empirical coverage rate of the 95\% bootstrap confidence intervals for $(\hat{\alpha}_{j,t}(\gamma_1), \hat{\beta}_{j,t}(\gamma_1))$ based on $(\hat{\alpha}_{j,t}^*(\gamma_1,\gamma_2), \hat{\beta}_{j,t}^*(\gamma_1,\gamma_2))$ taken over $t=1,\ldots,T$: $\overline{\mathrm{CR}}(\gamma_1,\gamma_2)= T^{-1}\sum_{t=1}^T\mathrm{CR}_t(\gamma_1,\gamma_2)$.
		\end{tablenotes}
	\end{threeparttable} \label{tab:gamma_selection_cr_g}
\end{table}

\begin{figure}
    \centering
    \begin{subfigure}{0.49\textheight}
        \includegraphics[width=\columnwidth]{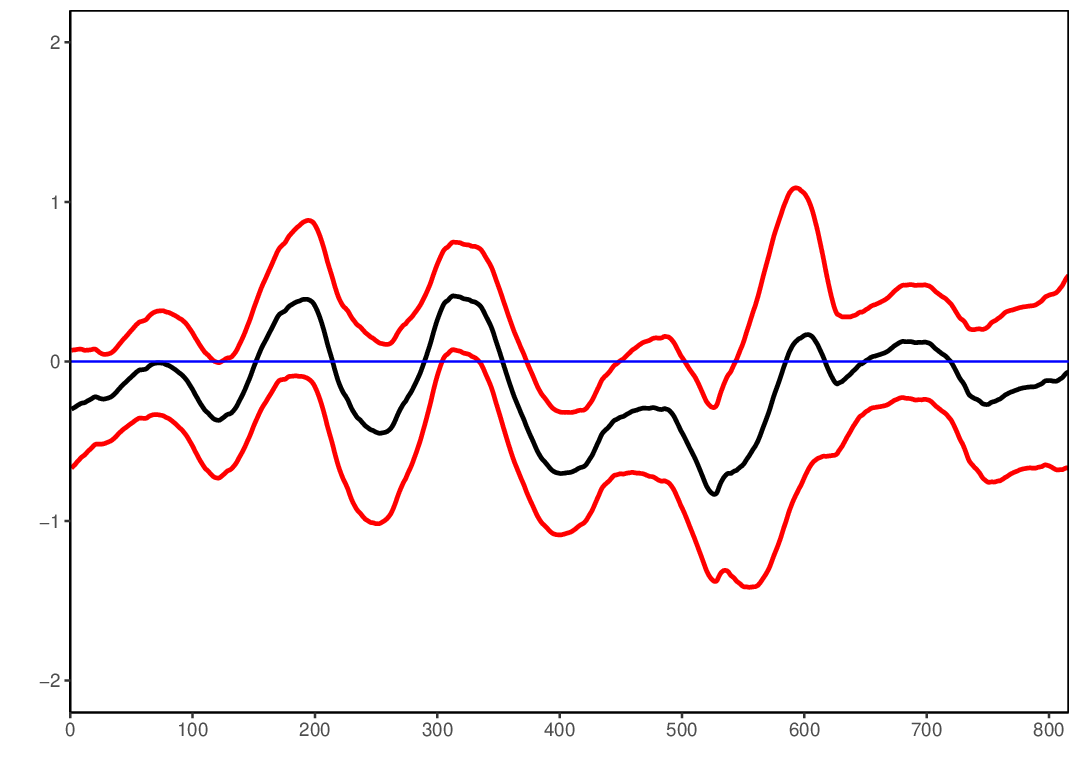}
    \caption{Plot of the time-varying alpha}
	\label{fig:g_int_kernel}
    \end{subfigure}
    \begin{subfigure}{0.49\textheight}
        \includegraphics[width=\columnwidth]{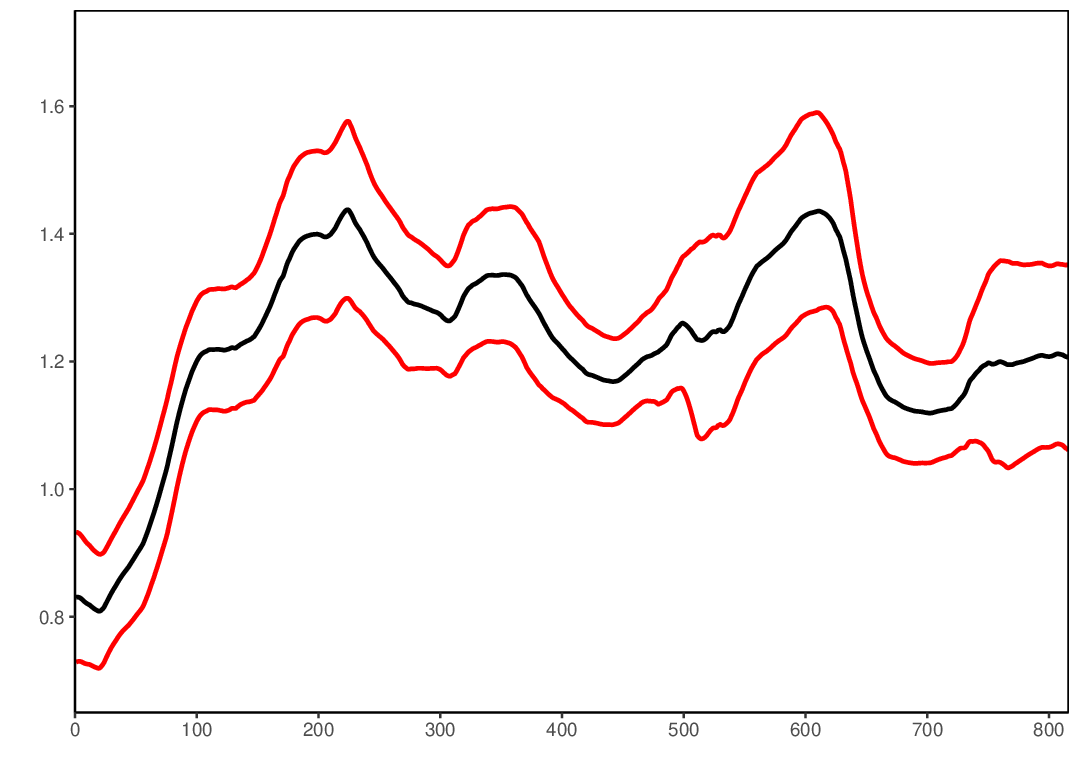}
    \caption{Plot of the time-varying beta}
	\label{fig:g_slope_kernel}
    \end{subfigure}
    \caption{Estimates and 95\% confidence band from the kernel-based method ($h=\hat{c}T^{-1/3}$) for G}
    (The horizontal line in (a) indicates the value zero.)
    \label{fig:g}
\end{figure}

\section{Estimates for V}
According to Table \ref{tab:gamma_selection_cr_v} and Figure \ref{fig:CV_gamma_v}, $\hat{\gamma}=1/3$ is supported ($\mathrm{CV}(\gamma)$ is minimized at $\gamma=-0.3$ for all $m$ considered).

Figure \ref{fig:v_int_kernel} shows the estimated time-varying alpha and its confidence band obtained from the kernel method with $h=\hat{c}T^{-1/3}$. There are several periods when the value zero is excluded from the band and the time-varying alpha exhibits positive effects. Figure \ref{fig:v_slope_kernel} depicts the estimated time-varying beta. From $t=1$ to $t=320$, it stays between 1 and 1.2, and then gradually drops and reaches 0.8 at $t=540$. Then, the time-varying beta starts to increase and, at $t=620$, re-enter the phase where it fluctuates around 1.2.

\begin{figure}
    \centering
    \includegraphics[width=0.85\textwidth]{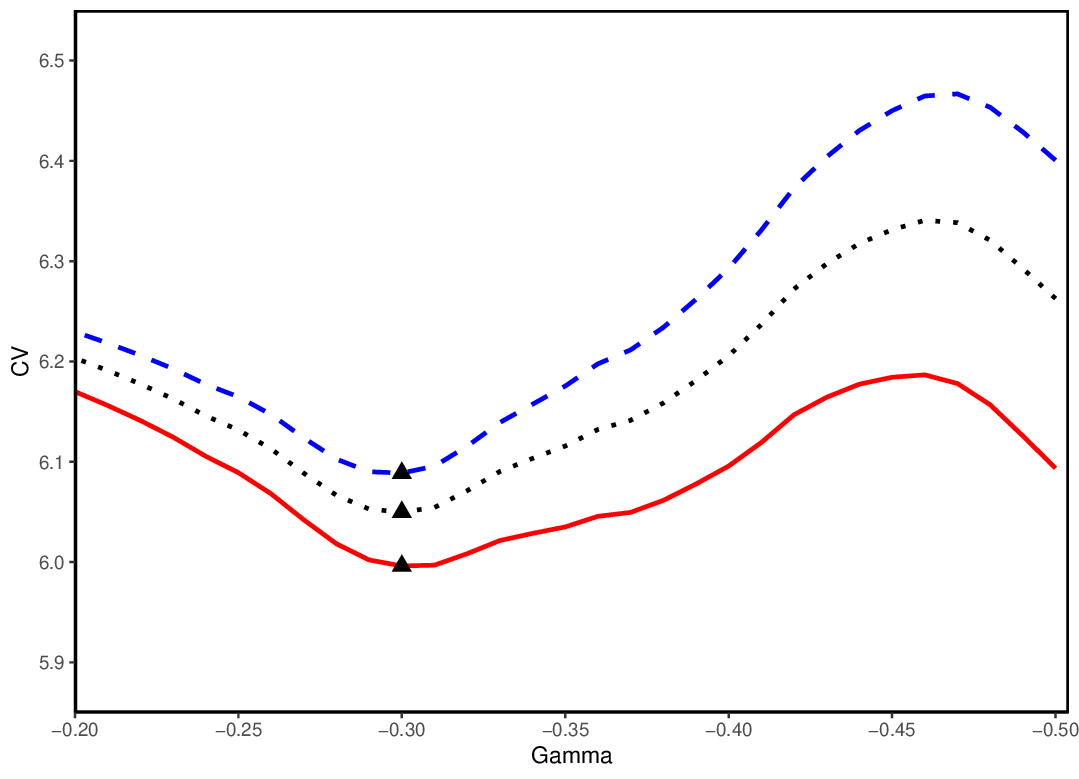}
    \caption{Cross-validation criteria calculated using leave-$(2m+1)$-out estimators with $h=T^{\gamma}$, for V}
    \caption*{\color{red}\full\color{black}: $m=0$, \quad \color{black}\dotted\color{black}: $m=1$, \quad \color{blue}\dashed\color{black}: $m=2$, \quad $\blacktriangle$: Minimum}
    \label{fig:CV_gamma_v}
\end{figure}

\begin{table} \caption{Mean empirical coverage rates of 95\% bootstrap confidence intervals for V}
	\centering
	\begin{threeparttable}
		\renewcommand{\arraystretch}{1.3}	\begin{tabular}{ccc@{\hskip 21pt}c@{\hskip 21pt}c@{\hskip 21pt}c}
			\hline  &  & \multicolumn{4}{c}{$\gamma_2$} \\
             \cline{3-6}& & -0.2 & -0.33 & -0.4 & -0.5 \\
             \hline \multirow{5}{*}{$\gamma_1$}& \multicolumn{1}{l}{-0.2} &0.875	&0.933	&0.934	&0.925 \\ 
             &\multicolumn{1}{l}{-0.33} &-	&0.912	&0.925	&0.917 \\
             &\multicolumn{1}{l}{-0.4} &-	&-	&0.907	&0.914 \\
             & \multicolumn{1}{l}{-0.5} &-	&-	&-	&0.886 \\
             \hline 
		\end{tabular}
		\begin{tablenotes}
			\footnotesize
			\item[]Note: Each entry denotes the mean empirical coverage rate of the 95\% bootstrap confidence intervals for $(\hat{\alpha}_{j,t}(\gamma_1), \hat{\beta}_{j,t}(\gamma_1))$ based on $(\hat{\alpha}_{j,t}^*(\gamma_1,\gamma_2), \hat{\beta}_{j,t}^*(\gamma_1,\gamma_2))$ taken over $t=1,\ldots,T$: $\overline{\mathrm{CR}}(\gamma_1,\gamma_2)= T^{-1}\sum_{t=1}^T\mathrm{CR}_t(\gamma_1,\gamma_2)$.
		\end{tablenotes}
	\end{threeparttable} \label{tab:gamma_selection_cr_v}
\end{table}

\clearpage
\begin{figure}
    \centering
    \begin{subfigure}{0.49\textheight}
        \includegraphics[width=\columnwidth]{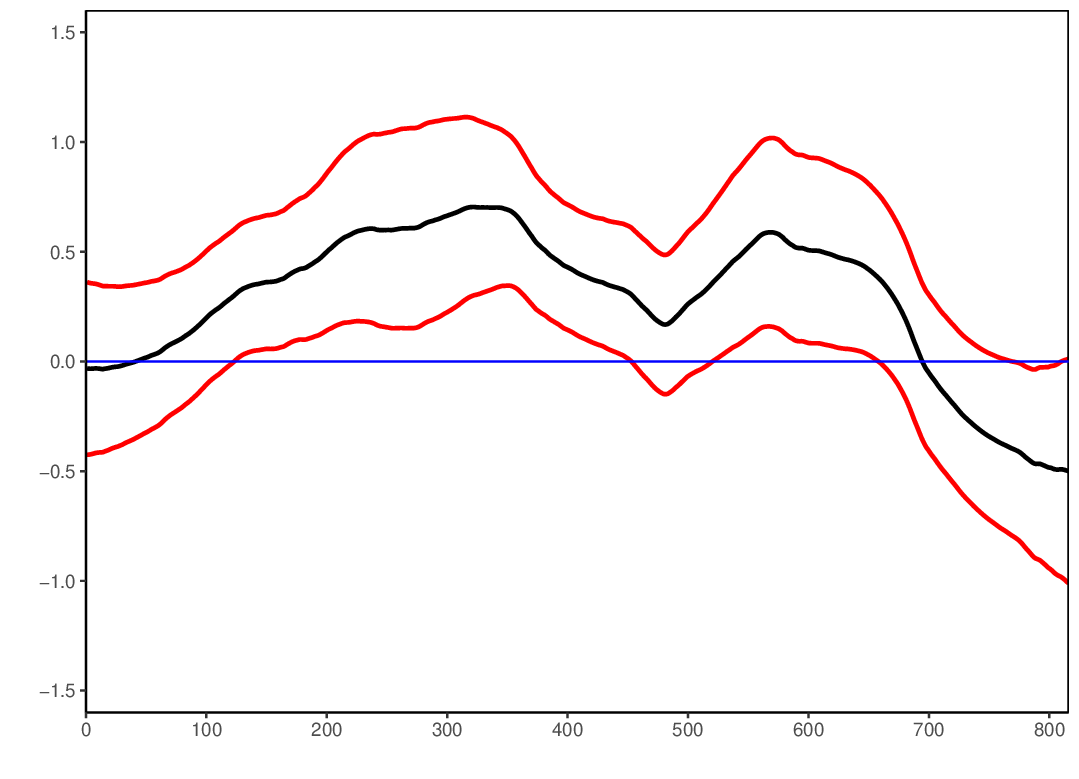}
    \caption{Plot of the time-varying alpha}
	\label{fig:v_int_kernel}
    \end{subfigure}
    \begin{subfigure}{0.49\textheight}
        \includegraphics[width=\columnwidth]{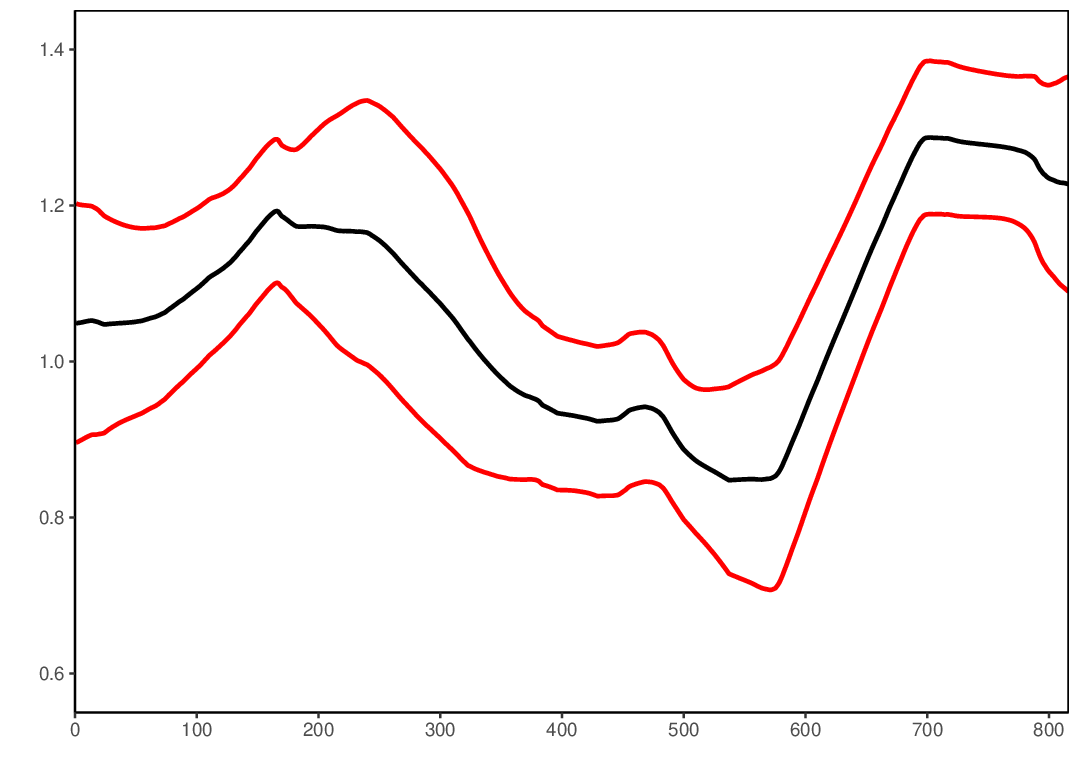}
    \caption{Plot of the time-varying beta}
	\label{fig:v_slope_kernel}
    \end{subfigure}
    \caption{Estimates and 95\% confidence band from the kernel-based method ($h=\hat{c}T^{-1/3}$) for V}
    (The horizontal line in (a) indicates the value zero.)
    \label{fig:v}
\end{figure}

\clearpage

\bibliographystyleapp{standard}
\bibliographyapp{TVP_Linear}

\end{document}